\documentclass[a4paper,11pt]{article}
\usepackage[margin=1in]{geometry} 
\usepackage{setspace}
\usepackage[english]{babel}
\usepackage{afterpage}
\usepackage{amsthm}
\usepackage{amssymb}
\usepackage{amsmath}
\usepackage{graphicx}
\usepackage[utf8]{inputenc}
\usepackage{dsfont} 
\usepackage{textcomp}
\usepackage{listings}
\usepackage{color}
\usepackage{titling}
\usepackage{subcaption}
\usepackage[round]{natbib}
\usepackage[table]{xcolor}
\usepackage{tabularx}
\usepackage{enumerate}
\usepackage{multirow}

\let\epsilon\varepsilon
\newtheorem{defi}{Definition}[section]
\newtheorem{rema}[defi]{Remark}
\newtheorem{lem}[defi]{Lemma}
\newtheorem{theo}[defi]{Theorem}
\newtheorem{ass}[defi]{Assumption}

\newtheorem{exam}[defi]{Example}

\allowdisplaybreaks

\newcommand{\gcal}{\mathcal{G}}

\def\-{\raisebox{.75pt}{-}}

\colorlet{g}{green!70}
\colorlet{p}{red!25}
\colorlet{vp}{red!60}
\colorlet{vg}{green!25}
\colorlet{b}{blue!20}
\colorlet{vb}{blue!50}

\begin{document}


\title{
Multinomial Backtesting of Distortion Risk Measures}

\author{
S\"oren Bettels \hspace{1.3cm} Sojung Kim \hspace{1.3cm} Stefan Weber\\[1.0ex] \textit{Leibniz Universit{\"a}t Hannover} }
\date{\today\thanks{House of Insurance \& Institute of Actuarial and Financial Mathematics, Leibniz Universit\"at Hannover, Welfengarten 1, 30167 Hannover, Germany. e-mail:  {\tt stefan.weber@insurance.uni-hannover.de}}}

\maketitle

\begin{abstract}
We extend the scope of risk measures for which backtesting methods are available by proposing a new approach for general distortion risk measures. The method relies on a stratification and randomization of risk levels. We illustrate the performance of our backtest in numerical case studies. 

\end{abstract}
\vspace{0.2cm}

\doublespacing
	
	\section{Introduction}
In the face of risk and uncertainty financial institutions need to measure and quantify the risk they are exposed to. These measurements can be used to determine the capital that is needed as a buffer against adverse scenarios. Risk measures also help to compare portfolios or balance sheets to each other and to guide management decisions. In this paper, we extend the scope of risk measures for which backtesting models are available.

A variety of risk measures has been suggested in the literature. Value at Risk (V@R) and Average Value at Risk (AV@R) are the basis of different solvency regimes. An axiomatic investigation of monetary risk measures goes back to \cite{Artzner1999}, \cite{FS02}, and \cite{frittelli2002}, see also \cite{FS} and \cite{FW15}.  We introduce a general methodology for backtesting an important class of risk measures that includes the regulatory benchmarks V@R and AV@R as special cases: distortion risk measures (DRMs).

DRMs are an important building block for distribution-based coherent risk measures. This fact is a direct consequence of a representation theorem by \cite{Kusuoka2001}. DRMs also include many risk measures that are not necessarily convex, e.g.~V@R and Range Value at Risk (RV@R). DRMs are important examples of comonotonic risk measures, i.e., risk measures for which risks simply add up for comonotonic positions. 
 
In practice, risk measures are used to assess the risk of future positions and balance sheets. They are applied to probabilistic models that are estimated from past data. Backtesting refers to methodologies that compare observed values of positions to model-based risk assessments. The methods test the adequacy of the risk measurement models of banks and insurance companies in the face of  uncertainty and help to identify  misspecifications that impair risk assessments. Thereby financial firms can validate their forecasting tools for investments and balance sheets positions on both the asset and liability side. 

It is extremely important that risk measures adequately capture those properties that are needed in the risk management process. Risk measures quantify risk. By numerically representing risk, they facilitate the communication within firms, with customers, investors and regulators and provides a solid basis for decisions. Risk measures reduce the complexity of risk by focussing on specific features of random positions.  These issues are systematically studied in the axiomatic theory of risk measures that also provides a variety of examples with different properties, cf. \cite{FS} and \cite{FW15}. In this paper, we complement this literature by expanding the class of risk measures for which powerful backtesting algorithms are available. 

\noindent Our contributions are the following:
\begin{enumerate}[\quad (i)]
\item We propose a multinomial backtesting method for general DRMs which extends the non-randomized AV@R-backtest of \cite{KLM}. The method relies on a stratification and randomization of risk levels. Our stratified mixture approach captures  important characteristics of the DRM by weighting quantiles according to their contribution.
\item We illustrate the performance of our methods in numerical case studies. First, we consider fixed distributions of loss positions under the null hypothesis and under the alternatives and evaluate the size and the power of our test in this simple setting. Second, we apply our method to asset-liability-management.
\item In the special case of AV@R, our backtesting methodology deviates from previously considered multinomial backtests suggested in  \cite{KLM} due to the randomization of risk levels.  A numerical comparison of both methods shows that our approach improves the power of the backtests in all case studies. 
\end{enumerate}

\subsubsection*{Literature}

For reviews on the theory of monetary risk measures, including distortion risk measures, we refer to \cite{FS} and \cite{FW15}.  References to various seminal papers, including Value at Risk (V@R) and Average Value at Risk (AV@R), also called Expected Shortfall, are given in the \emph{Bibliographical notes} to Chapter 4 in \cite{FS}, and in the \emph{Notes and Comments} to Chapter 2 in \cite{MFE2015}. Range Value at Risk was introduced in \cite{CDS2010}. The class of distortion risk measures and the closely related mathematical notion of Choquet integrals are, e.g., discussed in \cite{Choquet}, \cite{Greco}, \cite{Schmeidler}, \cite{W1},  \cite{W2}, \cite{Denneberg}, \cite{AcerbiCarlo}, \cite{Dhaene06},   \cite{Song2006}, \cite{Song2009},  \cite{SongYan2009},   \cite{embrechtsliuwang2016}, \cite{We}, and \cite{KIM2021}. 
Some of our arguments rely on the representation of distortion risk measures as mixtures of V@R as described in \cite{DKLT}. A specific example that we consider is GlueV@R, a risk measure proposed in \cite{BSGS} and \cite{BSGS2}. Further examples of DRMs have continuously received attention, for example proportional hazard transform in \cite{W1,W2}, min/max V@R transforms in \cite{ChernyMadan2008}, and Range V@R  in \cite{BT2016}. We provide a list of such examples in Appendix \ref{ex_DRM} with the respective references, originally compiled by \cite{MS2017}. The robust representation of coherent distribution-based risk measures is due to \cite{Kusuoka2001}.

The literature on backtesting the V@R is extensive. \cite{kupiec1995} describes an algorithm that considers the time and size of the first V@R exceedance. \cite{C} shows that backtests of V@R can be based on the fact that the sequence of V@R exceedances are independent Bernoulli random variables under the null hypothesis. This observation is a starting point for numerous backtesting schemes, see, e.g.,~\cite{christoffersenpelletier2004}, \cite{wong2010}, \cite{berkowitzChristoffersenPelletier2011}, and \cite{ziggel2014}.

Some strategies for backtesting are associated to the notion of elicitability. The coherent risk measure AV@R is, however, not elicitable, see  \cite{weber2006} and \cite{gneiting2011}. All coherent and convex elicitable risk measures are characterized in  \cite{weber2006}, \cite{BB2015}, and \cite{DBBZ2016}. Discussions on the issue of the possibility of backtesting AV@R can be found in \cite{carver2013}, \cite{carver2014}, \cite{chen2014}, \cite{acerbiSzekely2014}, and \cite{fisslerZiegelGneiting2015}. 

Specific backtests for AV@R are developed in the following papers: \cite{EKt} and \cite{KLM} consider discrete V@R level exceedances to implicitly backtest AV@R. \cite{costanzinocurran2015} and \cite{costanzinocurran2017} develop a traffic light system that is based on a backtest of weighted V@R exceedance indicators. \cite{duEscanciano2016} and \cite{loserWiedZiggel2018} test the martingale property of the cumulated violation process. Another approach relies on testing a forecast of the probability density of the P\&L distribution. As suggested in \cite{dieboldGuntherTay1997}, plugging the observed profits and losses into the forecast cumulative distribution function should lead to a  uniform distribution. This strategy is refined in  \cite{berkowitz2001}, \cite{kerkhofMelenbergSchumacher2003} and \cite{gordyMcneil2017} in the context of risk management. 

Our approach modifies and extends multinomial backtests as considered in \cite{KLM} in the context of AV@R where test statistics are adopted from classical multinomial tests. Pearson's $\chi^2$-test was developed\footnote{Seven alternative proofs for the asymptotic behavior of the test statistic under the null hypothesis can be found in \cite{benmel2018}.} in \cite{P}.  A finite sample correction of Pearson's $\chi^2$-test was developed by \cite{N}. We also consider a likelihood ratio test, cf.~Section 10.3 of \cite{CasellaBergerRoger2002}. For a multinomial null hypothesis and multinomial alternatives these tests are compared in \cite{CK}.

	    \subsubsection*{Outline}
	    The paper is structured as follows: Section~\ref{sec:ch2} and Sections~\ref{app:DRM} \&~\ref{ex_DRM} in the appendix review the definition, properties and examples of DRMs. We pay particular attention to the representation of a distortion function as a sum of right- and left-continuous functions and the corresponding decomposition of the DRM. In Section 3 we develop the backtesting methodology for general DRMs. We begin in Section~\ref{sec:original} with a review of a multinomial backtesting scheme for AV@R that was introduced by \cite{KLM}. In the next two sections we describe our stratified and randomized extension: in Section~\ref{sec:right-conti} for left- and right-continuous distortion functions, in Section~\ref{sec:general} for general distortion functions. Section~\ref{sec:DistStudies} illustrates the tractability and performance of our method in numerical experiments. We show in the special case of AV@R that randomization may improve the power of backtests; the numerical experiments are adopted from \cite{KLM}, and the results are compared to their paper.
  We also consider an example of a DRM with distortion functions with jumps, GlueV@R, and illustrate the application of the proposed backtest.  Section~\ref{sec:ALM} applies the method to a more complex asset-liability-model. Section~\ref{sec:concl} concludes and discusses further research. All proofs and auxiliary material are collected in an appendix.

	\section{Distortion Risk Measures}\label{sec:ch2}

Our backtesting methodology focuses on Distortion Risk Measures\footnote{A risk measure $\rho$ is a mapping which assigns to $X
	 \in \mathcal{X}$ a quantitative measurement of the probability and severity of losses:
	    \begin{defi}\label{def:drm}
	        A function $\rho: \mathcal{X} \rightarrow \mathbb{R}$ is called a monetary risk measure if it satisfies
	        \begin{itemize}
	            \item[(i)] \textit{Monotonicity:} \hspace{3mm}
	            If  $X \leq Y,$ $X, Y \in \mathcal{X}$, then $\rho(X) \leq \rho(Y)$.
	            \item[(ii)] \textit{Cash-Invariance:} If $X \in \mathcal{X}$ and $m \in \mathbb{R}$, then  $\rho(X+m) = \rho(X) + m.$
	        \end{itemize}
	    \end{defi}
	  \noindent A risk measure is normalized if $\rho(0) = 0$. If $\mathcal{X}$ is a space of random variables on some probability space $(\Omega, \mathcal{F}, \mathsf{P})$, the risk measure is called distribution-based if $\rho(X) = \rho(Y)$ whenever the distributions of $X$ and $Y$ under $\mathsf{P}$ are equal, i.e., $\mathsf{P}^X = \mathsf{P}^Y$ for $X, Y \in \mathcal{X}$. An excellent reference on scalar monetary risk measures is the book \cite{FS}. For a brief survey we refer to \cite{FW15}.} (DRM), as described in \cite{W2} and \cite{AcerbiCarlo}, operating on some vector space $\mathcal{X}$ of financial positions or insurance losses. More precisely, $\mathcal{X}$ is a vector space of measurable functions on a measurable space $(\Omega, \mathcal{F})$; we always assume that constant functions are included in $\mathcal{X}$. Our sign convention is that positive values correspond to losses and negative values to gains. 
	  
	  DRMs are a subset of comonotonic risk measures, cf.~\cite{FS},  and the definition of DRMs and their link to comonotonic risk measures is briefly summarized in Appendix \ref{app:DRM}. For our backtesting methodology, we need the following decomposition theorem which is related to the continuity properties of distortion functions, an important issue that is also investigated in \cite{DKLT} :
	    
		\begin{theo}\label{theo:uniqueconvexdecomposition}
		    Let $g$ be a distortion function. Then there exists a unique decomposition
		    \[ g(u) = c_r g_{sr}(u) + c_l g_{sl}(u) + c_c g_c(u)
		    \hspace{10mm} \forall u \in [0,1], \]
		    where $g_{sr}$, $g_{sl}$ are right- resp. left-continuous step distortion functions, $g_c$ is a continuous distortion function, and $c_r, c_l, c_c \in [0,1]$, $c_r + c_l + c_c =1 $. \\
		    In particular, the corresponding distortion risk measures satisfy the following relation:
		    $$ \rho_g  =  c_r \rho_{g_{sr}}  +    c_l \rho_{g_{sl}} +  c_c \rho_{g_c}  .   $$
		\end{theo}
		\begin{proof} See Appendix \ref{app:ch2}.
		\end{proof}
		\begin{rema}\label{rem:decomp}
		The decomposition of a distortion function $g$ according to Theorem \ref{theo:uniqueconvexdecomposition} can be computed as follows:
		    \[ g_{sr}(u) =  a \sum_{r \leq u} (g(r) - g(r-))
		    , \;  g_{sl}(u) = b \sum_{l < u} 
		    (g(l+) - g(l)) ,  \;  g_{c}(u) = c \left( g(u) - \frac{1}{a} g_{sr}(u) - \frac{1}{b} g_{sl}(u)
		    \right) \]
		    with normalizing constants $a,b,c$. Setting $c_r = a^{-1}$, $c_l = b^{-1}, c_c = c^{-1}$,  we have $$g= c_r g_{sr} + c_l g_{sl} + c_c g_c.$$
		    An example of a convex decomposition 
		    $ g= d_l h_{l} + d_r h_{r} $
		    according to Theorem~\ref{theo:convexcomposition} can then be obtained  by setting
		   \begin{eqnarray*}  
		   d_l &= & \left( c_l + \frac{c_c }{2}\right)  , \quad h_l \;= \; \left( c_l + \frac{c_c}{2} \right)^{-1} \left( c_l  g_{sl}(u) +
		         \frac{c_c}{2} g_c(u) \right), \\
		         d_r &=&  \left( c_r + \frac{c_c }{2}\right), \quad h_r \;= \; \left( c_r + \frac{c_c }{2}\right)^{-1} \left( c_r g_{sr}(u) + \frac{c_c}{2} g_c(u) \right),
		         \end{eqnarray*}
		  where $d_l + d_r = 1$, $d_l, d_r \geq 0$ and $h_l, h_r$ are left- resp.~right-continuous distortion functions.
		\end{rema}
		
Examples of DRMs and their distortion functions are given in \cite{MS2017}, for example, and are also provided in Appendix~\ref{app:DRM}.

	\section{Multinomial Tests for Distortion Risk Measures}\label{sec:MTDRM}
		
		\subsection{Preliminaries}\label{sec:original}
		
	The goal of this paper is to develop backtesting methods for general distortion risk measures and to improve upon existing approaches.\footnote{A multilevel V@R backtest is proposed by \cite{EKt} as an implicit backtesting method for AV@R; \cite{KLM} approximate AV@R by the sum of multiple V@R values at different levels and refine the original algorithm.} 
The true losses are generated according to a stochastic process $L= (L_t)_{t = 1, \dots, n}$ whose law is unknown. The process is adapted to the information filtration of the insurance company or bank and is observable. More specifically, we suppose that the information filtration is generated by $L$. For the purpose of risk measurement the firm uses a stochastic process $M= (M_t)_{t = 1, \dots, n}$ with known distribution, also called model. The conditional cumulative distribution functions $F_{M_t \vert M_{t-1}, \dots, M_1}$ are assumed to be continuous. We are interested in computing risk measures for future losses at time $t$ conditional on the information available at time $t-1$.

We begin by studying $AV@R_\alpha$ for some small $\alpha \in (0,1)$. An approximation of the risk measure can be obtained by considering an equidistant partition $\alpha_j := \frac{j}{m+1} \alpha$, $j = 0, 1, \dots, m+1$, of the interval $[0, \alpha]$. The $AV@R_\alpha$ at time $t-1$ of future losses at time $t$ is approximated by
		\begin{align*}
			AV@R^{t-1}_\alpha(M_t) &= \frac{1}{\alpha} \int_{0}^\alpha 
	    	V@R^{t-1}_\lambda(M_t) d\lambda \\
	    	&\approx \frac{1}{m+1} \left( V@R^{t-1}_{\alpha_1}(M_t)+ \dots+ V@R^{t-1}_{\alpha_{m+1}}(M_t)
	    	\right),
		\end{align*}	
		where the superscripts in $AV@R^{t-1}_\alpha$ and $V@R^{t-1}_\lambda$ indicate that that risk measures are computed from the conditional distribution of the arguments given the past, i.e., in this case from $F_{M_t \vert M_{t-1}, \dots, M_1}$ using the observations $M_{t-1}=L_{t-1}$, $\dots$, $M_1 = L_1$.
		
\cite{C} analyzes backtesting of V@R. Fixing a level $\beta$, we ask if 
$$ \forall t=1,2, \dots, n :  \quad  \mathsf{P} \left( L_t > V@R^{t-1}_{\beta}(M_t)   | L_{t-1}, \dots, L_1 \right) \; = \;  \beta .  $$
This question is equivalent to testing the hypothesis that the exception indicators $\mathds{1}_{\left\{L_t > V@R^{t-1}_{\beta}(M_t) \right\}},$ $t=1,2, \dots, n, $ are a sequence of independent Bernoulli random variables with parameter  $\beta$. This hypothesis is, of course, satisfied if the stochastic processes $L$ and $M$ possess the same law. It is often rephrased in terms of the following two properties:
		\begin{enumerate}[i.]
			\item The \textit{unconditional coverage hypothesis}:
			$\mathsf{E}\left[ \mathds{1}_{\left\{L_t > V@R^{t-1}_{\beta}(M_t) \right\}} \right] = \beta$
			for all $t $.
			\item The \textit{independence hypothesis}: The random variables $\mathds{1}_{\left\{L_t > V@R^{t-1}_{\beta}(M_t) \right\}}$, $t=1,2, \dots, n,$ are independent.
		\end{enumerate}
		
We now return to the equidistant partition 	$(\alpha_j)_{j=0,1, \dots, m+1}$ and define the random number of breached levels at time $t$ by	
\begin{equation}\label{eq:breached}
X_t  \; := \;  \sum_{j=1}^{m+1} \mathds{1}_{\left\{ L_t > V@R^{t-1}_{\alpha_j}(M_t) \right\}},\end{equation}
taking values in $\{0,1, \dots, m, m+1\}$. If the unconditional coverage hypothesis holds, the distribution of $X_t$ is multinomial with one trial, i.e., 
$ X_t \sim MN \left(1, (1-  \alpha_{m+1} , \alpha_{m+1} - \alpha_m, \dots, \alpha_1)\right). $
The observed cell counts $(O_0, O_1, \dots, O_{m+1})$, defined as
	    $O_i = \sum_{t=1}^n \mathds{1}_{\{ X_t = i\}}$, $i  = 0,1, \dots, m +1$, 
	    follow a multinomial distribution with $n$ trials
\begin{equation}\label{classical0}
 (O_0, O_1, \dots, O_{m+1})  \; \sim \;  MN \left(n, (
 1- \alpha_{m+1}, \alpha_{m+1} - \alpha_m, \dots, \alpha_1
 )\right), \end{equation}
if the independence hypothesis holds. This result provides a backtesting methodology for the question whether or not the model computations $AV@R^{t-1}_\alpha(M_t)$, $t=1,2, \dots,n$, are a proper basis for assessing the true risk. \cite{KLM} consider the null hypothesis  \eqref{classical0}.
        
This approach can also be used to backtest the models evaluated with a DRM, but it would neglect the weights introduced by the distortion function. Where the distortion function puts more weight, the corresponding losses are more important. Based on this observation, we propose a novel extension to the multinomial V@R backtests to DRMs that is better adapted to the significance of misspecifications.	
		
\subsection{Left- and Right-Continuous Distortion Functions}
		\label{sec:right-conti}

We begin with the case of left-continuous distortion functions. The general case will be studied in Section~\ref{sec:general}. 

For any distortion function $g$ we denote by $\rho_g^{t-1}(M_t)$ the distortion risk measure $\rho_g$ evaluated for the conditional distribution $F_{M_t \vert M_{t-1}, \dots, M_1}$ using the observations $M_{t-1}=L_{t-1}$, $\dots$, $M_1 = L_1$. The corresponding conditional quantile function is denoted by $q_{M_t}^{t-1}$. If $g$ is left-continuous, then the DRM can -- according to Theorem \ref{theo:ContDRM} -- be expressed as 
		\begin{equation}\label{eq:drm_rc}
		    \rho^{t-1}_g(M_t) = \int_{[0,1]} q^{t-1}_{M_t} (u )  d\bar g(u)
		\end{equation}
for a right-continuous distribution function $\bar g$ on the interval $[0,1]$. Letting $\bar G$ be a real-valued random variable, independent of $L$ and $M$, with distribution function $\bar g$ and $G= 1- \bar G$, we may rewrite the DRM risk measurement as
		$ \rho^{t-1}_g(M_t) = \mathsf{E} \left[ q^{t-1}_{M_t} (\bar G) \right] =  \mathsf{E} \left[ q^{t-1}_{M_t} ( 1-  G) \right] .  $

We now introduce a backtesting methodology that is based on a discrete approximation, generalizing the approach for AV@R. Let $0 = \alpha_0 < \alpha_1 < \dots < \alpha_m < \alpha_{m+1}= 1 $ be a partition of $[0,1]$ with $g(\alpha_j) - 
		g(\alpha_{j-1}) \neq 0$ for all $j =1, \dots, m+1$.
The exceptions indicators are replaced by randomized exception indicators:  For $t, j$, we let $G_{t,j}$ be independent random variables, independent of $L$ and $M$, with distribution $\mathcal{L} (G | G \in  [ \alpha_{j-1}, \alpha_j ) )$, i.e., the distribution of $G$ conditional on $G \in [ \alpha_{j-1}, \alpha_j )$. The randomized exception indicators that can be used for backtesting are
		\[ \mathds{1}_{t,j} = 
		\begin{cases}
				1 \hspace{15mm}	& \text{ if } L_t > q^{t-1}_{M_t} (1- G_{t,j}) \\
				0 & else.
			\end{cases}, \quad  j =1, \dots, m+1.
		\]
Obviously, if $L$ and $M$ possess the same law, then for all $j$:
$$
			\mathsf{P} \left( L_t > q^{t-1}_{M_t} (1-G_{t,j})\right) = \mathsf{P} \left( M_t > q^{t-1}_{M_t} (1-G_{t,j}) 	\right).$$
In this case, versions of the unconditional coverage hypothesis and independence hypothesis hold that can be used as the basis for a backtest.
		\begin{lem}\label{lem:RightContException}
			If $M$ and $L$ possess the same law, then 
			\begin{equation}\label{eq:ucc} \mathsf{E} \left[ \mathds{1}_{t,j} \right]= \frac{\mathsf{E} \left[ G \mathds{1}_{\{ G \in [\alpha_{j-1}, \alpha_j) \}}\right]}{g(\alpha_j ) - g(\alpha_{j-1})} \end{equation}
			for all $t=1, \dots, n$, and the random vectors $(\mathds{1}_{t,j})_{j=1,\dots, m+1}$, $t= 1, \dots, n$, are independent.
		\end{lem}
		\begin{proof}
		    See Appendix \ref{app:right-conti}.
		\end{proof}
The randomized numbers of breached levels are
		$X_t := \sum_{j=1}^{m+1} \mathds{1}_{t,j},$  $t \in \{1, \dots, n\}  . $
		\begin{lem}\label{lem:RightContBreachedLevels}
		Suppose that the unconditional coverage hypothesis \eqref{eq:ucc} holds and that the random vectors $(\mathds{1}_{t,j})_{j=1,\dots, m+1}$, $t= 1, \dots, n$, are independent.
		Then the number of breached levels $X_t$ satisfies
			\begin{eqnarray*}
			 \mathsf{P}( X_t \leq k) &= &1 - \frac{\mathsf{E} \left[ G \mathds{1}_{\{ G \in [
				\alpha_{m-k}, \alpha_{m-k+1}) \}} \right] }{g(
				\alpha_{m-k+1}) - g(\alpha_{m-k})}, \; \;  0 \le k \le m,  \\
			\mathsf{P}(X_t \leq m+1) &= & 1, 
			\end{eqnarray*}
			and the random variables $(X_t)_{t=1,\dots, n}$ are independent. 
		\end{lem}

		\begin{proof}
			See Appendix \ref{app:right-conti}.
		\end{proof}
    	The number of breached levels $X_t$ follows a multinomial distribution,
		$ X_t \sim MN (1, (p_0, \dots, p_{m+1})) , $
		where $MN(n, (p_0, \dots, p_{m+1}))$ denotes a multinomial distribution
		with $n$ trials and $m+2$ possible outcomes. The probabilities $p_k = \mathsf{P}(X_t = k)$ can be computed from Lemma~\ref{lem:RightContBreachedLevels}; they are stated explicitly in Theorem \ref{theo:RightContCellCount}.

The observed cell counts 
		\begin{equation}\label{def:cellcounts_right}
		    O_k := \sum_{t=1}^n \mathds{1}_{\{X_t = k\}}, \quad k = 0, 1, \dots, m+1
		\end{equation}
		are the key statistics of the multinomial backtest.
		\begin{theo}\label{theo:RightContCellCount} 
		   Suppose that the unconditional coverage hypothesis \eqref{eq:ucc} holds and that the random vectors $(\mathds{1}_{t,j})_{j=1,\dots, m+1}$, $t= 1, \dots, n$, are independent. Then the observed cell counts possess the multinomial distribution
		   \begin{equation}\label{eq:h0_l}
		   (O_0, O_1, \dots, O_{m+1})  \sim  MN (n, (p_0, p_1, p_2, \dots, p_{m+1} )), \quad \quad \quad\quad\quad\quad \quad \quad
		   \end{equation}
			\begin{eqnarray*}
				p_0 &=  &\mathsf{P}(X_t = 0) =1 - \frac{ \mathsf{E} \left[ G \mathds{1}_{\{ G \in 
				[\alpha_m, 1) \}} \right]}{g(1) - g(\alpha_m)} \\
				p_k &= & \mathsf{P}(X_t = k) = \frac{\mathsf{E} \left[ G
				\mathds{1}_{\{G \in [\alpha_{m+1-k},
				\alpha_{m+2-k} )\}} \right]}{g(\alpha_{m+2-k}) - g(\alpha_{m+1-k})}
				- \frac{\mathsf{E}\left[ G \mathds{1}_{\{G \in [\alpha_{m-k},
				\alpha_{m+1-k} )\}}\right]}{g(\alpha_{m+1-k}) - g(\alpha_{m-k})}, \quad 1 \le k \le m, \\
				p_{m+1} &= & \mathsf{P}(X_t = m+1) = \frac{ \mathsf{E} \left[ G \mathds{1}_{\{ G \in 
				[0, \alpha_1) \}} \right]}{g(\alpha_1)}. 
			\end{eqnarray*}
		\end{theo}
		\begin{proof}
			See Appendix \ref{app:right-conti}.
		\end{proof}
		
		As stated in Lemma \ref{lem:RightContException}, if $M$ and $L$ possess the same law, the conditions of Theorem \ref{theo:RightContCellCount} are satisfied. In generalization of previous results, we thus suggest to use the multinomial distribution $(O_0, O_1, \dots, O_{m+1})$ in \eqref{eq:h0_l} as a starting point for the analysis of the null hypothesis in a backtest for a DRM. The corresponding results for general distortion functions are stated in the next section.
		
		\begin{rema}
			In contrast to the approach of \cite{KLM} that we reviewed in Section~\ref{sec:original}, our approach includes an additional randomization. When applied to AV@R, the levels of the V@R-thresholds in the computation of the breached levels in \eqref{eq:breached} are randomized. This leads to alternative tests that are more powerful according to our case studies in Section \ref{sec:DistStudies}. The corresponding multinomial distribution of the observed cell counts $(O_0, O_1,
		\dots, O_{m+1})$ under the null hypothesis is characterized in Section \ref{sec:avarstudy} and in Appendix \ref{app:avarstudy} for
	 arbitrary partitions of $[0,1]$. 
		\end{rema}
		
\begin{rema}
In the case of right-continuous distortion functions the results need to be adjusted as follows: 
\begin{enumerate}[i.]
\item The distortion function $g$ is the distribution function of a probability measure on $[0,1]$. We denote by $G$ a random variable, independent of $L$ and $M$, with this distribution, and, for all $t, j$, by $G_{t,j}$  independent random variables, independent of $L$ and $M$, with distribution $\mathcal{L} (G | G \in  ( \alpha_{j-1}, \alpha_j ] )$, i.e., the distribution of $G$ conditional on $G \in (\alpha_{j-1}, \alpha_j ]$.
\item The randomized exception indicators are 
		\[ \mathds{1}_{t,j} = 
		\begin{cases}
				1 \hspace{15mm}	& \text{ if } L_t > q_{M_t}^{+, t-1} (1- G_{t,j}) \\
				0 & else
			\end{cases}, \quad  j =1, \dots, m+1.
		\]
\item If $M$ and $L$ possess the same law, then 
			$ \mathsf{E} \left[ \mathds{1}_{t,j} \right]= \frac{\mathsf{E} \left[ G \mathds{1}_{\{ G \in (\alpha_{j-1}, \alpha_j] \}}\right]}{g(\alpha_j ) - g(\alpha_{j-1})} $
			for all $t=1, \dots, n$. 
\item In Theorem~\ref{theo:RightContCellCount} all left-closed and right-open intervals must be replaced by left-open and right-closed intervals.
\end{enumerate}
With these technical modifications, all results, stated before for left-continuous distortion functions, also hold in the right-continuous case.
\end{rema}

		\subsection{General Distortion Functions }\label{sec:general}
		
		General DRMs are slight more challenging. In this case, we split the distortion function $g$ of the DRM into three components, a left-continuous, a right-continuous and a continuous part. This admits to work with a mixture of three distributions and to use a  similar approach as described in Section~\ref{sec:right-conti}. \\		
		
We denote by $\rho_g^{t-1}(M_t)$ the distortion risk measure $\rho_g$ evaluated for the conditional distribution $F_{M_t \vert M_{t-1}, \dots, M_1}$ using the observations $M_{t-1}=L_{t-1}$, $\dots$, $M_1 = L_1 $. According to Theorem \ref{theo:uniqueconvexdecomposition} this risk measurement may be rewritten as
	   \[ \rho_g^{t-1}( M_t) = c_r \rho_{g_{sr}}^{t-1} (M_t) + c_l \rho_{g_{sl}}^{t-1}(M_t) + c_c \rho_{g_c}^{t-1}
	    ( M_t)	. \]
We simplify the notation by writing  $ g_{l}, g_{r}$ instead of $g_{sl}, g_{sr}$. We denote by $G^l, G^r, G^c$ random variables, independent of $L$ and $M$, distributed according to $g_l, g_r, g_c$, respectively, and independent of the random variable $C$ that takes the values $l, r, c$ with probabilities $c_l, c_r, c_c$, respectively. We also choose $C$ independently of $L$ and $M$. Setting 
$$ G \; = \;  \mathds{1}_{\{ C = r\}} G^r  \; + \;   \mathds{1}_{\{ C = l\}} G^l  \; +  \;  \mathds{1}_{\{ C = c\}} G^c, $$
the random variable $G$ has the mixture distribution function $g = c_r g_r + c_l g_l + c_c g_c$. With this notation, the risk measurement can be expressed as 
	   \begin{align*}
	   \rho_g^{t-1} (M_t) & = c_r \mathsf{E} \left[ q^{+, t-1}_{M_t}( 1 - G^r) \right] + c_l \mathsf{E} \left[ q^{t-1}
	   _{M_t} ( 1 - G^l) \right] + c_c \mathsf{E} \left[ q^{t-1}_{M_t} ( 1 - G^c) \right] \\
	   		& = \mathsf{E} \left[ \mathds{1}_{\{ C = r\}} q^{+, t-1}_{M_t}(1 - G^r) + 
	   		\mathds{1}_{\{ C = l \}} q^{t-1}_{M_t}(1 - G^l) + \mathds{1}_{\{C = c\}} q^{t-1}_{M_t}( 1- G^c) \right] \\
	   		&= \mathsf{E} \left[ \mathds{1}_{\{ C = r\}} q^{+, t-1}_{M_t}(1 - G) + 
	   		\mathds{1}_{\{ C = l \}} q^{t-1}_{M_t}(1 - G) + \mathds{1}_{\{C = c\}} q^{t-1}_{M_t}( 1- G) \right].
		\end{align*}	  
This equation can be used as a basis for the construction of a backtesting procedure. 

Again we consider a partition $0 = \alpha_0 < \alpha_1 < \dots < \alpha_m < \alpha_{m+1} =1$ of $[0,1]$, but this time and deviating from Section~\ref{sec:right-conti} we impose the requirement that $g$ does not jump at  $\alpha_j$, ${j=1,2,\dots, m+1}$. This will be unproblematic, since for a given normative choice of a distortion function $g$ the selection of a corresponding partition $(\alpha_j)_{j=1,2,\dots, m+1}$ is quite flexible. This is contrast to the data-generating mechanism $L$ and the descriptive model $M$ whose performance and adequacy is tested.
		
		\begin{ass}\label{assu:alphagen} 
		~\\
			The function $g$ is continuous in $\alpha_j$ and  $g( \alpha_j ) - g( \alpha_{j-1}) \neq 0$ for  all $j =1, 
			\dots, m+1$.   
		\end{ass}
		\begin{rema} 
		\begin{itemize}
		\item[i.] Since $g$ is increasing, it
			posses only countable many discontinuities. Hence, Assumption~\ref{assu:alphagen}  does not substantially restrict the generality of the method.
		\item[ii.] The assumption ensures that the procedure will not generate different results for the intervals $(\alpha_{j-1}, \alpha_j]$,  $(\alpha_{j-1}, \alpha_j)$, $[\alpha_{j-1}, \alpha_j]$, or $[\alpha_{j-1}, \alpha_j)$, since  
		the set $\{\alpha_j: \; j =1,2, \dots, m+1\}$ has probability measure zero for the distribution function $g$.
		\end{itemize}	
		\end{rema}

	  	We are now in the position to define the randomized exception indicators that are used for backtesting. Letting $G_{t, j}$ be independent random variables, independent of $L$ and $M$, 
	  	with distribution $\mathcal{L}( G \vert G \in [ \alpha_{j-1}, \alpha_j))$, $G_{t,j}$ follows the corresponding mixture of the conditional distributions of $G^l, G^r, G^c$. With $C_{t,j}$, $t=1,2, \dots, n$, $j=1, \dots, m+1$, being independent replications of $C$, independent of $L$ and $M$, we define
	  	\begin{align*}
	  		\mathds{1}_{t,j} := \begin{cases}
	  			1 \hspace{10mm} & L_t > \mathds{1}_{\{ C_{t,j} = r\}} q^{+, t-1}_{M_t}(1 - G_{t,j})  
	  			+ \mathds{1}_{\{C_{t,j} = l \}} q^{t-1}_{M_t} ( 1 - G_{t,j})	 + \mathds{1}_{\{ C_{t,j}= c\}} q^{t-1}_{M_t} 
	  			(1 - G_{t,j}) \\
	  			0 & \text{else.}		
	  		\end{cases}
	  	\end{align*}
	  When the processes $L$ and $M$ possess the same law, suitable versions of the unconditional coverage hypothesis and 
	  independence hypothesis hold that can be used for backtesting. 
	  \begin{lem} \label{lem:GenExProb}
	  	If $M$ and $L$ possess the same law, then
	  	\begin{equation}\label{eq:uccgen}
	  		\mathsf{E}[\mathds{1}_{t,j}] = \frac{\mathsf{E}[ G \mathds{1}_{\{ G \in [ \alpha_{j-1}, \alpha_j) \} }]}
	  		{g(\alpha_j) - g(\alpha_{j-1})}
	  	\end{equation}  
	  	for all $t = 1 , \dots, n$,  and the random vectors $(\mathds{1}_{t,j})_{j=1,\dots, m+1}$, $t= 1, \dots, n$, are independent. 
	  \end{lem}
	  
	  \begin{proof}
	  	The proof is analogous to the proof of  Lemma~\ref{lem:RightContException}.	 
		 \end{proof}
We define the number of breached levels,
		$ X_t = \sum_{j=1}^{m+1} \mathds{1}_{t,j},$  $t =1, 2, \dots, n,$
		  and the observed cell counts,
		\begin{equation}\label{def:cellcounts_general}
		    O_k = \sum_{t=1}^n \mathds{1}_{\{ X_t = k \}}  , \quad\quad k = 0, 1, \dots, m+1.
		\end{equation}

		\begin{theo}\label{theo:genProbs}
		    Suppose that the unconditional coverage hypothesis (\ref{eq:uccgen}) holds and that the random vectors $(\mathds{1}_{t,j})_{j=1,\dots, m+1}$, $t= 1, \dots, n$, are independent. Then the following statements hold:
		    \begin{itemize}
		        \item[i)] The number of breached level has a multinomial distribution, i.e,  $X_t \sim \mathrm{MN}(1, (p_0,\dots, p_{m+1})$, $t = 1, \dots, n $, with 
		      
		        \begin{eqnarray*}
					p_0 &=  &\mathsf{P}(X_t = 0) =1 - \frac{ \mathsf{E} \left[ G \mathds{1}_{\{ G \in 
					[\alpha_m, 1) \}} \right]}{g(1) - g(\alpha_m)}  ,\\
					p_k &= & \mathsf{P}(X_t = k) = \frac{\mathsf{E} \left[ G
					\mathds{1}_{\{G \in [\alpha_{m+1-k},
					\alpha_{m+2-k} )\}} \right]}{g(\alpha_{m+2-k}) - g(\alpha_{m+1-k})}
					- \frac{\mathsf{E}\left[ G \mathds{1}_{\{G \in [\alpha_{m-k},
					\alpha_{m+1-k} )\}}\right]}{g(\alpha_{m+1-k}) - g(\alpha_{m-k})}, \\
					&& \quad\quad 1 \le k \le m, \\
					p_{m+1} &= & \mathsf{P}(X_t = m+1) = \frac{ \mathsf{E} \left[ G \mathds{1}_{\{ G \in 
					[0, \alpha_1) \}} \right]}{g(\alpha_1)} .
				\end{eqnarray*}
		        Moreover, the random variables $(X_t)_{t=1,\dots, n}$ are independent.  
		        \item[ii)] The observed cell counts possess the following multinomial distribution:
		        \begin{equation}\label{eq:nullgen}
		        	 (O_0, O_1, \dots, O_{m+1}) \sim \mathrm{MN} (n, (p_0, p_1, \dots, p_{m+1})).
				\end{equation}		         
		    \end{itemize}
		\end{theo}
		\begin{proof}
		    The proof is analogous to the proofs of  Lemma~\ref{lem:RightContBreachedLevels} and Theorem~\ref{theo:RightContCellCount}.
		\end{proof}
		
As in Section \ref{sec:right-conti}, we suggest to use the cell counts $(O_0, O_1, \dots,
		O_{m+1})$ for backtesting. If $M$ and $L$ possess the same law, its distribution is multinomial. 

	    \begin{rema}
	    The statement
	    \begin{equation}\label{eq:cond}
	    (O_0, O_1, \dots, O_{m+1})  \; \sim \;  \mathrm{MN} (n, (p_0, p_1, \dots, p_{m+1}))
	    \end{equation}
	 derived from Theorem~\ref{theo:genProbs} is typically not equivalent to $L$ and $M$ possessing the same distribution. The cell counts capture only certain features of the distribution of a process such that a null hypothesis \eqref{eq:cond}  of a multinomial distribution with these parameters includes a larger class of models. This condition approximates for a specific partition the statement that $M$ provides a reasonable model for measuring the considered DRM of the true losses $L$.
	 
	 The fact that \eqref{eq:cond}  is weaker than $L\overset{d}{=}M$ can be illustrated by a simple example.
	        Setting $m=0$ and 
	        $g(u) = \mathds{1}_{\{ 0.5 < u \leq 1\}}$, condition \eqref{eq:cond}  is equivalent to 
	        \[ \mathsf{E} [ \mathds{1}_{t,1} ]  \; = \;  \mathsf{P} ( L_t > q_{M_t}^{t-1}(0.5) )
	        \; = \;  0.5 \]
	        This is already true, if the conditional distributions of $L_t$ and $M_t$ given the past have the same median.	
	        \end{rema}

	\section{Distributional Simulations} \label{sec:DistStudies}
         We provide a numerical illustration of the DRM backtesting procedure. We consider the setting described in Section~\ref{sec:MTDRM}, extending the methodology of \cite{KLM}. Tests for the multinomial distribution are reviewed in Appendix~\ref{sec:tests},  compiling the relevant results from \cite{CK} and \cite{KLM}.  
         
         The null hypothesis $H_0$ states that the components of  $ L = (L_t)_{t=1,2, \dots, n}$ are independent random variables with standard normal distribution $\mathcal{N}$. This law is also used for the model $M=(M_t)_{t=1,2, \dots, n}$ on which risk computations are based. 
         
The test statistics are computed from the observed cell counts $(O_0, O_1, \dots, O_{m+1} )$ as defined in \eqref{def:cellcounts_right} and \eqref{def:cellcounts_general}. Under the null hypothesis the cell counts possess a multinomial distribution with $n$ trials and $m+2$ possible outcomes; the parameters $p_0, p_1, \dots, p_{m+1}$ will be computed in the case studies and depend on the chosen distortion function $g$ and the partition $(\alpha_j)_{j=0,1, \dots, m+1 }$. All tests are derived from asymptotic distributions;  finite sample distributions of the test statistics are not explicitly considered. This implies that the parameter $\kappa$ that specifies the level of the tests is typically not identically to their sizes. Instead we will compute the size of each test on the basis of simulations. 
         
We consider three alternatives $H_1$, labelled by T3, T5 and ST. In contrast to the normal distribution these 
include heavy tails and possibly skew. To specify the alternatives, consider an auxiliary process $\tilde L = (\tilde L_t)_{t=1,2, \dots, n}$ with independent components. In all cases, we assume that the true losses under the alternative $H_1$ are scaled and shifted such that they possess expectation $0$ and unit variance as in the standard normal model, i.e., 
$$ L_t \; = \; (\tilde L_t - E[\tilde L_t]) / \sqrt{\mathrm{Var}(\tilde L_t)}, \quad t = 1,2, \dots, n . $$
The alternatives T3 and T5 choose $\tilde L_t$, $t=1,2, \dots, n$, as student-t with three and five degrees of freedom, respectively. ST considers the skewed-t-distribution of \cite{FS1} that we recall in Appendix~\ref{app:ST}. 

        We test the null hypothesis $H_0$ versus the three alternative $H_1$ using the procedures introduced in Section \ref{sec:tests}. The size of the test, the probability of falsely rejecting the hypothesis if it is true, can be estimated from simulations as
    \begin{equation}\label{eq:simu} 
    \frac{1}{N} \sum_{i=1}^N \mathds{1}_{\{H_0 \text{ is rejected} \}} (L_{1,i}, \dots, L_{n,i}),
   \end{equation}
	    where the observed losses $L_{t,i}$,  $t = 1, \dots, n$, $i = 1, \dots, N $, are sampled from independent standard normal distributions.
        The power, the probability of correctly rejecting the hypothesis if the alternative is true, can again be estimated by \eqref{eq:simu}, but with  the observed losses $L_{t,i}$,  $t = 1, \dots, n$, $i = 1, \dots, N $, sampled under the distributions specified by the alternatives T3, T5, and ST, respectively. 
        
        In the numerical experiments we focus on three different DRMs: AV@R is described in Section \ref{sec:avarstudy}, GlueV@R in Section \ref{sec:gvsim}, and a DRM with a neither right- nor left-continuous distortion function in Section \ref{sec:genstudy}. Section \ref{sec:dsresults} analyzes and compares the results.  An additional case study with the risk measure Range Value at Risk (RV@R) is provided in Appendix~\ref{sec:rv@r}. We also apply the methodology to data from the S\&P 500 in Appendix~\ref{sec:s&p500}.

\subsection{Distortion Risk Measures}    

We consider three different distortion risk measures to illustrate our backtesting methodology.

    \subsubsection{AV@R}\label{sec:avarstudy}
        
 We consider AV@R at level $\alpha = 0.025$ as in Basel III corresponding to the distortion function
	    $$g(u) := \frac{u}{0.025} \mathds{1}_{\{ 0 \leq u \leq 0.025 \}} + \mathds{1}_{\{u > 0.025 \}}.$$
	    Since $g(u)$ is a continuous, we apply the framework of Section \ref{sec:right-conti}.  In contrast to the approach of \cite{KLM}, our multinomial backtest of AV@R includes an additional randomization. Our numerical analysis will indicate that this is beneficial for backtests. We set\footnote{A discussion on how to find a good partition is beyond the scope of the current paper. We choose for simplicity an equidistant partition to illustrate the potential of the backtesting methods.} $\alpha_0 = 0$, $\alpha_{m+1} = 1$ and
    	$ \alpha_j = \frac{j}{m+1} \alpha, $
    	for $j  = 1, \dots, m $.
    	The computation of the probabilities $(p_0, p_1, \dots, p_{m+1})$ can be found 
        in Appendix \ref{app:avarstudy}.
	        
    \subsubsection{GlueV@R}\label{sec:gvsim} 
    
The second risk measure that we consider is GlueV@R introduced in Example \ref{ex:drm} in Section~\ref{app:DRM} in the appendix. Specifically, we choose
        \[ \mathrm{GlueV@R}^{h_1, h_2}_{\beta, \alpha}(X) = \frac{1}{3} \mathrm{AV@R}_{0.05}(X) + \frac{1}{3}\mathrm{AV@R}_{0.01}(X) + \frac{1}{3} \mathrm{V@R}_{0.05}(X), \]
corresponding to the parameters $\alpha = 0.05 $, $\beta = 0.01$, $h_1 = 2/5$ and $h_2 = 2/3$.  Since a risk measure is a normative instrument, we fix these arbitrary parameters simply to illustrate the properties of the backtesting methodology. The  distortion function is left-continuous which allows us to apply the backtesting procedure described in Section \ref{sec:right-conti}. The partition is again set such that $(\alpha_1, \dots, \alpha_m)$ are equidistant in $[0, \alpha]$ and $\alpha_0 = 0$, $\alpha_{m+1}= 1$. The calculations of the probabilities $(p_0, p_1, \dots, p_{m+1})$
and a method to sample from the corresponding distortion function are described in Appendix \ref{app:gvsim}. 
        
	\subsubsection{A Distortion Function that Is Neither Right- Nor Left-Continuous}\label{sec:genstudy}
	
	 As a third example, we consider a DRM corresponding to a distortion function that is neither right- nor left-continuous. In this case, we apply the framework proposed in Section~\ref{sec:general}. 
	    Modifying the distortion function of the Glue-V@R, we study a distortion function 
        \begin{equation*}\label{eq:g4}
            g(u) = 
            \left\{\begin{array}{ll}
            \frac{h_1}{\beta} u &\, 0 \le u \le \beta \\
            h_2  + \frac{h_3-h_2}{\alpha-\beta}(u-\beta)  &\, \beta < u < \alpha \\
            1 &\, \alpha \le u \le 1,
            \end{array}\right.
        \end{equation*}
        with  $0 \le h_1 < h_2 < h_3 < 1$, $h_3 - h_2 + h_1 \leq 1$ and $0 \leq \beta < \alpha \leq 1$.
        By Theorem~\ref{theo:uniqueconvexdecomposition} the distortion function can be decomposed into a continuous, right-continuous and left-continuous part, $g(u) = c_r g_r(u) + c_l g_l(u) + c_c g_c(u)$,
        as shown in Figure~\ref{fig:DFDecomp}.  We use an equidistant partition $(\alpha_1, \dots, \alpha_m)$ of $[0, \alpha]$ with $\alpha_0 = 0$ and $\alpha_{m+1} = 1$. For the simulation study we set the parameters as $\alpha = 0.1$, $\beta = 0.01$, $h_1 = 1/5$, $h_2= 2/5$ and $h_3 = 2/3$. Appendix~\ref{app:4.3} provides additional information, i.e., the decomposition of $g$, the computation of the probabilities $(p_0, p_1, \dots, p_{m+1})$ according to Theorem \ref{theo:genProbs}, and a method for sampling from the mixture distribution of the distortion functions $g_r, g_l, g_c$.
        
        \begin{figure}
            \centering
            \includegraphics[scale = 0.5]{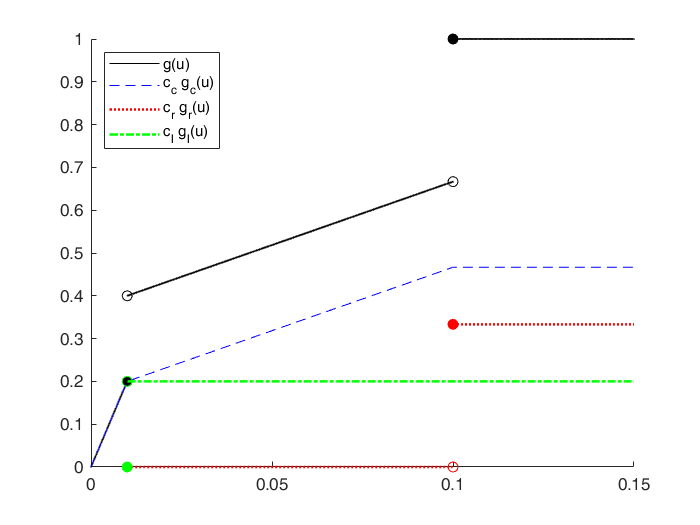}
            \caption{The decomposition of the distortion function $g(u)$ in right- and left-continuous parts scaled with $c_r$, $c_l$.}
            \label{fig:DFDecomp}
        \end{figure}

    \subsection{Results}\label{sec:dsresults}
    
 The results of the numerical experiments are displayed in Tables~\ref{tab:avarres}~--~\ref{tab:genres}. Throughout the experiment, we use $N = 20000$  samples to determine the size and power of the tests. The parameter $n$ signifies the length of the considered time series, the parameter $m$ determines the number of considered cells. The level $\kappa$ in the construction of the approximate tests\footnote{We will see in the numerical analysis that for large $m$ the desired level $\kappa$ might significantly deviate from the actual size for some of the approximate tests. The Nass' test, however, performs uniformly well.} is set to $5\%$.  Table~\ref{tab:avarres} shows the results for AV@R, Table~\ref{tab:avardiff} compares these to \cite{KLM}. Table~\ref{tab:gvres} shows the result of the backtest for GlueV@R and Table~\ref{tab:genres} for the distortion function that is neither right- nor left-continuous. We also provide a comparison of AV@R with the risk measure range value at risk (RV@R) in Appendix~\ref{sec:rv@r}.

	   \begin{table}[h!]
    	\captionsetup{font=scriptsize}
    	\begin{minipage}[t]{0.47\textwidth}
    	        \centering
    	        \tiny
    	        \begin{tabularx}{1.0\textwidth}{p{2mm} | X | X X X X X X X}
    	        \multicolumn{9}{c}{\textbf{Pearson}}  \\
    	        \hline \\
    	        $L_t$ & $n | m$ & 1 & 2 & 4 & 8 & 16 & 32 & 64 \\
    	        \hline \hline 
    	        $\mathcal{N}$ 
    	        & 250 &\cellcolor{g} 0.97   &  1.26    & 1.39  &  \cellcolor{p} 1.67  & \cellcolor{vp}   2.26 & \cellcolor{vp}   3.01 &  \cellcolor{vp}  4.35\\	        
    	        & 500 &\cellcolor{vg} 0.84 &\cellcolor{vg} 1.11  & \cellcolor{vg} 1.17  &  1.40 &  \cellcolor{p} 1.82 &  \cellcolor{vp}   2.42  &  \cellcolor{vp} 3.28 \\
    	         & 1000 &\cellcolor{g}   0.93 & \cellcolor{g}   0.97  &\cellcolor{g}   1.03 &   1.29  &  1.48 &   \cellcolor{p} 1.87  & \cellcolor{vp} 2.55\\    	        
    	        & 2000 &\cellcolor{g}  0.92   &\cellcolor{g}  1.02  & \cellcolor{g}  1.06   &\cellcolor{vg} 1.13 &   1.22 &  \cellcolor{p} 1.58   &\cellcolor{p} 1.95\\
    	        \hline 
    	        T3  &   250  & \cellcolor{p} 16.31 & \cellcolor{p} 24.17 & \cellcolor{p} 28.97 & \cellcolor{p} 27.41 & 38.34 & 38.00 & 36.83 \\
    	        & 500 & \cellcolor{p}  27.97 &  38.09 & 44.51 & 47.01 & 52.54 & 49.53 & 48.89 \\
    	        &1000&  52.75 & 66.15 & \cellcolor{g} 74.09 & \cellcolor{g}77.18 & \cellcolor{g}78.97 &\cellcolor{g} 76.79 & \cellcolor{g} 74.30\\
    	        &2000& \cellcolor{g} 83.33 &\cellcolor{g}  93.66 &\cellcolor{g}  97.01 &\cellcolor{g}  97.87 &\cellcolor{g}  97.63 & \cellcolor{g} 96.74 &\cellcolor{g}  95.75 \\
    	        \hline
    	        T5 & 250 & \cellcolor{p} 17.72 &\cellcolor{p} 23.61 &\cellcolor{p} 26.16 &\cellcolor{p} 26.12 &\cellcolor{p} 29.53 & 36.35 & 37.66\\
    	        & 500 &  \cellcolor{p} 25.56 & 32.20 & 36.42 & 40.58 & 43.95 & 43.84 & 47.80\\
    	        &  1000 & 42.80 & 51.59 & 57.49 & 61.40 & 64.09 & 64.29 & 64.52\\
    	        & 2000 & 68.10 & \cellcolor{g} 78.47 & \cellcolor{g} 85.17 &  \cellcolor{g}87.93 &  \cellcolor{g}88.62 &  \cellcolor{g}88.83 & \cellcolor{g} 87.27\\
    	        \hline
    	        ST & 250 & 35.58 & 44.04 & 49.95 & 49.21 & 52.42 & 60.90 & 58.81\\
    	        & 500 & 54.65 & 64.81 & \cellcolor{g}71.29 &\cellcolor{g}75.52 &\cellcolor{g} 78.99 & \cellcolor{g}77.45 & \cellcolor{g}78.26\\
    	        & 1000& \cellcolor{g}82.30 & \cellcolor{g}89.97 &\cellcolor{g} 93.13 & \cellcolor{g}95.13 &\cellcolor{g} 95.89 & \cellcolor{g}95.43 &\cellcolor{g} 95.61 \\
    	        &2000&  \cellcolor{g}97.87 & \cellcolor{g} 99.49 &\cellcolor{g} 99.77 &\cellcolor{g}99.90 & \cellcolor{g}99.95 & \cellcolor{g}99.90 &\cellcolor{g}99.88\\
    	        \end{tabularx} \\
    	        \vspace{3mm}
    	        \begin{tabularx}{1.0\textwidth}{p{2mm} | X | X X X X X X X}
    	        \multicolumn{9}{c}{\textbf{Nass}}  \\
    	        \hline \\
    	        $L_t$ & $n | m$ & 1 & 2 & 4 & 8 & 16 & 32 & 64 \\
    	        \hline \hline 
    	        $\mathcal{N}$ & 250 & 0.79  & \cellcolor{g} 0.95 & \cellcolor{g}   1.05 &   \cellcolor{vg} 1.12  &\cellcolor{g}  1.01  & \cellcolor{g}1.08 & \cellcolor{g}  1.05 \\
    	        & 500 & 0.77 &\cellcolor{g} 0.97   & \cellcolor{g} 0.94  &\cellcolor{g}   1.03  & \cellcolor{g}  1.08   &  \cellcolor{g}1.03  & \cellcolor{g} 1.10  \\
    	        & 1000&\cellcolor{vg} 0.88 &   \cellcolor{vg} 1.11  & \cellcolor{g}  0.93  &  \cellcolor{g} 1.03  &  \cellcolor{g} 1.08   &\cellcolor{vg}  1.12  & \cellcolor{g}  1.10\\
    	        & 2000  & \cellcolor{vg} 0.83 & \cellcolor{g}  0.97  & \cellcolor{g}  0.99  &  \cellcolor{g}1.02  &\cellcolor{g}  1.01  &  \cellcolor{g}1.08  &  \cellcolor{vg}1.11  \\        
    	        \hline
    	         T3 &  250   &  \cellcolor{p} 16.26 &\cellcolor{p} 21.61 & \cellcolor{p}21.97 &\cellcolor{p} 25.54 &\cellcolor{p} 20.41 &\cellcolor{p} 20.89 & \cellcolor{p}22.73\\
    	        & 500  &\cellcolor{p} 27.82 & 37.00 & 41.95 & 44.80 & 41.51 & 45.66 & 38.55  \\
    	        &1000&  51.06 & 65.32 & \cellcolor{g} 73.06 &  \cellcolor{g}75.30 & \cellcolor{g} 74.13 & \cellcolor{g} 74.07 & 67.04\\
    	        &2000&\cellcolor{g} 83.04 &\cellcolor{g} 93.36 &\cellcolor{g}96.78 &\cellcolor{g} 97.64 & \cellcolor{g}97.27 \cellcolor{g}& 95.98 & \cellcolor{g}94.62 \\
    	        \hline
    	         T5 & 250  & \cellcolor{p}17.31 &\cellcolor{p} 20.83 &\cellcolor{p} 20.21 &\cellcolor{p} 22.53 &\cellcolor{p} 19.60 & \cellcolor{p}20.44 &\cellcolor{p} 19.46\\
    	        & 500 &   \cellcolor{p}25.00 & 30.43 & 34.07 & 37.00 & 34.69 & 34.52 & 30.85 \\
    	        &  1000  & 41.93 & 50.69 & 56.35 & 58.89 & 59.13 & 58.52 & 53.65    \\
    	        &2000 &  67.67 & \cellcolor{g}77.70 & \cellcolor{g}84.69 & \cellcolor{g}87.26 &\cellcolor{g} 87.48 &\cellcolor{g} 86.43 & \cellcolor{g}83.57\\
    	        \hline
    	        ST & 250 & 35.43 & 41.40 & 41.73 & 46.53 & 42.03 & 43.10 & 43.64
    	        \\
    	        & 500 & 54.33 & 63.24 & 69.56 & \cellcolor{g}73.59 & \cellcolor{g}71.25 &\cellcolor{g} 73.45 & 67.83\\
    	        & 1000 & \cellcolor{g} 81.55 & \cellcolor{g} 89.66 &\cellcolor{g}  92.81 & \cellcolor{g} 94.49 & \cellcolor{g} 94.74 & \cellcolor{g} 94.41 & \cellcolor{g} 93.00\\
    	        &2000 & \cellcolor{g}97.84 & \cellcolor{g}99.45 &\cellcolor{g} 99.76 & \cellcolor{g}99.88 & \cellcolor{g}99.94 & \cellcolor{g}99.87 & \cellcolor{g}99.82
    	        \end{tabularx}\\
    	        \vspace{3mm}
    	        \begin{tabularx}{1.0\textwidth}{p{2mm} | X | X X X X X X X}
    	        \multicolumn{9}{c}{\textbf{LRT}}  \\
    	        \hline \\
    	        $L_t$ & $n | m$ & 1 & 2 & 4 & 8 & 16 & 32 & 64 \\
    	        \hline \hline 
    	        $\mathcal{N}$ & 250 & 1.36  &\cellcolor{g} 0.94 &    0.79 &   0.42 &  0.09 &   0.01   &   0 \\
    	        & 500 & \cellcolor{p}1.65  &  1.42  &  1.39 &  1.24   & 0.53   & 0.07     &    0 \\
    	        & 1000& \cellcolor{g} 1.08  &   1.23  &   1.36  & \cellcolor{p} 1.60  & \cellcolor{p} 1.68 &   0.76  &  0.04 \\
    	        & 2000& \cellcolor{g} 0.97   &\cellcolor{g} 1.08  &  \cellcolor{vg}1.20 &   1.30  &\cellcolor{p}  1.79 & \cellcolor{vp}  2.47  &  1.34  \\
    	        \hline 
    	        T3  &  250  &\cellcolor{p} 22.83 &\cellcolor{p} 22.89 &\cellcolor{p} 22.32 & \cellcolor{p}14.74 &\cellcolor{vp} 6.35 &\cellcolor{vp} 0.52 & \cellcolor{vp}  0.01  \\
    	        &500  &  31.26 & 39.86 & 48.78 & 48.14 & 34.79 & \cellcolor{vp}9.32 &\cellcolor{vp} 0.43 \\
    	        & 1000 & 53.79 & 65.89 & \cellcolor{g} 75.52 & \cellcolor{g}79.68 & \cellcolor{g}81.43 & 67.68 & \cellcolor{p}21.52  \\
    	        &2000 & \cellcolor{g}83.76 &\cellcolor{g} 94.01 &\cellcolor{g} 97.08 & \cellcolor{g}97.95 & \cellcolor{g}98.29 & \cellcolor{g}98.37 & \cellcolor{g}95.09\\
    	        \hline
    	        T5 & 250 &\cellcolor{p} 15.85 &\cellcolor{p} 16.90 & \cellcolor{p}16.65 & \cellcolor{p}11.09 & \cellcolor{vp}5.07 &\cellcolor{vp} 0.56 & \cellcolor{vp}    0 \\ 
    	        & 500 &\cellcolor{p}  21.04 &\cellcolor{p} 25.70 &\cellcolor{p} 29.79 & 31.01 & 2\cellcolor{p}3.49 & \cellcolor{vp} 6.82 & \cellcolor{vp} 0.27  \\
    	        &  1000   &  37.55 & 43.61 & 49.17 & 52.78 & 55.75 & 45.74 & \cellcolor{p} 12.20\\
    	        & 2000 & 63.24  &  \cellcolor{g}73.31 &\cellcolor{g} 79.98 & \cellcolor{g}81.65 & \cellcolor{g}82.34 & \cellcolor{g}84.45 & 77.24 \\
    	        \hline 
                ST & 250 & 30.84 & 35.43 & 37.67 & 30.90 & \cellcolor{p}19.61 & \cellcolor{vp} 03.91 & \cellcolor{vp} 0.07 \\
                & 500 &  48.46 & 57.35 & 64.61 & 67.94 & 61.45 & 32.82 & \cellcolor{vp}4.62 \\
    	         & 1000 & \cellcolor{g} 78.40 & \cellcolor{g} 86.22 &\cellcolor{g}  90.19 &\cellcolor{g}  92.63 & \cellcolor{g} 93.52 & \cellcolor{g} 90.09 & 62.56\\
    	        &2000 & \cellcolor{g}97.11 & \cellcolor{g}99.18 &\cellcolor{g} 99.66 & \cellcolor{g}99.82 &\cellcolor{g} 99.86 & \cellcolor{g}99.88 &\cellcolor{g} 99.63 \\
    	        \end{tabularx}
    	       \caption{Backtesting AV@R: Estimated size (for the hypothesis $H_0$ with distribution $\mathcal{N}$) and power in $\%$ (for the alternatives $H_1$ with distributions T3, T5, ST, respectively) 
    	        for the Pearson test, Nass test and LRT. The size is represented as the fraction of the true size according to our simulation divided by the desired level $\kappa=5\%$.  The colouring scheme for the size is as follows: Values between 0.8 - 1.2 are green, values between 0.9 - 1.1 are dark green; values above 1.5 are red, above 2 dark red.
	       	        The colouring scheme for the power is adopted from \cite{KLM}: Green refers to a power $\geq 70\%$; light red indicates a power $\leq 30\%$; dark red
    	        indicates poor results with a power $\leq 10\%$.}
    	        \label{tab:avarres}
    	\end{minipage}
    	\hfill
    	\begin{minipage}[t]{0.47\textwidth}
    	    \centering
    	    \tiny
    	    \begin{tabularx}{1.0\textwidth}{p{2mm} | X | X X X X X X X}
    	        \multicolumn{9}{c}{\textbf{Pearson}}  \\
    	        \hline \\
    	        $L_t$ & $n | m$ & 1 & 2 & 4 & 8 & 16 & 32 & 64 \\
    	        \hline \hline 
    	         $\mathcal{N}$ & 250 & 0.78  &  \cellcolor{g}0.94  &  \cellcolor{vg}1.12 &  \cellcolor{p}  1.70 &\cellcolor{vp}  2.10  &\cellcolor{vp}  2.82  & \cellcolor{vp} 4.30 \\
    	         & 500 & 0.78   &   \cellcolor{vg}0.88 &  \cellcolor{g}  1.04   & 1.32 &   \cellcolor{p} 1.72   &\cellcolor{vp}  2.46 & \cellcolor{vp}   3.24 \\
    	         & 1000&   \cellcolor{g}1.00  &   \cellcolor{g}1.04 &   \cellcolor{g} 1.00 &   \cellcolor{vg}  1.12   & 1.44   &\cellcolor{p} 1.80   &\cellcolor{vp} 2.40 \\
    	         & 2000& \cellcolor{g}1.00  & \cellcolor{g} 0.90   &\cellcolor{g} 0.96  & \cellcolor{g} 1.00 &   1.26  &  1.44  & \cellcolor{p}1.76 \\
    	        \hline 
    	        T3 & 250& \cellcolor{vg}12.21 &	\cellcolor{vg}13.97 &\cellcolor{vg}	18.77	& 6.61 & 	5.94 &	\cellcolor{vg}11 & 2.63 
\\
    	         & 500 &\cellcolor{g}  22.77 &\cellcolor{g}	22.39 &	\cellcolor{g} 28.81 &	\cellcolor{vg}18.61 &\cellcolor{g} 	20.34	& \cellcolor{vg} 13.33 & 9.09
 \\
    	        & 1000 &  \cellcolor{g}45.85 &	\cellcolor{g}39.45 &	\cellcolor{g}47.39 &\cellcolor{g}	28.98 &\cellcolor{g}	25.97 &	\cellcolor{g}21.99 &	\cellcolor{vg}18.5\\
    	       & 2000 &\cellcolor{g} 76.03 &	\cellcolor{g}46.46 &\cellcolor{g}	49.81 &	\cellcolor{vg}18.57 &	\cellcolor{vg}15.13 &	\cellcolor{vg}13.94 &	\cellcolor{vg}13.75 \\
    	        \hline
    	        T5 & 250 & \cellcolor{vg}13.62 &\cellcolor{vg}	13.41 &	\cellcolor{vg}15.96 &	5.32 &	7.13 &	9.35 &	3.46  \\
    	        & 500 &  \cellcolor{g}20.36 &	 \cellcolor{vg}16.5 &  \cellcolor{g}	20.72 & \cellcolor{vg}	12.18 &	 \cellcolor{vg}11.75 &	7.64 &	8\\
    	       	& 1000 & \cellcolor{g} 35.9 & \cellcolor{g}	24.89 &	\cellcolor{g}30.79 &	\cellcolor{vg} 13.2 &	\cellcolor{vg}11.09 &	9.49 &	8.72 \\
    	        & 2000 & \cellcolor{g}60.8 &	\cellcolor{g}31.27 &	\cellcolor{g} 37.97 &	8.63 &	6.12 &	6.03 &	5.27\\
    	        \hline 
    	        ST & 250 &\cellcolor{g} 30.18 &	\cellcolor{g}25.14 &\cellcolor{g}	31.05 &\cellcolor{vg}	10.51 &\cellcolor{vg}	13.72 &	\cellcolor{vg}14.6 &	8.31 \\
    	        &500&\cellcolor{g} 47.75 &\cellcolor{g}	29.91 &	\cellcolor{g}36.39 &\cellcolor{vg}	10.92 &	\cellcolor{vg}14.39 &	7.95 &	8.06\\
    	        &1000&\cellcolor{g}72.8 &\cellcolor{g}	27.67 &	\cellcolor{g}30.83 &	3.83 &	4.59 &	3.33 &	3.61 \\
    	        &2000& \cellcolor{g} 85.67 &	8.79 &	9.07 &	0.2 &	0.15 &	0.1	 & 0.18 \\
    	    \end{tabularx} \\ \vspace{3mm} 
    	     \begin{tabularx}{1.0\textwidth}{p{2mm} | X | X X X X X X X}
    	        \multicolumn{9}{c}{\textbf{Nass}}  \\
    	        \hline \\
    	        $L_t$ & $n | m$ & 1 & 2 & 4 & 8 & 16 & 32 & 64 \\
    	        \hline \hline     
    	        $\mathcal{N}$ & 250 &0.78   &  0.70   & \cellcolor{g}1.00 & \cellcolor{g}  0.94  & \cellcolor{g} 1.02  & \cellcolor{g} 1.00   & \cellcolor{g}0.96 \\
    	        & 500 & 0.78 &   0.78  & \cellcolor{g} 0.94  &\cellcolor{g}  0.94 & \cellcolor{g} 1.10 & \cellcolor{g}   1.10  & \cellcolor{g} 1.06 	 \\
    	        & 1000& \cellcolor{g}1.00 & \cellcolor{g} 0.96 &  \cellcolor{g}0.94  & \cellcolor{g} 0.98   &\cellcolor{g} 1.02   &\cellcolor{g} 1.06   &\cellcolor{g} 1.02 \\      
    	        & 2000 & \cellcolor{g}1.00  &  \cellcolor{vg}0.86  &\cellcolor{g}  0.90  &\cellcolor{g}  0.90  & \cellcolor{g} 1.06  & \cellcolor{g} 1.02 &   \cellcolor{g} 0.98 \\      
    	         \\
    	        \hline
    	        T3 & 250 &\cellcolor{vg} 12.66 &\cellcolor{vg}	16.01 &	9.87 &	\cellcolor{vg}10.74 &	7.01 &	7.69 &	9.13\\
    	        &500& \cellcolor{g}23.02 &	\cellcolor{g}21.5 & 	\cellcolor{vg}19.55 &	\cellcolor{vg}16.1 &	9.21 &\cellcolor{vg}	16.26 &	\cellcolor{vg}12.15 \\
    	  		&1000& \cellcolor{g}41.16 &\cellcolor{g}	30.12 &\cellcolor{vg}	18.96 &\cellcolor{vg}	15  &\cellcolor{vg}	12.73 &	\cellcolor{vg}19.37 &\cellcolor{vg}	12.34\\
    	        &2000& \cellcolor{g}66.44 &\cellcolor{g}	20.66 &	6.28 &	3.44 &	2.97 &	6.38 &	5.02 \\
    	        \hline 
    	        T5 & 250 &\cellcolor{vg}	 13.21 &\cellcolor{vg}	 13.13 &	7.41 &	8.43 &	6.2&	6.04 &	6.46\\
    	        &500&\cellcolor{vg} 19.8 &\cellcolor{vg}	16.13 &	\cellcolor{vg}13.57 &\cellcolor{vg}	12.5 &	8.09 &	8.52 &	8.15\\
    	        & 1000&  \cellcolor{g}35.03 & \cellcolor{g}	25.19 &	\cellcolor{vg}16.85 &	\cellcolor{vg}12.69 &	\cellcolor{vg}10.53 &	\cellcolor{vg}10.82 &	9.85 \\
    	        &2000& \cellcolor{g}60.37 &	\cellcolor{g}30.7 &	\cellcolor{vg} 15.09 &	9.06 &	6.68 &	6.23 &	6.57\\
    	        \hline
    	        ST & 250 &\cellcolor{g}30.03 &\cellcolor{g}	26.1 &	\cellcolor{vg}15.43 & \cellcolor{vg}	16.03 &	\cellcolor{vg} 11.83 &\cellcolor{vg} 	12.4 &\cellcolor{vg} 	12.94 \\
    	        & 500 & \cellcolor{g}	47.43 &	\cellcolor{g}	30.04 &\cellcolor{g}	21.96 &\cellcolor{vg}		17.39 &	9.85 &	\cellcolor{vg}	16.65 &	11.03 \\
    	        &1000& \cellcolor{g}72.05 &\cellcolor{g}	28.26 &	\cellcolor{vg}10.51 &	6.39 &	4.74  &	6.51 &	5.1 \\
    	        &2000&\cellcolor{g} 85.64 &	 8.75 &	1.16 &	0.18 &	0.24 &	0.37 &	0.32\\ 
    	   	\end{tabularx}  \\ \vspace{3mm} 
			\begin{tabularx}{1.0\textwidth}{p{2mm} | X | X X X X X X X}   
    	    	\multicolumn{9}{c}{\textbf{LRT}}  \\
    	        \hline \\
    	        $L_t$ & $n | m$ & 1 & 2 & 4 & 8 & 16 & 32 & 64 \\
    	        \hline \hline 
    	        $\mathcal{N}$ & 250 & \cellcolor{p}1.50  & \cellcolor{vp} 2.00  & 1.34  &  0.46  &  0.14  &  0.02   &      0 \\
    	        & 500 & \cellcolor{vg}1.18   & \cellcolor{vg}1.16  &  1.34 &   1.38  &  0.64  &  0.06    &     0 \\
    	        & 1000& 0.82  \cellcolor{vg}	 & \cellcolor{g}	 1.10  &  \cellcolor{g}	1.08 & \cellcolor{p}	  1.60 & \cellcolor{p}	  1.80 & \cellcolor{vg}	  0.88   & 0.04 \\
    	        &2000&\cellcolor{vg} 0.84 &  \cellcolor{g}  0.98 & \cellcolor{g}  1.04  & \cellcolor{vg} 1.20 &  \cellcolor{p} 1.78  & \cellcolor{vp} 2.60 &   1.40\\
    	        \hline
    	         T3 & 250 & \cellcolor{vg} 12.53 &	\text{-1.51} &	2.82 &	2.14 &	1.85 &	0.12 &	0.01\\
    	        &500& \cellcolor{g}21.76 &	\cellcolor{vg}13.66 &	9.88 &	5.74 &	5.69 &	1.52 &	0.13\\
    	        &1000&\cellcolor{g}44.09 &\cellcolor{vg}	18.69 &\cellcolor{vg}	11.22 &	5.58 &	4.23 &	5.18 &	6.52\\
    	    	&2000& \cellcolor{g}67.26 &	\cellcolor{vg}14.51 &	3.88 &	1.85 & 	1.29 &	0.87 &	1.69 \\
    	        \hline
    	        T5 & 250 & 8.95 &	 2.5 &	3.65 &	1.09 &	1.27 &	0.16 &	0\\
    	        &500& \cellcolor{vg} 14.54 & \cellcolor{vg}10.2 &	6.19 &	4.81 &	3.59 &	1.42 &	0.07\\
    	        &1000& \cellcolor{g}32.35 &\cellcolor{vg}	17.51 &	\cellcolor{vg}11.27 &	6.88 &	5.35 &	3.84 &	2.7 \\
    	        &2000& \cellcolor{g}57.44 &\cellcolor{g} 25.31 &\cellcolor{vg}	12.28 &	6.85 &	3.94 &	2.45 &	2.64 \\
    	        \hline
    	        ST & 250 &\cellcolor{g}22.84 &	\cellcolor{vg}10.83  &\cellcolor{vg} 10.47 &	5.2 &	5.11 &	1.21 &	0.07\\
    	        &500 & \cellcolor{g}40.56 &\cellcolor{g}	21.45 &\cellcolor{vg}11.71 & 8.94 &	7.55 &	5.62 &	1.42\\
    	        & 1000&\cellcolor{g} 71.5 &	\cellcolor{g} 23.92 &	8.69 &	4.63 &	 2.42 &	2.59 &	8.46 \\
    	        & 2000& \cellcolor{g}87.31 &	7.58 &	1.16 &	0.32 &	0.26 &	0.08 &	0.03 \\
    	   \end{tabularx}
    	       \caption{Backtesting AV@R: Comparison of the results in Table \ref{tab:avarres} to the results of \cite{KLM}. We display size and differences in power. The size is represented as the fraction of the true size according to \cite{KLM} divided by the desired level $\kappa=5\%$.  The colouring scheme for the size is as follows: Values between 0.8 - 1.2 are green, values between 0.9 - 1.1 are dark green; values above 1.5 are red, above 2 dark red. For the alternative T3, T5, and ST, the table shows the difference of the power of our method and the method of \cite{KLM}.  The colouring scheme for the power is as follows: Dark
    	   green are notable improvements of the power $\geq 20\%$; light green are improvements $\geq 10 \%$.} 
    	   \label{tab:avardiff}
    	\end{minipage}
    	\end{table}
	    
	    \begin{table}[h!]
    	\captionsetup{font=scriptsize}
    	\begin{minipage}[t]{0.47\textwidth}
    	        \centering
    	        \tiny
    	        \begin{tabularx}{1.0\textwidth}{p{2mm} | X | X X X X X X X}
    	        \multicolumn{9}{c}{\textbf{Pearson}}  \\
    	        \hline \\
    	        $L_t$ & $n | m$ & 1 & 2 & 4 & 8 & 16 & 32 & 64 \\
    	        \hline \hline 
    	        $\mathcal{N}$ & 250 & \cellcolor{vg}  0.88  & \cellcolor{g}  1.05  &  \cellcolor{vg}  1.18  &  1.44  & \cellcolor{p}   1.78 &  \cellcolor{vp}   2.46  &  \cellcolor{vp}  3.24 \\ 
    	        & 500 &\cellcolor{vg} 0.87  & \cellcolor{g} 1.04 & \cellcolor{g}  1.06   & \cellcolor{vg}1.20 &  1.46 &  \cellcolor{p} 1.87  & \cellcolor{vp}  2.38 \\
    	        & 1000& \cellcolor{g} 0.93  &   \cellcolor{g}0.96   & \cellcolor{g} 1.03  & \cellcolor{vg}  1.15   & 1.26  &  1.44  & \cellcolor{p}  1.94 \\
    	        & 2000&\cellcolor{g} 0.98 &  \cellcolor{g} 0.97 &\cellcolor{g}   1.00 &\cellcolor{g}   1.03 &\cellcolor{vg}  1.17  & 1.28  &\cellcolor{p}  1.51 \\
    	        \hline 
    	         T3  &   250  &  \cellcolor{p} 	29.3 & \cellcolor{p} 		26.04 &	 \cellcolor{p} 	19.25 & \cellcolor{p} 		18.96 &	 \cellcolor{p} 	23.03 &	 \cellcolor{p} 	19.91 & \cellcolor{p} 		18.44\\
    	        & 500 & 60.01 &	60.38 &	49.2 &	43.65 &	40.1 &	35.18 &	31.48\\
    	        &1000&\cellcolor{g} 93.08 &\cellcolor{g}94.11 &	\cellcolor{g}90.53 &\cellcolor{g}	87.19 &	\cellcolor{g}79.04 &	69.84 &	60.94 \\
    	        &2000&\cellcolor{g} 99.93 &	\cellcolor{g}99.97 &\cellcolor{g}	99.89 &\cellcolor{g}	99.87 &\cellcolor{g}	99.67 &\cellcolor{g}	98.15 &\cellcolor{g}	94.65 \\
    	        \hline
    	        T5 & 250 & \cellcolor{p} 16.7 &	\cellcolor{p} 18.4 &	\cellcolor{p} 15.65 &	\cellcolor{p} 16.89 &	\cellcolor{p} 19.75 &	\cellcolor{p} 19.51 &\cellcolor{p} 	21.94 \\
    	        & 500 &\cellcolor{p}  29.37 &	30.71 &	\cellcolor{p} 26.66 &	\cellcolor{p} 29.03 &	\cellcolor{p} 29.07 &	\cellcolor{p} 28.42 &\cellcolor{p} 	28.76 \\
    	        & 1000& 52.85 &	56.72 &	52.3 &	53.93 &	52.75 &	50.77 &	46.67 \\
    	        & 2000 &\cellcolor{g} 83.76 &\cellcolor{g}	87.16 &\cellcolor{g}	84.84 &	\cellcolor{g}87.35 &	\cellcolor{g}86.29 &\cellcolor{g}	83.06 &\cellcolor{g}	78  \\
    	        \hline
    	        ST & 250 & 39.9 &	41.34 &	35.68 &	37.65 &	43.33 &	39.7 &	39.29 \\
    	        & 500 & 68.42 &	\cellcolor{g}71.82 &	66.81 &	67.35 &	67.27 &	65.3 &	62.91\\
    	        & 1000&\cellcolor{g} 94.6 &\cellcolor{g}	96.18 &\cellcolor{g}95.25 &\cellcolor{g}	95.33 &\cellcolor{g}	94.45 &	\cellcolor{g}92.6 &	\cellcolor{g}90\\
    	        &2000&99.92 &	\cellcolor{g}99.97 &	\cellcolor{g}99.98 &\cellcolor{g}	100 &\cellcolor{g}	99.97 &\cellcolor{g}	99.93 &\cellcolor{g}	99.73\\
    	        \end{tabularx} \\
    	        \vspace{3mm}
    	        \begin{tabularx}{1.0\textwidth}{p{2mm} | X | X X X X X X X}
    	        \multicolumn{9}{c}{\textbf{Nass}}  \\
    	        \hline \\
    	        $L_t$ & $n | m$ & 1 & 2 & 4 & 8 & 16 & 32 & 64 \\
    	        \hline \hline 
    	        $\mathcal{N}$ & 250 &\cellcolor{vg}0.81  &\cellcolor{g}  0.93 & \cellcolor{g}0.97  & \cellcolor{g} 1.09   &\cellcolor{vg} 1.12   & \cellcolor{vg}1.12  &  \cellcolor{g}1.08  \\
    	        & 500 &\cellcolor{vg} 0.85  &\cellcolor{g}  0.97  &\cellcolor{g}   0.95 & \cellcolor{g}  1.00 & \cellcolor{g}   1.07   &\cellcolor{g}  1.10 &   \cellcolor{g} 1.06\\
    	        & 1000& \cellcolor{g} 0.91  &  \cellcolor{g}0.93  &  \cellcolor{g}0.97  &  \cellcolor{g}1.04  &  \cellcolor{g}1.02  & \cellcolor{g} 1.05  &  \cellcolor{vg}1.15 \\
    	        & 2000&\cellcolor{g} 0.97 &  \cellcolor{g} 0.93   &\cellcolor{g} 0.96  &\cellcolor{g}  0.96  &  \cellcolor{g}1.04 & \cellcolor{g} 1.08  &  \cellcolor{g} 1.04\\
    	        \hline
    	         T3 &  250   &\cellcolor{p} 27.12 &\cellcolor{p}	23.78 &	\cellcolor{p}16.51 &	\cellcolor{p}16.89 &	\cellcolor{p}14.36 &	\cellcolor{p}17.77 &\cellcolor{p}	13.68 \\
    	        & 500  &60.01 &	59.04 &	47.24 &	40.16 &	35.45 &	32.56 &	\cellcolor{p}25.81 \\
    	        &1000&\cellcolor{g}93.08 &\cellcolor{g}	93.91 &\cellcolor{g}	89.94 &\cellcolor{g}	85.86 &\cellcolor{g}	76.88 &	66.17 &	57.1 \\
    	        &2000&\cellcolor{g} 99.93 &\cellcolor{g} 99.97 &\cellcolor{g}	99.89 &	\cellcolor{g}99.86 &	\cellcolor{g}99.6 &\cellcolor{g}	97.84 &	\cellcolor{g}93.49\\
    	       \hline
    	        T5 & 250 &\cellcolor{p}16.23 &\cellcolor{p}	17.46 &\cellcolor{p}	13.96 &\cellcolor{p}	14.84 &\cellcolor{p}	13.26 &	\cellcolor{p}13.96 &\cellcolor{p}	10.8 \\
    	        & 500 &\cellcolor{p} 29.37 &\cellcolor{p}	29.97 &\cellcolor{p}	25.48 &\cellcolor{p}	26.71&	\cellcolor{p}25 &\cellcolor{p}	23.76 &\cellcolor{p}	19.86 \\
    	        & 1000 & 52.84 &	56.19 &	51.25 &	52.25 &	50.18 &	46.19 &	40.33\\
    	        &2000 &\cellcolor{g} 83.56 &\cellcolor{g}	87.03 &	\cellcolor{g}84.49 &	\cellcolor{g}86.9 &	\cellcolor{g}85.42 &	\cellcolor{g}81.41 &\cellcolor{g}	74.99\\
    	        \hline
    	        ST & 250 &50.38 &	44.84 &	35.28 &	35.65 &	32.98 &	36.06 &	30.22
    	        \\
    	        & 500 &\cellcolor{g}75.36 &\cellcolor{g}	76.41 &	68.63 &	66.25 &	63.83 &	62.02 &  55.15\\
    	        & 1000 &\cellcolor{g}97 &	\cellcolor{g}97.42 &	\cellcolor{g}96.04 &	\cellcolor{g}95.42 &	\cellcolor{g}93.96 &\cellcolor{g}	91.27 &	\cellcolor{g}88.14 \\
    	        &2000 &\cellcolor{g}99.96 &	\cellcolor{g}100 &	\cellcolor{g}99.98 &\cellcolor{g}	100 &	\cellcolor{g}99.97 &\cellcolor{g}	99.92 &\cellcolor{g} 99.68
    	        \end{tabularx}\\
    	        \vspace{3mm}
    	         \begin{tabularx}{1.0\textwidth}{p{2mm} | X | X X X X X X X}
    	        \multicolumn{9}{c}{\textbf{LRT}}  \\
    	        \hline \\
    	        $L_t$ & $n | m$ & 1 & 2 & 4 & 8 & 16 & 32 & 64 \\
    	        \hline \hline 
    	        $\mathcal{N}$ & 250 &\cellcolor{g} 1.08 & \cellcolor{g} 1.05  & 1.42  &  1.23  &  0.58  &  0.08  &       0  \\
    	        & 500 &  \cellcolor{vg} 1.13 & \cellcolor{vg}   1.17 &   1.26  & \cellcolor{p}  1.60 &  \cellcolor{p}   1.72  &  \cellcolor{vg}  0.80   & 0.05 \\
    	        & 1000  &\cellcolor{g} 1.01 & \cellcolor{g} 1.03  &  \cellcolor{vg}1.14  &  1.30 &  \cellcolor{p}1.81 &   \cellcolor{vp}2.46  &   1.40 \\
    	        & 2000 & \cellcolor{g}1.01  &\cellcolor{g} 1.02 & \cellcolor{g}  1.06 &\cellcolor{g} 1.10 &  1.38  &\cellcolor{vp} 2.14   &\cellcolor{vp} 3.90  \\
    	        \hline 
    	        T3  &  250  & 42.19 &	39.52 &	40 &	35.42 &	\cellcolor{p}12.83 &	\cellcolor{vp}1.44 &\cellcolor{vp}	0.01 \\
    	        & 500 & 67.25 &	69.66 &	66.42 &	69.22 &	66.22 &	34.17 &	\cellcolor{vp}2.08 \\
    	        & 1000 & \cellcolor{g}94.55 &	\cellcolor{g}95.83 &\cellcolor{g}	94.49 &\cellcolor{g}	94.24 &	\cellcolor{g}93.8 &	\cellcolor{g}93.45 &\cellcolor{g}	71.55 \\
    	        &2000 &\cellcolor{g} 99.93 &\cellcolor{g}	99.97 &\cellcolor{g}	99.93 &\cellcolor{g}	99.97 &\cellcolor{g}99.91 &\cellcolor{g}	99.83 &\cellcolor{g}	99.86 \\
    	        \hline
    	        T5 & 250 &\cellcolor{p}18.98 &	\cellcolor{p}17.92 &	\cellcolor{p}19.3 &\cellcolor{p}	17.99 &	\cellcolor{vp}8.38 &\cellcolor{vp}	1.18 &	\cellcolor{vp}0.01 \\ 
                & 500 &\cellcolor{p}27.46 &	\cellcolor{p}29.94 &	\cellcolor{p}29.22 &	33.61 &	33.92 &	\cellcolor{p}18.53 &\cellcolor{vp}	1.32 \\
                & 1000& 51.34 &	55.83 &	53.42 &	54.66 &	57.35 &	60.63&	40.72 \\
    	        & 2000&\cellcolor{g} 83.32 &\cellcolor{g}	86.92 &\cellcolor{g}	85.3 &\cellcolor{g}	87.03 &\cellcolor{g}85.81 &\cellcolor{g}	85.69&	\cellcolor{g}88.4 \\
    	        \hline 
    	        ST & 250 &44.8 &	43.72 &	43.95 &	43.5 &	\cellcolor{p}	26.86 &	\cellcolor{vp}6.19 &\cellcolor{vp}	0.16\\
    	        & 500 &68.62 &	\cellcolor{g}	73.04 &\cellcolor{g}		71.3 &\cellcolor{g}		75.11 &\cellcolor{g}		75.69 &	57.18 &	\cellcolor{p}	12.13 \\
    	        & 1000 &\cellcolor{g}	94.7 &	\cellcolor{g}	96.5 &	\cellcolor{g}	95.91 &\cellcolor{g}		96.43 &\cellcolor{g}		96.53 &\cellcolor{g}		96.99 &\cellcolor{g}		90.08 \\
    	        &2000 &\cellcolor{g}	99.92 &\cellcolor{g}		 99.97 &\cellcolor{g}	 99.98 & \cellcolor{g}	100 &\cellcolor{g}		99.96 &	\cellcolor{g}	99.97 &\cellcolor{g}		99.97
    	  \end{tabularx}
    	       \caption{Backtesting GlueV@R: Estimated size (for the hypothesis $H_0$ with distribution $\mathcal{N}$) and power in $\%$ (for the alternatives $H_1$ with distributions T3, T5, ST, respectively)
    	        for the Pearson test, Nass test and LRT. The size is represented as the fraction of the true size according to our simulation divided by the desired level $\kappa=5\%$.  The colouring scheme for the size is as follows: Values between 0.8 - 1.2 are green, values between 0.9 - 1.1 are dark green; values above 1.5 are red, above 2 dark red.
	       	        The colouring scheme for the power is adopted from \cite{KLM}: Green refers to a power $\geq 70\%$; light red indicates a power $\leq 30\%$; dark red
    	        indicates poor results with a power $\leq 10\%$.}
    	  \label{tab:gvres}
    	\end{minipage}
    	\hfill
    	\begin{minipage}[t]{0.47\textwidth}
    	    \centering
    	    \tiny
    	   \begin{tabularx}{1.0\textwidth}{p{2mm} | X | X X X X X X X}
    	        \multicolumn{9}{c}{\textbf{Pearson}}  \\
    	        \hline \\
    	        $L_t$ & $n | m$ & 1 & 2 & 4 & 8 & 16 & 32 & 64 \\
    	        \hline \hline 
    	      	$\mathcal{N}$ & 250 & \cellcolor{g} 0.91   &  \cellcolor{g}0.97& \cellcolor{g} 1.00  &  1.22  &  1.43 & \cellcolor{p} 1.93  &\cellcolor{vp}  2.30 
    	      \\
    	        & 500 &  \cellcolor{g}1.02 &   \cellcolor{g} 0.92 & \cellcolor{g}  1.07  & \cellcolor{vg} 1.11  &  1.23  &  1.44   &\cellcolor{p} 1.88
    	       \\
    	        & 1000 & \cellcolor{g}  1.06 & \cellcolor{g}   1.02 &  \cellcolor{g}  1.06  & \cellcolor{g}  1.06  & \cellcolor{g} 1.09  &  1.30  & \cellcolor{p}  1.52\\
    	        & 2000 &\cellcolor{g}0.96 &   \cellcolor{g} 1.02  &\cellcolor{g}  0.94  & \cellcolor{g}1.03 &  \cellcolor{g} 1.09 &  \cellcolor{g}1.06   & 1.22 \\
    	        \hline 
    	        T3  &   250  & 62.31 &	56,63 &	40.22 &	\cellcolor{p}15.42 &	\cellcolor{p}13.66 &	\cellcolor{p}10.98 &\cellcolor{vp}	8.92 \\
    	        & 500 &\cellcolor{g}94.67 &	\cellcolor{g}93.06 &\cellcolor{g}	87.89 &68.13 &	47.84 &	30.08 &\cellcolor{p}	19.8 \\
    	        &1000&\cellcolor{g} 99.94 &\cellcolor{g}	99.95 &\cellcolor{g}	99.89 &\cellcolor{g}	99.41 &\cellcolor{g}	97.72 &\cellcolor{g}	85.58 &	58.99 \\
    	        &2000& \cellcolor{g}100	 &\cellcolor{g}	100	 &\cellcolor{g}	100	 &	\cellcolor{g}100	 &	\cellcolor{g}100	 &\cellcolor{g}	100	 &\cellcolor{g}	99.27 \\
    	     	\hline
    	        T5 & 250 &\cellcolor{p}20.95 &\cellcolor{p}	20.06 &\cellcolor{p}	14.29 &\cellcolor{vp}	8.54 &\cellcolor{p}	10.98 &	\cellcolor{p}10.7 &	\cellcolor{p}11.11\\
    	        & 500 & 41.06 &	40.21 &	33.11 &\cellcolor{p}	20.04 &\cellcolor{p}	21.55 &\cellcolor{p}	19.13 &	\cellcolor{p}16.5\\
    	        & 1000& \cellcolor{g}72.12 &\cellcolor{g}	73.24 &	67.47 &	52.55 &	51.34 &	42.45 &	35.06 \\
    	        & 2000 &\cellcolor{g} 96.58 &\cellcolor{g}	97.15 &	\cellcolor{g}96.4 &	\cellcolor{g}91.7 &	\cellcolor{g}91.38 &	\cellcolor{g}85.29 &\cellcolor{g}	74.42 \\
    	        \hline
    	        ST & 250 &60.06 &	58.77 &	47.03 &\cellcolor{p}	26.51 &	\cellcolor{p}29.97 &	\cellcolor{p}26.42 &\cellcolor{p}	23.73 \\
    	        & 500 &\cellcolor{g} 92.67 &\cellcolor{g}	92.54 &	\cellcolor{g}88.77 &\cellcolor{g}	74.05 &	67.87 &	56.88 &	47.83 \\
    	        & 1000&\cellcolor{g}99.86 &	\cellcolor{g}99.94 &	\cellcolor{g}99.82 &\cellcolor{g}	99.4 &	\cellcolor{g}98.7 &\cellcolor{g}	95.3 &	\cellcolor{g}87.76\\
    	        &2000&\cellcolor{g}100	 &\cellcolor{g}	100	 &\cellcolor{g}	100	 &\cellcolor{g}	100	 &	\cellcolor{g}100	 &	\cellcolor{g}100	 &\cellcolor{g}	99.96\\
    	        \end{tabularx} \\
    	        \vspace{3mm}
    	        \begin{tabularx}{1.0\textwidth}{p{2mm} | X | X X X X X X X}
    	        \multicolumn{9}{c}{\textbf{Nass}}  \\
    	        \hline \\
    	        $L_t$ & $n | m$ & 1 & 2 & 4 & 8 & 16 & 32 & 64 \\
    	        \hline \hline 
    	        $\mathcal{N}$ & 250 &\cellcolor{vg}0.87  & \cellcolor{g} 0.91 &\cellcolor{g}   0.90  &\cellcolor{g} 1.06   &\cellcolor{g} 1.08 & \cellcolor{vg}  1.13  &\cellcolor{g}  1.08\\
    	        & 500 & \cellcolor{g}0.97 & \cellcolor{vg}  0.88 & \cellcolor{g} 1.01   & \cellcolor{g}1.00 &  \cellcolor{g} 1.00   &\cellcolor{g} 1.05  &  \cellcolor{vg}1.11\\
    	        & 1000& \cellcolor{g}  1.04  & \cellcolor{g}   0.99  & \cellcolor{g}   1.02 & \cellcolor{g}   0.99 &    \cellcolor{g} 0.98  &  \cellcolor{g}  1.09  &  \cellcolor{g}  1.09\\
    	        & 2000& \cellcolor{g}0.95 & \cellcolor{g}  1.00 & \cellcolor{g}  0.92  &\cellcolor{g}  1.00 & \cellcolor{g}  1.04 &  \cellcolor{g} 0.97 &   \cellcolor{g} 1.00\\
    	        \hline
    	        T3 &  250   &  61.07 &	54.45 &	37.19 &\cellcolor{p}		12.86 &	\cellcolor{p}	11.13 &\cellcolor{vp}		9.88 &\cellcolor{vp}		6.62 \\
    	        & 500  &\cellcolor{g}	 93.72 &	\cellcolor{g}	92.41 &	\cellcolor{g}	87.04 &	65.86 &	43.77 &\cellcolor{p}		26.51 &\cellcolor{p}		17.62\\
    	        &1000& \cellcolor{g}	99.94 &\cellcolor{g}	 99.95 &\cellcolor{g}		99.88 &\cellcolor{g}		99.36 &\cellcolor{g}		97.37 & \cellcolor{g}	83.04 &	54.73 \\
    	        &2000&100\cellcolor{g}		 &\cellcolor{g}		100	 &\cellcolor{g}		100	 &	\cellcolor{g}	100 &	\cellcolor{g}	100 &	\cellcolor{g}	100	 &\cellcolor{g}		99.14\\
    	        \hline
    	        T5 & 250 &\cellcolor{p}	20.39 &\cellcolor{p}		19.1 &\cellcolor{p}		13.08 &\cellcolor{vp}		7.37 &	\cellcolor{vp}	8.66 &\cellcolor{vp}		8.23 &\cellcolor{vp}		6.35 \\
    	        & 500 & 39.68 &	39.11 &	31.98 &	\cellcolor{p}	18.67 &	\cellcolor{p}	19.61 &\cellcolor{p}		16.18 &	\cellcolor{p}	12.67\\
    	        & 1000 & \cellcolor{g} 72.03 & \cellcolor{g}	72.81 &	67  &	51.48 &	49.75 &	39.6 &	31.21 \\
    	        &2000 & \cellcolor{g} 96.54 & \cellcolor{g}	97.11 & \cellcolor{g}	96.33 & \cellcolor{g}	91.47 & \cellcolor{g}	90.94 & \cellcolor{g}	84.33 & \cellcolor{g}	72.28 \\
    	        \hline
    	        ST & 250 & 59.51 &	57.23 &	44.3 &\cellcolor{p}	23.68 &\cellcolor{p}	25.75 &\cellcolor{p}	24.42 &\cellcolor{p}	18.5\\
    	        & 500 &\cellcolor{g} 91.72 &	\cellcolor{g}92.01 &\cellcolor{g}	88.11 &\cellcolor{g}	72.21 &	64.9 &	53.16 &	43.88\\
    	        & 1000 &\cellcolor{g}99.86 &\cellcolor{g}	99.93 &\cellcolor{g}	99.82 &\cellcolor{g}	99.36 &\cellcolor{g}	98.58 &\cellcolor{g}	94.48 &\cellcolor{g}	85.74 \\
    	        &2000&\cellcolor{g}100	 &\cellcolor{g}	100 &\cellcolor{g}	100	&	\cellcolor{g}100	 &\cellcolor{g}	100	 &\cellcolor{g}100	&\cellcolor{g}	99.96
    	        \end{tabularx}\\
    	        \vspace{3mm}
    	        \begin{tabularx}{1.0\textwidth}{p{2mm} | X | X X X X X X X}
    	        \multicolumn{9}{c}{\textbf{LRT}}  \\
    	        \hline \\
    	        $L_t$ & $n | m$ & 1 & 2 & 4 & 8 & 16 & 32 & 64 \\
    	        \hline \hline 
    	        $\mathcal{N}$ & 250 & \cellcolor{p}1.54 &\cellcolor{g}  1.10 & \cellcolor{vg} 1.17  & \cellcolor{p} 1.65 & \cellcolor{p} 1.71 &\cellcolor{vg}  0.80 &   0.04 \\
    	        & 500 & \cellcolor{g} 1.07 &  \cellcolor{g}1.05 &  \cellcolor{vg} 1.17 &  1.38  &  \cellcolor{p} 1.83 &    \cellcolor{vp}2.49  & 1.37\\
    	        & 1000  &   \cellcolor{vg}1.12 &   \cellcolor{g}  1.07   &   \cellcolor{g}1.09  &  \cellcolor{g} 1.10 &    1.32  &  \cellcolor{vp} 2.16   &   \cellcolor{vp}3.79 \\
    	        & 2000 &\cellcolor{g} 0.98 &  \cellcolor{g} 1.02  & \cellcolor{g} 0.97  &\cellcolor{g}  1.05  & \cellcolor{vg}1.15 &  1.36 & \cellcolor{vp}  2.69\\
    	        \hline 
    	        T3  &  250 &\cellcolor{g} 70.37 &	68.39 &	64.67 &	60.66 &	55.03 &	\cellcolor{p}13.94 &\cellcolor{vp}	0.1 \\
    	        &500  & \cellcolor{g}95.55 &	\cellcolor{g}95.47 &\cellcolor{g}	93.66 &	\cellcolor{g}88.75 &\cellcolor{g}	88.31 &\cellcolor{g}	86.59 &	41.03 \\
    	        & 1000 &\cellcolor{g}99.98 &\cellcolor{g}	99.97 &	\cellcolor{g}99.95 &\cellcolor{g}	99.79 &	\cellcolor{g}99.64 &\cellcolor{g}	99.34 &	\cellcolor{g}99.26 \\
    	        &2000 &\cellcolor{g}100	 &\cellcolor{g}	100	 &	\cellcolor{g}100	 &	\cellcolor{g}100	 &	\cellcolor{g}100	 &	\cellcolor{g}100	 &\cellcolor{g}	100 \\
    	        \hline
    	        T5 & 250 &\cellcolor{p}23.86 &\cellcolor{p}	23.66 &	\cellcolor{p}22.31 &\cellcolor{p}	22.71 &	\cellcolor{p}22.12 &\cellcolor{vp}	7.46 &\cellcolor{vp}	0.17 \\ 
                & 500 & 42.74 &	43.38 &	40.33 &	34.53 &	39.87 &	42.01 &	\cellcolor{p}17.36 \\
                & 1000&\cellcolor{g} 74.04 &\cellcolor{g}	75.52 &\cellcolor{g}	71.97 &	63.42 &	64.29 &	65.74 &\cellcolor{g}	71.6 \\
    	        & 2000&\cellcolor{g}96.97 &\cellcolor{g}	97.4 &\cellcolor{g}	97.04 &\cellcolor{g}	94.11 &\cellcolor{g}	94.09 &\cellcolor{g}	91.54 &\cellcolor{g}	91.32 \\
    	        \hline 
    	        ST & 250 &65.23 &	66.43 &	63.11 &	59.75 &	59.07 &	\cellcolor{p}25.49 &	\cellcolor{vp}0.82\\
    	        & 500 &\cellcolor{g}93.35 &	\cellcolor{g}94.42 &\cellcolor{g}	92.84 &\cellcolor{g}	88.36 &	\cellcolor{g}89.5 &\cellcolor{g}	89.27 &	60.48\\
    	        & 1000 &\cellcolor{g}99.89 &\cellcolor{g}	99.95 &\cellcolor{g}	99.89 &	\cellcolor{g}99.81 &	\cellcolor{g}99.69 &\cellcolor{g}	99.43 &\cellcolor{g}	99.56 \\
    	        &2000 &\cellcolor{g} 100	 &	\cellcolor{g}100	 &	\cellcolor{g}100	 &\cellcolor{g}	100	 &	\cellcolor{g}100	 &\cellcolor{g}	100	 &	\cellcolor{g}100
    	  \end{tabularx}
    	       \caption{Backtesting the DRM corresponding to a distortion function that is neither left- nor right-continuous: Estimated size (for the hypothesis $H_0$ with distribution $\mathcal{N}$) and power in $\%$ (for the alternatives $H_1$ with distributions T3, T5, ST, respectively) 
    	        for the Pearson test, Nass test and LRT. The size is represented as the fraction of the true size according to our simulation divided by the desired level $\kappa=5\%$.  The colouring scheme for the size is as follows: Values between 0.8 - 1.2 are green, values between 0.9 - 1.1 are dark green; values above 1.5 are red, above 2 dark red.
	       	        The colouring scheme for the power is adopted from \cite{KLM}: Green refers to a power $\geq 70\%$; light red indicates a power $\leq 30\%$; dark red
    	        indicates poor results with a power $\leq 10\%$.}
    	  \label{tab:genres}
    	\end{minipage}
    	\end{table}
    	
        \textit{Backtesting AV@R.} 
        Table~\ref{tab:avarres} shows the size and power of the tests when backtesting AV@R. The size is represented as the fraction of the true size according to our simulation divided by the desired level $\kappa=5\%$. The size of the Pearson test is only close to the desired level for sufficiently small $m$. In the case of the LRT, the size of the test is not always very close to $\kappa$, sometimes smaller, sometimes larger. Table~\ref{tab:avardiff} shows  the results of the approach\footnote{Their results are displayed in Table 3 of their paper, \cite{KLM}.} of \cite{KLM} for comparison. The qualitative behavior of the size of the Pearson and the LRT test is not very different for both methodologies. Both tests are not ideal for backtesting AV@R. In contrast, the size of Nass' test is very close to desired level of $\kappa=5\%$ in all cases,  as displayed in Table~\ref{tab:avarres}. This is qualitatively not very different from \cite{KLM}, see Table~\ref{tab:avardiff}. This indicates that the Nass' test is the preferred choice for backtesting AV@R, while the Pearson and the LRT test should not be chosen.
        
Table~\ref{tab:avarres} also shows the power of our method for the different alternatives and various choices of the parameters $m$ and $n$. As expected, the power is good if $n$ is large. Most interesting is the Nass' test. For $m$ smaller than $8$, the power exceeds for $n= 500$ a value of $40\%$ for T3, $25\%$ for T5, and $66\%$ for ST, respectively. For $m$ smaller than $64$ and $n=2000$, the power exceeds $93\%$ for T3, $74\%$ for T5, and $99.6\%$ for ST, respectively.

In comparison to \cite{KLM}, our method improves the power of all tests for most case studies. Table~\ref{tab:avardiff} shows for the alternatives the difference between the power of our method and the results of \cite{KLM}. Positive numbers indicate that our approach improves the power. Except for one entry, all numbers are positive. In particular, for the Nass' test our method uniformly improves the power substantially. 

Our randomized method removes the discretization error. Here, $m$ can be interpreted as a stratification parameter. If $m$ becomes larger, stratification moves our method closer to the method of \cite{KLM} with the same $m$. At the same time, the discretization error becomes smaller in their method. This is the reason why the improvement of our method in comparison to  \cite{KLM} becomes smaller, when $m$ increases.  However, as the experiments show, the improvements are still not negligible and large $m$ does not necessarily lead to the best power.

        \textit{Comparing different risk measures.} We now study the size across the three considered risk measures. The estimated size of the tests is shown for different $n$ and $m$ in Tables~\ref{tab:avarres}, \ref{tab:gvres} and \ref{tab:genres}. As before, the size is represented as the fraction of the true size according to our simulation divided by the desired level $\kappa=5\%$. The data clearly demonstrate that Nass' test performs well in terms of approximating the desired level and clearly outperformes the Pearson and the LRT test in this respect. The latter two work well if $m$ is not too large and $n$ is not too small. 
        
         In all case studies, as expected, the power is best for $n$ large. Good backtests require $n$ to be larger than 500 to 1000 for the chosen hypothesis and alternatives. The best results are obtained for $2 \leq m \leq 16$. There is no indication that Nass' test performs worse than the other tests in correctly identifying the alternatives if they are true. 
        
        \textit{Conclusion.} The Nass' test outperforms the Pearson and the LRT test, since its size is generally close to the desired level and its overall power is not worse than or at least rather close to the power of the alternative tests across most case studies.

    \section{An Application to Asset-Liability Management}\label{sec:ALM}
   
Financial institutions need to manage their risks arising from the evolution of their assets and liabilities. Asset-liability management (ALM) requires probabilistic models that enable a stochastic projection of the arising risks into the future. In this section, we apply our backtesting method to a company's net asset value in order to validate risk measurements in an ALM model.
    
    \subsection{The Model}
    Inspired by  \cite{WeberHamm2014} and 
    \cite{HammKnispelWeber2019}, we consider the assets and liabilities of a non-life insurance firm.  
    Time is discrete and enumerated by $t = 1, \dots, n$. Each time period could be interpreted as years, months, weeks or even days.     
    Denote with $A_t$, $L_t$, $E_t$ the book value of assets, liabilities resp. the net asset value. 
	At every point in time the value of the assets is equal to the liabilities and the net asset value, i.e. $A_t = L_t + E_t$.     
    
\paragraph{Asset Model.}
    The market consists of two primary products, a riskless bond and a risky stock with price processes $B= (B_t)_t$ and $S= (S_t)_t$, respectively. We assume that $B_t = \exp (rt)$ and
    \[ S_t = \exp \left(  \left( \mu - \frac{\sigma^2}{2} 
    \right) t + \sigma W_t \right), \]
    where $W_t$ is a Wiener process. 
  At each point in time $t$, $\eta_t^B$ and $\eta_t^S$ denote the number of shares held in the bond and the stock, respectively. The resulting value of the asset portfolio is     \[ A_t = \tilde{S}_t + \tilde{B}_t  \quad \mbox{with } \; \tilde{S}_t = \eta_t^S S_t, \;  \tilde{B}_t = \eta_t^B B_t.  \]  For simplicity, we assume that $r=0$.
  
\paragraph{Investment Strategy.} We assume that at the beginning of each period a fraction $b \in [0,1]$ of the book value of the liabilities and equity is invested into the stock, while the remaining fraction $1-b$ is invested into the bond. This implies that
$$  \eta_t^S  = b \cdot \frac{A_t}{S_t}, \hspace{10mm} \eta_t^B = (1-b) \cdot A_t   $$    

\paragraph{Liability Model.}
	The insurer has a constant claims reserve $v$, such that $L_t = v$ at every point in time $t$. At the beginning of every time period the insurer takes in
	constant premiums $\pi$. 
	Insurance claims at the end of every period $t$ are assumed to follow a collective model
    \[ C_t = \sum_{k=1}^{N_t} \xi_{t,k} , \quad t=1,\dots, n,\]
    where the frequencies $N_t \in \mathbb{N}$ and the severities $\xi_{t,k} \geq 0$, $k \in \mathbb{N}$, $t=1, \dots, n$ are independent.

\paragraph{Evolution of the Net Asset Value.} At time $t=1, \dots, n$ the insurer must pay $C_{t}$ for the claims incurred in the previous period and receives premium payments $\pi$ for the next period, i.e., the amount to be reinvested into the assets equals 
    $$
   A_t  =  \eta_{t-1}^S S_{t} + \eta_{t-1}^B  B_t - C_{t} + \pi =  A_{t-1} \cdot \left( b \cdot \frac{ S_t}{S_{t-1}} +  (1-b)   \right)  - C_{t} + \pi  .
    $$

This can be rewritten in terms of the net asset value:
$$ E_t =   E_{t-1} +  (E_{t-1} +v ) \cdot b \cdot \left( \frac{ S_t}{S_{t-1}} -  1  \right)  - C_{t} + \pi  . $$
      
    \subsection{Simulation Design}
     
    To illustrate our backtesting methodology in the context of the model, we consider as in Section~\ref{sec:gvsim} the GlueV@R risk measure. We choose the parameters 
    $h_1 = 2/5$, $h_2 = 2/3$, $\alpha  = 0.05$ and $\beta = 0.01$.  Further case studies for $AV@R$ show similar results and are presented in Appendix~\ref{app:ALM}.
    Individual time periods are interpreted as days. 
    The parameters defining the evolution of the assets are set to  $\mu = \log(1.1) / 360$, $\sigma = 0.2 / \sqrt{360}$ and $b = 0.05$. 
    In order to specify the liability model, we assume that $N_t$ are iid Poisson distributed random variables with parameter $\lambda > 0$. Letting $\lambda = 7$, expectation and variance are equal to $7$. The claims $\xi_{t,k}$ are iid exponentially distributed with parameter $1 / \theta > 0$. With $\theta = 1000$, we obtain an expectation of $1000$ and a variance of $10^6$.
    We set $E_0 = 20000$. 
    The premiums per day $\pi$  equal expected claims plus a $3\%$ safety margin, i.e., $\pi = 1.03 \lambda \theta$. 
    The reserve is calculated as the expectation of the annual claims plus a $3\%$ margin, i.e., $v = 360 \cdot 1.03 \lambda \theta$. 
    
  In this simple experimental ALM case study, we consider only the Nass test which showed the best performance in the case studies in Section \ref{sec:DistStudies} (just as in \cite{KLM}) and constitutes the most promising methodology.
    The size of the Nass test is estimated as in Section~\ref{sec:DistStudies}. We consider the following alternatives labeled as NB, PAR and LOGN: 
	\begin{itemize}
		\item[(NB)]
		We replace the Poisson distributed frequencies $N_t$ by frequencies $N^\prime_t$ with a negative binomial distribution with a number of failures $r \in \mathbb{N}$  and a success probability $q \in [0,1]$.  Setting $q = 1 / ( 1 +  \lambda / r)$,  $N^\prime_t$ and $N_t$ possess the same expectation.
Letting $r = 7$, the variance of the negative binomial distribution equals $\mathsf{Var}[N^\prime_t] = \lambda ( 1 + \lambda / r) = 14$. 
		
		\item[(PAR)] Claim sizes are given by $\xi^1_{t,k} - 1$ with $\xi^1_{t,k}$ being Pareto distributed with scale $x_0 = 1$ and shape $a > 0$.
		Setting $a  = (\theta + 1) / \theta$ guarantees that the expectation of the claim sizes $\xi^1_{t,k} - 1$ of this alternative equals the expectation of the claim sizes $\xi_{t,k}$ of the null hypothesis.
		
		\item[(LOGN)] Under this alternative the claim sizes $\xi^2_{t,k}$ are log normally distributed with parameters $\mu \in \mathbb{R}$ and $\sigma > 0$. 
	 	We set $\mu = \log ( \theta) - \sigma^2 / 2$ such that $\xi^2_{t,k}$ has the same expectation as $\xi_{t,k}$, and we choose $\sigma = 1$ such that $\mathsf{Var}[\xi^2_{t,k}] = (\exp(\sigma^2) -1 ) \theta^2 \approx 1.7183 \cdot 10^6$. 
	\end{itemize}

	\begin{table}[h!]
    	\captionsetup{font=scriptsize}
    	\begin{minipage}[t]{1\textwidth}
    	        \centering
    	        \begin{tabularx}{1.0\textwidth}{p{10mm} | p{8mm} | X X X X X X X}
    	        \multicolumn{9}{c}{\textbf{Nass}}  \\
    	        \hline \\
    	        & $n | m$ & 1 & 2 & 4 & 8 & 16 & 32 & 64 \\
    	        \hline
    	        Size & 250 &\cellcolor{p}1.53 &   1.34 &   1.35 &  \cellcolor{g} 1.07 &  \cellcolor{g} 1.03  &  1.24  &  \cellcolor{g}1.06 \\
    	        &  500 & \cellcolor{p}1.57 &   1.36 &\cellcolor{vg}    1.19  &   1.21 &  \cellcolor{vg}    1.19  & \cellcolor{vg}    1.14   &\cellcolor{vg}    1.19 \\ 
    	         &  1000 &\cellcolor{p} 1.77  &  1.46  &  1.30 &   1.22  & \cellcolor{vg}   1.18 & \cellcolor{vg}    1.16  &  1.25 \\
    	         & 2000 &\cellcolor{p}1.76  &\cellcolor{p}  1.50  &  1.34  &  1.25  &  1.27  & \cellcolor{vg}   1.20   & 1.25   \\
    	         \hline
    	         NB & 250  &  69.59  & 65.27 &   64.09  &  61.46 &   59.20 &   59.08  &  60.20\\
    	         & 500 & \cellcolor{g}  89.44  &  \cellcolor{g}  89.56  &  \cellcolor{g}  84.74  &   \cellcolor{g} 84.72  &  \cellcolor{g}  82.61 &   \cellcolor{g}  79.70 &  \cellcolor{g}   82.08 \\
    	         & 1000& \cellcolor{g}   99.02  & \cellcolor{g}  98.88  &  \cellcolor{g}  98.45  &  \cellcolor{g}  98.28  &  \cellcolor{g}  96.74 &  \cellcolor{g}   96.49   &  \cellcolor{g} 95.97\\
    	         & 2000 &  \cellcolor{g} 99.99  &  \cellcolor{g} 100.00  &   \cellcolor{g} 99.99   &  \cellcolor{g} 100 &   \cellcolor{g} 99.98 &   \cellcolor{g} 99.93  &  \cellcolor{g} 99.90 \\
    	         \hline
    	         PAR & 250 & \cellcolor{g} 97.60  &\cellcolor{g} 98.31&\cellcolor{g}  98.27 & \cellcolor{g} 97.97 &\cellcolor{g}  97.75 & \cellcolor{g} 97.72 & \cellcolor{g} 97.29\\
    	         & 500 & \cellcolor{g}99.96 & \cellcolor{g} 100&\cellcolor{g} 99.93 & \cellcolor{g} 99.97  & \cellcolor{g}99.96 &  \cellcolor{g}99.97  & \cellcolor{g}99.95 \\
    	         & 1000 &\cellcolor{g} 100 &\cellcolor{g} 100 & \cellcolor{g}100 &\cellcolor{g} 100 &\cellcolor{g} 100 & \cellcolor{g}100& \cellcolor{g} 100\\
			   	& 2000 &  \cellcolor{g}100 &\cellcolor{g} 100 &\cellcolor{g}100 & \cellcolor{g}100 & \cellcolor{g}100 &\cellcolor{g} 100& \cellcolor{g} 100\\  
			   	\hline     	         
			   	LOGN & 250 & 53.11 &   51.63 &   50.85  &  47.94  &  47.83   & 46.28   & 44.05 \\
			   	& 500 &\cellcolor{g} 75.63  &\cellcolor{g} 72.63   &\cellcolor{g} 70.84  &\cellcolor{g} 70.22 &  69.41 &  69.09 & 67.51 \\
			   	& 1000 & \cellcolor{g}93.46  &\cellcolor{g} 94.03 &\cellcolor{g}  92.31 &\cellcolor{g}  92.33 & \cellcolor{g} 91.68 &\cellcolor{g}  90.70  
			   	&\cellcolor{g} 89.85\\
			   	& 2000 &\cellcolor{g} 99.67 &\cellcolor{g} 99.68 &\cellcolor{g} 99.53 &\cellcolor{g} 99.60  & \cellcolor{g}99.62 &\cellcolor{g}  99.41
			   	 &  \cellcolor{g}99.22
			   	\end{tabularx}
    	 \end{minipage}  	 
    	 \caption{Backtesting the ALM model: Estimated size and power in \% for the alternatives NB, PAR and LOGN of the Nass test. The size is 
    	 represented as the fraction of estimated size divided by the desired level $\kappa = 5\%$. Values of the size between $0.9 - 1.1$ are dark green, 
    	 between $0.8 - 1.2$ are light green, above $1.5$ are red and above $2$ dark red. For the power green refers to a power $\geq 70\%$; light red to a 
    	 power $\leq 30 \%$ and dark red indicates a power $\leq 10\%$.  }
    	 \label{tab:ALMres}
    \end{table}

      \subsection{Size and Power}
    
    We backtest the quantitative risk measurement of the net asset value process $(E_t)_{t=1, \dots, n}$, applying GlueV@R.\footnote{Formally, the argument of the risk measure is  $-E_t$, $t=1, \dots, n$, due to our sign convention.}  The numerical results are summarized in Table \ref{tab:ALMres}. 
	We used $N = 20000$ samples to estimate the size and the power of the Nass tests. 
	The parameters $n$ and $m$ determine the number of observed data and the number of cells considered. 
	As before, the desired level of the approximate test is set to $\kappa = 5\%$. 
  
	Overall the size of the Nass test in the ALM model is slightly too large as displayed in the first panel of Table~\ref{tab:ALMres} where the quotient of the estimated size and the desired level $\kappa$ is shown. We see that the size is close to the desired level $\kappa$ for large $m$.
	
	The overall power of the test is quite good. This is due to the fact that the null hypothesis and the considered alternatives are sufficiently different in all cases. The power is slightly higher for lower cell counts $m$ and also grows in the number of observed data $n$.
	For $n = 500$, the power exceeds $89 \%$ for NB, $99\%$ for PAR and $75 \%$ for LOGN.  The biggest power is estimated at $n = 2000$ exceeding $99\%$ for all alternatives. 

    \section{Conclusion}\label{sec:concl}
    
    This paper proposes a multinomial backtesting methodology for distortion risk measures that is based on a stratification and randomization of risk levels, extending the non-randomized AV@R-backtest of \cite{KLM}. The method is applicable to a wide range of risk measures -- being at the same time highly tractable. The best results are obtained for the Nass test. Numerical experiments based on artificial data demonstrate the good performance of the method if the null hypothesis and the considered alternatives are sufficiently different from each other. 
For AV@R, our randomized backtesting method improves upon the multinomial backtest of \cite{KLM}. 

Future research should study the performance of DRM backtesting methods on the basis of real statistical data. Another interesting, but challenging question would be to compute lower bounds for the power of the method in terms of the number of data points and a measure of the distance between the null hypothesis and the alternative. 
 
\onehalfspacing

\bibliographystyle{plainnat}
\bibliography{BTDRM}

\begin{thebibliography}{78}
\providecommand{\natexlab}[1]{#1}
\providecommand{\url}[1]{\texttt{#1}}
\expandafter\ifx\csname urlstyle\endcsname\relax
  \providecommand{\doi}[1]{doi: #1}\else
  \providecommand{\doi}{doi: \begingroup \urlstyle{rm}\Url}\fi

\bibitem[Acerbi(2002)]{AcerbiCarlo}
Carlo Acerbi.
\newblock Spectral measures of risk: A coherent representation of subjective
  risk aversion.
\newblock \emph{Journal of Banking \& Finance}, 26\penalty0 (7):\penalty0
  1505--1518, 2002.

\bibitem[Acerbi and Szekely(2014)]{acerbiSzekely2014}
Carlo Acerbi and Balazs Szekely.
\newblock Back-testing {E}xpected {S}hortfall.
\newblock \emph{Risk}, 27\penalty0 (11):\penalty0 76--81, 2014.

\bibitem[Artzner et~al.(1999)Artzner, Delbaen, Eber, and Heath]{Artzner1999}
Philippe Artzner, Freddy Delbaen, Jean-Marc Eber, and David Heath.
\newblock Coherent measures of risk.
\newblock \emph{Mathematical finance}, 9\penalty0 (3):\penalty0 203--228, 1999.

\bibitem[Balb{\'a}s et~al.(2009)Balb{\'a}s, Garrido, and Mayoral]{BGM}
Alejandro Balb{\'a}s, Jos{\'e} Garrido, and Silvia Mayoral.
\newblock Properties of distortion risk measures.
\newblock \emph{Methodology and Computing in Applied Probability}, 11\penalty0
  (3):\penalty0 385, 2009.

\bibitem[Bann{\"o}r and Scherer(2014)]{BS2014}
Karl~F. Bann{\"o}r and Matthias Scherer.
\newblock On the calibration of distortion risk measures to bid-ask prices.
\newblock \emph{Quantitative Finance}, 14\penalty0 (7):\penalty0 1217--1228,
  2014.

\bibitem[Becker et~al.(2014)Becker, Cottin, Fahrenwaldt, Hamm, N{\"o}rtemann,
  and Weber]{WeberHamm2014}
Torsten Becker, Claudia Cottin, Matthias Fahrenwaldt, Anna-Maria Hamm, Stefan
  N{\"o}rtemann, and Stefan Weber.
\newblock Market consistent embedded value -- eine praxisorientierte
  {E}inf{\"u}hrung.
\newblock \emph{Der Aktuar}, 1:\penalty0 4--8, 2014.

\bibitem[Belles-Sampera et~al.(2014{\natexlab{a}})Belles-Sampera, Guill{\'e}n,
  and Santolino]{BSGS}
Jaume Belles-Sampera, Montserrat Guill{\'e}n, and Miguel Santolino.
\newblock Beyond {V}alue-at-{R}isk: Glue{VaR} distortion risk measures.
\newblock \emph{Risk Analysis}, 34\penalty0 (1):\penalty0 121--134,
  2014{\natexlab{a}}.

\bibitem[Belles-Sampera et~al.(2014{\natexlab{b}})Belles-Sampera, Guill{\'e}n,
  and Santolino]{BSGS2}
Jaume Belles-Sampera, Montserrat Guill{\'e}n, and Miguel Santolino.
\newblock Glue{VaR} risk measures in capital allocation applications.
\newblock \emph{Insurance: Mathematics and Economics}, 58:\penalty0 132--137,
  2014{\natexlab{b}}.

\bibitem[Bellini and Bignozzi(2015)]{BB2015}
Fabio Bellini and Valeria Bignozzi.
\newblock On elicitable risk measures.
\newblock \emph{Quantitative Finance}, 15\penalty0 (5):\penalty0 725--733,
  2015.

\bibitem[Benhamou and Melot(2018)]{benmel2018}
Eric Benhamou and Valentin Melot.
\newblock Seven proofs of the pearson chi-squared independence test and its
  graphical interpretation.
\newblock \emph{SSRN}, 2018.

\bibitem[Berkowitz(2001)]{berkowitz2001}
Jeremy Berkowitz.
\newblock Testing density forecasts, with applications to risk management.
\newblock \emph{Journal of Business \& Economic Statistics}, 19\penalty0
  (4):\penalty0 465--474, 2001.

\bibitem[Berkowitz et~al.(2011)Berkowitz, Christoffersen, and
  Pelletier]{berkowitzChristoffersenPelletier2011}
Jeremy Berkowitz, Peter Christoffersen, and Denis Pelletier.
\newblock Evaluating {V}alue-at-{R}isk models with desk-level data.
\newblock \emph{Management Science}, 57\penalty0 (12):\penalty0 2213--2227,
  2011.

\bibitem[Bignozzi and Tsanakas(2016)]{BT2016}
Valeria Bignozzi and Andreas Tsanakas.
\newblock Parameter uncertainty and residual estimation risk.
\newblock \emph{The Journal of Risk and Insurance}, 83\penalty0 (4):\penalty0
  949--978, 2016.

\bibitem[Bollerslev(1986)]{BOLLERSLEV1986307}
Tim Bollerslev.
\newblock Generalized autoregressive conditional heteroskedasticity.
\newblock \emph{Journal of Econometrics}, 31\penalty0 (3):\penalty0 307--327,
  1986.

\bibitem[Cai and Krishnamoorthy(2006)]{CK}
Yong Cai and Kalimuthu Krishnamoorthy.
\newblock Exact size and power properties of five tests for multinomial
  proportions.
\newblock \emph{Communications in Statistics---Simulation and
  Computation{\textregistered}}, 35\penalty0 (1):\penalty0 149--160, 2006.

\bibitem[Carver(2013)]{carver2013}
Laurie Carver.
\newblock Mooted {VaR} substitute cannot be back-tested, says top quant.
\newblock \emph{Risk}, 8, 2013.

\bibitem[Carver(2014)]{carver2014}
Laurie Carver.
\newblock Back-testing {E}xpected {S}hortfall: {M}ission impossible.
\newblock \emph{Risk}, 27\penalty0 (10), 2014.

\bibitem[Casella and Berger(2002)]{CasellaBergerRoger2002}
George Casella and Roger~L. Berger.
\newblock \emph{Statistical inference}, volume~2.
\newblock Duxbury Pacific Grove, CA, 2002.

\bibitem[Chen(2014)]{chen2014}
James~Ming Chen.
\newblock Measuring market risk under the {B}asel accords: {VaR}, {S}tressed
  {VaR}, and {E}xpected {S}hortfall.
\newblock \emph{Journal of Finance}, 8:\penalty0 184--201, 2014.

\bibitem[Cherny and Madan(2009)]{ChernyMadan2008}
Alexander Cherny and Dilip Madan.
\newblock New measures for performance evaluation.
\newblock \emph{The Review of Financial Studies}, 22\penalty0 (7):\penalty0
  2571--2606, 2009.

\bibitem[{Choquet}(1954)]{Choquet}
Gustave {Choquet}.
\newblock {Theory of capacities.}
\newblock \emph{{Ann. Inst. Fourier}}, 5:\penalty0 131--295, 1954.

\bibitem[Christoffersen and Pelletier(2004)]{christoffersenpelletier2004}
Peter Christoffersen and Denis Pelletier.
\newblock Backtesting {V}alue-at-{R}isk: A duration-based approach.
\newblock \emph{Journal of Financial Econometrics}, 2\penalty0 (1):\penalty0
  84--108, 2004.

\bibitem[Christoffersen(1998)]{C}
Peter~F. Christoffersen.
\newblock Evaluating interval forecasts.
\newblock \emph{International economic review}, 39:\penalty0 841--862, 1998.

\bibitem[Cont et~al.(2010)Cont, Deguest, and Scandolo]{CDS2010}
Rama Cont, Romain Deguest, and Giacomo Scandolo.
\newblock Robustness and sensitivity analysis of risk measurement procedures.
\newblock \emph{Quantitative Finance}, 10\penalty0 (6):\penalty0 593--606,
  2010.

\bibitem[Costanzino and Curran(2015)]{costanzinocurran2015}
Nick Costanzino and Mike Curran.
\newblock Backtesting general spectral risk measures with application to
  {E}xpected {S}hortfall.
\newblock \emph{Journal of Risk Model Validation}, 9\penalty0 (1):\penalty0
  21--31, 2015.

\bibitem[Costanzino and Curran(2018)]{costanzinocurran2017}
Nick Costanzino and Mike Curran.
\newblock A simple traffic light approach to backtesting {E}xpected
  {S}hortfall.
\newblock \emph{Risks}, 6\penalty0 (1), 2018.

\bibitem[Delbaen et~al.(2016)Delbaen, Bellini, Bignozzi, and Ziegel]{DBBZ2016}
Freddy Delbaen, Fabio Bellini, Valeria Bignozzi, and Johanna Ziegel.
\newblock Risk measures with the cxls property.
\newblock \emph{Finance and Stochastics}, 20:\penalty0 433--453, 2016.

\bibitem[{Denneberg}(1994)]{Denneberg}
Dieter {Denneberg}.
\newblock \emph{{Non-additive measure and integral.}}
\newblock Dordrecht: Kluwer Academic Publishers, 1994.

\bibitem[Dhaene et~al.(2006)Dhaene, Vanduffel, Goovaerts, Kaas, Tang, and
  Vyncke]{Dhaene06}
Jan Dhaene, Steven Vanduffel, Marc~J. Goovaerts, Rob Kaas, Qihe Tang, and David
  Vyncke.
\newblock Risk measures and comonotonicity: A review.
\newblock \emph{Stochastic Models}, 22\penalty0 (4):\penalty0 573--606, 2006.

\bibitem[Dhaene et~al.(2012)Dhaene, Kukush, Linders, and Tang]{DKLT}
Jan Dhaene, Alexander Kukush, Dani{\"e}l Linders, and Qihe Tang.
\newblock Remarks on quantiles and distortion risk measures.
\newblock \emph{European Actuarial Journal}, 2\penalty0 (2):\penalty0 319--328,
  2012.

\bibitem[Diebold et~al.(1998)Diebold, Gunther, and Tay]{dieboldGuntherTay1997}
Francis~X. Diebold, Todd~A. Gunther, and Anthony Tay.
\newblock Evaluating density forecasts.
\newblock \emph{International Economic Review}, 39\penalty0 (4):\penalty0 863
  -- 883, 1998.

\bibitem[Dowd et~al.(2008)Dowd, Cotter, and Sorwar]{Dowd2008}
Kevin Dowd, John Cotter, and Ghulam Sorwar.
\newblock Spectral risk measures: Properties and limitations.
\newblock \emph{Journal of Financial Services Research}, 34:\penalty0 61--75,
  2008.

\bibitem[Du and Escanciano(2016)]{duEscanciano2016}
Zaichao Du and Juan~Carlos Escanciano.
\newblock Backtesting {E}xpected {S}hortfall: accounting for tail risk.
\newblock \emph{Management Science}, 63\penalty0 (4):\penalty0 940--958, 2016.

\bibitem[Embrechts et~al.(2018)Embrechts, Liu, and Wang]{embrechtsliuwang2016}
Paul Embrechts, Haiyan Liu, and Ruodo Wang.
\newblock Quantile-based risk sharing.
\newblock \emph{Operations Research}, 66\penalty0 (4):\penalty0 936--949, 2018.

\bibitem[Emmer et~al.(2015)Emmer, Kratz, and Tasche]{EKt}
Susanne Emmer, Marie Kratz, and Dirk Tasche.
\newblock What is the best risk measure in practice? a comparison of standard
  measures.
\newblock \emph{Journal of Risk}, 18:\penalty0 31 -- 61, 2015.

\bibitem[Fern{\'a}ndez and Steel(1998)]{FS1}
Carmen Fern{\'a}ndez and Mark F.~J. Steel.
\newblock On {B}ayesian modeling of fat tails and skewness.
\newblock \emph{Journal of the American Statistical Association}, 93\penalty0
  (441):\penalty0 359--371, 1998.

\bibitem[Fissler et~al.(2016)Fissler, Ziegel, and
  Gneiting]{fisslerZiegelGneiting2015}
Tobias Fissler, Johanna~F. Ziegel, and Tilmann Gneiting.
\newblock Expected {S}hortfall is jointly elicitable with {V}alue at {R}isk -
  implications for backtesting.
\newblock \emph{Risk}, January:\penalty0 58--61, 2016.

\bibitem[F{\"o}llmer and Schied(2002)]{FS02}
Hans F{\"o}llmer and Alexander Schied.
\newblock Convex measures of risk and trading constraints.
\newblock \emph{Finance and Stochastics}, 6\penalty0 (4):\penalty0 429 -- 447,
  2002.

\bibitem[F{\"o}llmer and Schied(2016)]{FS}
Hans F{\"o}llmer and Alexander Schied.
\newblock \emph{Stochastic finance: an introduction in discrete time}.
\newblock Walter de Gruyter, 4th edition, 2016.

\bibitem[F\"ollmer and Weber(2015)]{FW15}
Hans F\"ollmer and Stefan Weber.
\newblock The axiomatic approach to risk measurement for capital determination.
\newblock \emph{Annual Review of Financial Economics}, 7:\penalty0 301--337,
  2015.

\bibitem[Frittelli and Gianin(2002)]{frittelli2002}
Marco Frittelli and Emanuela~Rosazza Gianin.
\newblock Putting order in risk measures.
\newblock \emph{Journal of Banking \& Finance}, 26\penalty0 (7):\penalty0
  1473--1486, 2002.

\bibitem[Gneiting(2011)]{gneiting2011}
Tilmann Gneiting.
\newblock Making and evaluating point forecasts.
\newblock \emph{Journal of the American Statistical Association}, 106\penalty0
  (494):\penalty0 746--762, 2011.

\bibitem[Goovaerts and Laeven(2008)]{GOOVAERTS2008}
Marc~J. Goovaerts and Roger~J.A. Laeven.
\newblock Actuarial risk measures for financial derivative pricing.
\newblock \emph{Insurance: Mathematics and Economics}, 42\penalty0
  (2):\penalty0 540--547, 2008.
\newblock \doi{https://doi.org/10.1016/j.insmatheco.2007.04.001}.

\bibitem[Gordy and McNeil(2020)]{gordyMcneil2017}
Michael~B. Gordy and Alexander~J. McNeil.
\newblock Spectral backtests of forecast distributions with application to risk
  management.
\newblock \emph{Journal of Banking \& Finance}, 116, 2020.

\bibitem[Greco(1982)]{Greco}
Gabriele~H. Greco.
\newblock Sulla rappresentazione di funzionali mediante integrali.
\newblock \emph{Rendiconti del Seminario Matematico della Universita di
  Padova}, 66:\penalty0 21--42, 1982.

\bibitem[Guegan and Hassani(2015)]{GH2015}
Dominique Guegan and Bertrand Hassani.
\newblock Distortion risk measure or the transformation of unimodal
  distributions into multimodal functions.
\newblock In \emph{Future perspectives in risk models and finance}, pages
  71--88. Springer, 2015.

\bibitem[Guillen et~al.(2018)Guillen, Sarabia, Belles-Sampera, and
  Prieto]{GSBP2018}
Montserrat Guillen, Jos{\'e} Sarabia, Jaume Belles-Sampera, and Faustino
  Prieto.
\newblock Distortion risk measures for nonnegative multivariate risks.
\newblock \emph{Journal of Operational Risk}, 13:\penalty0 35--57, 06 2018.

\bibitem[Hamm et~al.(2020)Hamm, Knispel, and Weber]{HammKnispelWeber2019}
Anna-Maria Hamm, Thomas Knispel, and Stefan Weber.
\newblock Optimal risk sharing in insurance networks.
\newblock \emph{European Actuarial Journal}, 10\penalty0 (1):\penalty0
  203--234, 2020.

\bibitem[Kerkhof et~al.(2003)Kerkhof, Melenberg, Schumacher,
  et~al.]{kerkhofMelenbergSchumacher2003}
Jeroen Kerkhof, Bertrand Melenberg, Hans Schumacher, et~al.
\newblock \emph{Testing {E}xpected {S}hortfall models for derivative
  positions}.
\newblock Tilburg University, 2003.

\bibitem[Kim and Weber(2021)]{KIM2021}
Sojung Kim and Stefan Weber.
\newblock Simulation methods for robust risk assessment and the distorted mix
  approach.
\newblock \emph{European Journal of Operational Research}, 2021.

\bibitem[Kratz et~al.(2018)Kratz, Lok, and McNeil]{KLM}
Marie Kratz, Yen~H. Lok, and Alexander~J. McNeil.
\newblock Multinomial {VaR} backtests: A simple implicit approach to
  backtesting {E}xpected {S}hortfall.
\newblock \emph{Journal of Banking \& Finance}, 88:\penalty0 393--407, 2018.

\bibitem[Kupiec(1995)]{kupiec1995}
Paul Kupiec.
\newblock Techniques for verifying the accuracy of risk measurement models.
\newblock \emph{The Journal of Derivatives}, 3\penalty0 (2):\penalty0 73--84,
  1995.

\bibitem[Kusuoka(2001)]{Kusuoka2001}
Shigeo Kusuoka.
\newblock On law invariant coherent risk measures.
\newblock \emph{Advances in Mathematical Economics}, pages 83--95, 2001.

\bibitem[Labuschagne and Offwood(2010)]{LABUSCHAGNE2010}
Coenraad~C.A. Labuschagne and Theresa~M. Offwood.
\newblock A note on the connection between the esscher–girsanov transform and
  the wang transform.
\newblock \emph{Insurance: Mathematics and Economics}, 47\penalty0
  (3):\penalty0 385--390, 2010.
\newblock ISSN 0167-6687.
\newblock \doi{https://doi.org/10.1016/j.insmatheco.2010.08.004}.

\bibitem[Li et~al.(2018)Li, Shao, Wang, and Yang]{LHRJ2018}
Lujun Li, Hui Shao, Ruodu Wang, and Jingping Yang.
\newblock Worst-case range value-at-risk with partial information.
\newblock \emph{SIAM Journal on Financial Mathematics}, 9\penalty0
  (1):\penalty0 190--218, 2018.

\bibitem[Lindner(2009)]{Lindner2009}
Alexander~M. Lindner.
\newblock Stationarity, mixing, distributional properties and moments of
  $\mbox{GARCH}(p, q)$--processes.
\newblock In Thomas Mikosch, Jens-Peter Krei{\ss}, Richard~A. Davis, and
  Torben~Gustav Andersen, editors, \emph{Handbook of Financial Time Series},
  pages 43--69. Springer Berlin Heidelberg, Berlin, Heidelberg, 2009.

\bibitem[L{\"o}ser et~al.(2019)L{\"o}ser, Wied, and
  Ziggel]{loserWiedZiggel2018}
Robert L{\"o}ser, Dominik Wied, and Daniel Ziggel.
\newblock New backtests for unconditional coverage of expected shortfall.
\newblock \emph{Journal of Risk}, 21\penalty0 (4), 2019.

\bibitem[{Lynn Wirch} and Hardy(1999)]{WirchHardy1999}
Julia {Lynn Wirch} and Mary~R. Hardy.
\newblock A synthesis of risk measures for capital adequacy.
\newblock \emph{Insurance: Mathematics and Economics}, 25\penalty0
  (3):\penalty0 337--347, 1999.

\bibitem[McNeil et~al.(2015)McNeil, Frey, and Embrechts]{MFE2015}
Alexander~J. McNeil, R{\"u}diger Frey, and Paul Embrechts.
\newblock \emph{Quantitative Risk Management}.
\newblock Princeton University Press, 2 edition, 2015.

\bibitem[Methni and Stupfler(2017)]{MS2017}
Jonathan~El Methni and Gilles Stupfler.
\newblock Extreme versions of wang risk measures and their estimation for
  heavy-tailed distributions.
\newblock \emph{Statistica Sinica}, 27\penalty0 (2):\penalty0 907--930, 2017.

\bibitem[Nass(1959)]{N}
Charles Albert~Gerard Nass.
\newblock The $\chi^2$ test for small expectations in contingency tables, with
  special reference to accidents and absenteeism.
\newblock \emph{Biometrika}, 46\penalty0 (3/4):\penalty0 365--385, 1959.

\bibitem[Pearson(1900)]{P}
Karl Pearson.
\newblock On the criterion that a given system of deviations from the probable
  in the case of a correlated system of variables is such that it can be
  reasonably supposed to have arisen from random sampling.
\newblock \emph{The London, Edinburgh, and Dublin Philosophical Magazine and
  Journal of Science}, 50\penalty0 (302):\penalty0 157--175, 1900.

\bibitem[Robert(2014)]{R}
Christian~P. Robert.
\newblock First moments of the truncated and absolute {S}tudent's variates.
\newblock \emph{arXiv: Statistics Theory}, 2014.

\bibitem[Samanthi and Sepanski(2019)]{SS2019}
Ranadeera~G.M. Samanthi and Jungsywan Sepanski.
\newblock Methods for generating coherent distortion risk measures.
\newblock \emph{Annals of Actuarial Science}, 13\penalty0 (2):\penalty0
  400--416, 2019.

\bibitem[{Schmeidler}(1986)]{Schmeidler}
David {Schmeidler}.
\newblock {Integral representation without additivity.}
\newblock \emph{{Proc. Am. Math. Soc.}}, 97:\penalty0 255--261, 1986.

\bibitem[Song and Yan(2006)]{Song2006}
Yongsheng Song and Jia-An Yan.
\newblock The representations of two types of functionals on
  \mbox{$L^\infty(\Omega, \mathcal{F})$} and \mbox{$L^\infty(\Omega,
  \mathcal{F}, P)$}.
\newblock \emph{Science in China Series A: Mathematics}, 49\penalty0
  (10):\penalty0 1376--1382, 2006.

\bibitem[Song and Yan(2009{\natexlab{a}})]{Song2009}
Yongsheng Song and Jia-An Yan.
\newblock {An overview of representation theorems for static risk measures.}
\newblock \emph{{Sci. China, Ser. A}}, 52\penalty0 (7):\penalty0 1412--1422,
  2009{\natexlab{a}}.

\bibitem[Song and Yan(2009{\natexlab{b}})]{SongYan2009}
Yongsheng Song and Jia-An Yan.
\newblock Risk measures with comonotonic subadditivity or convexity and
  respecting stochastic orders.
\newblock \emph{Insurance: Mathematics and Economics}, 45\penalty0
  (3):\penalty0 459--465, 2009{\natexlab{b}}.

\bibitem[Wang(1995)]{W1}
Shaun~S. Wang.
\newblock Insurance pricing and increased limits ratemaking by proportional
  hazard transforms.
\newblock \emph{Insurance: Mathematics and Economics}, 17\penalty0
  (1):\penalty0 43 -- 54, 1995.

\bibitem[Wang(1996)]{W2}
Shaun~S. Wang.
\newblock Premium calculation by transforming the layer premium density.
\newblock \emph{ASTIN Bulletin}, 26\penalty0 (1):\penalty0 71 -- 92, 1996.

\bibitem[Wang(2000)]{Wang2000}
Shaun~S. Wang.
\newblock A class of distortion operators for pricing financial and insurance
  risks.
\newblock \emph{The Journal of Risk and Insurance}, 67\penalty0 (1):\penalty0
  15--36, 2000.

\bibitem[Wang(2001)]{Wang2001}
Shaun~S Wang.
\newblock A risk measure that goes beyond coherence.
\newblock \emph{12th AFIR International Colloquium}, 2001.

\bibitem[Weber(2006)]{weber2006}
Stefan Weber.
\newblock Distribution-invariant risk measures, information, and dynamic
  consistency.
\newblock \emph{Mathematical Finance}, 16\penalty0 (2):\penalty0 419--441,
  2006.

\bibitem[Weber(2018)]{We}
Stefan Weber.
\newblock Solvency {II}, or how to sweep the downside risk under the carpet.
\newblock \emph{Insurance: Mathematics and Economics}, 82:\penalty0 191--200,
  2018.

\bibitem[Wong(2010)]{wong2010}
Woon~K. Wong.
\newblock Backtesting {V}alue-at-{R}isk based on tail losses.
\newblock \emph{Journal of Empirical Finance}, 17\penalty0 (3):\penalty0
  526--538, 2010.

\bibitem[Wozabal(2014)]{David2014}
David Wozabal.
\newblock Robustifying convex risk measures for linear portfolios: A
  nonparametric approach.
\newblock \emph{Operations Research}, 62\penalty0 (6):\penalty0 1302--1315,
  2014.

\bibitem[Yitzhaki(1982)]{Yit1982}
Shlomo Yitzhaki.
\newblock Stochastic dominance, mean variance, and gini's mean difference.
\newblock \emph{The American Economic Review}, 72\penalty0 (1):\penalty0
  178--185, 1982.

\bibitem[Ziggel et~al.(2014)Ziggel, Berens, Wei{\ss}, and Wied]{ziggel2014}
Daniel Ziggel, Tobias Berens, Gregor N.~F. Wei{\ss}, and Dominik Wied.
\newblock A new set of improved {V}alue-at-{R}isk backtests.
\newblock \emph{Journal of Banking \& Finance}, 48:\penalty0 29--41, 2014.

\end{thebibliography}

\newpage    
\appendix
\doublespacing    
    
    \section{Online Appendix}\label{sec:app}
    
    \subsection{Comonotonic Risk Measures and DRMs} \label{app:DRM}

        \begin{defi}    
	        \begin{itemize}
	            \item[(i)] A non decreasing function $g: [0,1] \rightarrow [0,1]$, with $g(0) = 0$ and $g(1) = 1$, is called a distortion function.
	            \item[(ii)] Let $\mathsf{P}$ be a probability measure on $(\Omega, \mathcal{F})$ and $g$ be a distortion function. The monetary risk measure $\rho_g : \mathcal{X} \rightarrow \mathbb{R}$ defined as 
	            \[ \rho_g (X) := \int_{-\infty}^0 \left[ g( \mathsf{P}( \{ X > x \} ) ) - 1 \right] dx + 
	            \int_0^\infty g( \mathsf{P}( \{ X > x \})) dx, \]
	            is called a DRM with respect to $g$.
	        \end{itemize}
	            
        \end{defi}

   DRMs can be expressed as mixtures of the quantile functions, if the distortion function is either left- or right-continuous. This is described in \cite{DKLT}.
	    \begin{theo}\label{theo:ContDRM}
		(i)	If the distortion function $g$ is right-continuous, the DRM $\rho_g(X)$ is represented by a
			Lebesgue-Stieltje integral:
			\[ \rho_g(X) = \int_{[0,1]} q_X^{+} (1-u) \ dg(u) \]
			where $q_X^{+}(u) := \sup \{ x \vert F_X(x) \leq u \}$. \\
		(ii) If the distortion function $g$ is left-continuous, the DRM $\rho_g(X)$ can be written as:
			\[ \rho_g(X) = \int_{[0,1]} q_X (1-u) \ dg(u)    =   \int_{[0,1]} q_X (u) \ d \bar g(u)\]
			where $q_X  (u) = \inf\{x \vert F_X(x) \geq u \}$ and $\bar g(q) = 1- g(1-q)$, $0 \leq q \leq 1$.
		\end{theo}

Several well known risk measures can be expressed as DRMs, see for example \cite{ChernyMadan2008}, \cite{BGM}, \cite{FS}, and \cite{We}.  We consider some important examples in our applications.
		
		\begin{exam}\label{ex:drm}
		\begin{itemize}
		\item[(i)] Choosing the distortion function  
		$g(u) = \mathds{1}_{\{ \alpha < u \leq 1 \}}$
		yields the Value at Risk at level $\alpha \in (0,1)$:
		\[ \rho_g(X) = V@R_\alpha(X) := \inf \{x \vert F_X(x) 
			\geq 1- \alpha\} .\]
	    \item[(ii)] The Average Value at Risk at level $\alpha$ corresponds to  a DRM with 
		distortion function $g(u) = \frac{u}{\alpha} \mathds{1}_{\{ 0 \leq u \leq \alpha\}}
		+ \mathds{1}_{\{ u > \alpha \}}$:
		\[ \rho_g(X) = AV@R_\alpha(X) = \frac{1}{\alpha} \int_{0}^\alpha V@R_\lambda(X) d\lambda.\]
		\item[(iii)] The GlueV@R is a DRM\footnote{This example was suggested by \cite{BSGS} and \cite{BSGS2}.} with distortion function 
		\[ g(u) = \begin{cases} \frac{h_1}{ \beta} u \hspace{30mm} & 
		\text{ if } 0 \leq u \leq \beta \\
		h_1 + \frac{h_2 - h_1}{ \alpha - \beta} ( u-  \beta) 
		& \text{ if }  \beta < u \leq \alpha \\
		1 & \text{ if }  \alpha < u \leq 1 
		\end{cases}, \]
		where $0  \leq \beta < \alpha\leq 1$ are levels and $0 \leq  h_1 \leq h_2 \leq 1$ describe the corresponding distorted probabilities. 
		The distortion function of the GlueV@R is a piecewise combination of the distortion functions of $V@R$ and $AV@R$.  The GlueV@R can be expressed as a linear combination of these risk measures, i.e.,
		\[ GlueV@R_{\beta, \alpha}^{h_1, h_2} (X) = \rho_g(X)
		= w_1 AV@R_\beta(X) + w_2 AV@R_\alpha(X) + w_3 V@R_\alpha(X), \]
		with $w_1 = h_1 - \frac{(h_2 - h_1)\beta}{ \alpha - \beta}$,
		$w_2 = \frac{h_2 - h_1 }{\alpha-\beta } \alpha$ and
		$w_3 = 1 - w_1 - w_2 = 1- h_2$.
		\end{itemize}
		\end{exam}

		\cite{DKLT} show that every distortion function 
		can be written as convex combination of a left- and right-continuous
		function. This is an important observation: If a distortion function $g$  is a convex combination of distortion functions $g_1$ and $g_2$, the distortion risk $\rho_g$  is a convex combination of the distortion risk measures $\rho_{g_1}$  and $\rho_{g_2}$, i.e., if $c_1, c_2 \geq 0$, $c_1 + c_2 = 1$, then
		$g = c_1 g_1 + c_2 g_2$  implies that $\rho_{g} = c_1 \rho_{g_1} + c_2 \rho_{g_2} . $

		\begin{theo}\label{theo:convexcomposition}
			Let $g$ be a distortion function. Then there exist right- and left-continuous distortion functions $h_r$, $h_l$ such that
			$g(u)= d_r h_r(u) + d_l h_l(u) \;
			\forall u \in [0,1]$
			with $d_r$, $d_l \in [0,1]$, $d_r + d_l = 1$. 
		\end{theo}
	    The decomposition of the distortion function is not unique unless $g$ is a step function. A unique decomposition is provided in Theorem \ref{theo:uniqueconvexdecomposition}.

The link between comonotonic risk measures, the Choquet integral and DRMs is discussed in detail in  Chapter 4 of \cite{FS} and \cite{SongYan2009}.  Comonotonic risk measures with an absolutely continuous capacity with respect to the underlying probability measure take the form of a DRM.    
        
    \begin{table}
 	\small
 	\renewcommand{\arraystretch}{1.4}
	\begin{tabularx}{1.0\textwidth}{|l|c|X|l|}
				\hline 
				Name & Distortion & Closed form &  Reference \\
				\hline 
				\hline
				$\mathrm{MINV@R}$ & $1 - (1 - u)^n$ & \scriptsize $- \mathsf{E}[\min \{ - X_1, \dots, - X_n \}]$ &  \scriptsize \cite{ChernyMadan2008} \\
				&& \scriptsize $= \mathsf{E}[ \max \{ X_1, \dots, X_n\}]$ & \scriptsize\cite{FS} \\
				& & & \scriptsize \cite{BS2014} \\
				\hline
				$\mathrm{MAXV@R} $ & $u^{ 1 / n}$ & \scriptsize  $- \mathsf{E}[ Y_1 ]$ &\scriptsize \cite{ChernyMadan2008} \\
				&&\scriptsize \mbox{such that~}  & \scriptsize \cite{FS} \\
				& & \scriptsize $\max \{Y_1, \dots, Y_n\} \sim - X$ & \scriptsize \cite{BS2014} \\
				\hline 
				$\mathrm{MINMAXV@R}$ & $1 - ( 1- u^{1 / n})^n$ 
				& \scriptsize $- \mathsf{E}[ \min \{ Y_1, \dots, Y_n\}]
				$  & \scriptsize \cite{ChernyMadan2008} \\
				& & \scriptsize \mbox{such that}& \scriptsize \cite{FS} \\
				& & \scriptsize$\max\{Y_1, \dots, Y_n \} \sim - X$& \scriptsize \cite{BS2014} \\
				\hline
				$\mathrm{MAXMINV@R}$ &  $(1 - (1 - u)^n)^{1 / n}$ &\scriptsize $- \mathsf{E} [Y_1]$ & 
				\scriptsize \cite{ChernyMadan2008} \\
				& & \scriptsize\mbox{such that } & \scriptsize \cite{FS} \\
				& & \scriptsize $\max\{Y_1, \dots, Y_n\} $ & \scriptsize \cite{BS2014} \\
				& & \scriptsize $\sim \min \{ - X_1, \dots, -X_n\} $ & \\
				\hline
				$RV@R$ & $\frac{u - \beta}{\alpha - \beta} \mathds{1}_{\{\beta < u \leq \alpha \}} 
				+ \mathds{1}_{\{ u > \alpha \}}$ & $ \frac{1}{\alpha - \beta} \int_\beta^{\alpha} V@R_\lambda(X) d\lambda$ & \scriptsize
				\cite{BT2016} \\
				(Range $V@R$) &$0 < \beta < \alpha < 1$  & &  \scriptsize \cite{We}, \cite{LHRJ2018} \\
				\hline
				\footnotesize Proportional   &  $u^{1 / \gamma}$ & $\int_0^\infty (1 - F_X(x))^{1 / \gamma} dx,$ &
				\scriptsize \cite{W1,W2}  \\
				\footnotesize hazard transform &  $\gamma > 1$ &if $X \geq 0$ a.s. &  \scriptsize \cite{GSBP2018}\\
				\hline
				Dual power & $1 - (1 - u)^\gamma$ & $\int_0^\infty 1 - F_X(x)^\gamma dx$, & \scriptsize
				\cite{WirchHardy1999} \\
				transform & $\gamma > 1$ &if $X \geq 0$ a.s. & \scriptsize \cite{GSBP2018}\\
				\hline 
				Gini's principle & $(1-\theta) u + \theta u^2$ & $\mathsf{E}[X] + \frac{\theta}{2} \mathsf{E}[|X-X_1|]$ & \scriptsize \cite{Yit1982},\cite{David2014}\\
				& $0 < \theta < 1$ & & \scriptsize \cite{GSBP2018}\\
				\hline
				Exponential & $\frac{1 - \exp(-ru)}{1-\exp(-r)}$ if $r>0$ &\centering - & \scriptsize \cite{MS2017}\\
				transform & $u$ if $r=0$ & & \scriptsize \cite{Dowd2008}\\
				\hline
				Inverse S-shaped & $a \left[ \frac{u^3}{6} - \frac{\delta u^2}{2} + \left(\frac{\delta^2}{2}+\beta \right) u\right]$& & \scriptsize \cite{GH2015}\\
				polynomial & $a = \left( \frac{1}{6} - \frac{\delta}{2} + \frac{\delta^2}{2} + \beta \right)^{-1}$& \centering - & \scriptsize \cite{MS2017}\\
				of degree 3 & $0 < \delta < 1, \beta \in \mathbb{R}$& & \\
				\hline
				Beta family &  $\int_0^u  \frac{t^{a-1}(1-t)^{b-1}}{B(a,b)}  dt$ & \centering - & \scriptsize \cite{SS2019} \\
				& $a,b >0$ & &\scriptsize \cite{WirchHardy1999}  \\
				\hline
				Wang transform & \scriptsize $\Phi(\Phi^{-1}(u) - \Phi^{-1}(q) )$ & \centering - & \scriptsize \cite{Wang2000, Wang2001}\\
				& $0 < q < 1$ & & \scriptsize \cite{David2014}\\
				\hline
	\end{tabularx}
	\caption{Further examples of distortion risk measures of a random variable $X$.  Table 1 of the online appendix of \cite{MS2017} also provides these examples of distortion functions; we also include this table of examples as a convenient reference for the reader. In the third column, $X_1, \dots, X_n$ denote independent copies of $X$, $n \in \mathbb{N}$;  $Y_1, \dots, Y_n$ are suitable i.i.d. random variables that satisfy the conditions that are specified in the third column of the table. $B$ denotes the beta function, $\Phi, \Phi^{-1}$ the distribution and quantile function of the standard normal distribution respectively. Observe that the Wang transform is closely related to Esscher and Girsanov transforms, as shown in \cite{GOOVAERTS2008} and \cite{LABUSCHAGNE2010}; this observation is particularly relevant in the context of arbitrage-free pricing.}
	\label{tab:DRMExam}
\end{table}

        \subsection{Examples of DRMs}\label{ex_DRM}
        
Examples of DRMs are provided in Table~\ref{tab:DRMExam} which is taken from \cite{MS2017}.

    \subsection{Proofs}
    \subsubsection{Appendix to Section \ref{sec:ch2}}\label{app:ch2}
    \begin{proof}[Proof of Theorem \ref{theo:uniqueconvexdecomposition}]
        As any distortion function $g$ is non decreasing, $g$ has at most countably many discontinuities. We define the sets of jumps by setting
		    \[ R := \{u \in [0,1] \vert g(u) - g(u-) > 0 \},
		    \hspace{5mm} L := \{u \in [0,1] \vert g(u+) - g(u) > 0 \} \]
	    with corresponding jump heights
	        \[
	             g_+ (u) = g(u) - g(u-), \hspace{10mm} g_-(u) = g(u+) - g(u), 
	        \]
respectively. Observe that $R$ and $L$ may possess a non empty intersection of points for which $g$ is neither right- nor left-continuous. We set
	     $$
	            g_{sr}(u) := a \sum_{r \in R} g_+(r) \mathds{1}_{\{r \leq u \}}, \hspace{10mm}
	            g_{sl}(u) := b \sum_{l \in L } g_-(l)\mathds{1}_{\{ l < u \}}, 
	        $$
	    where $a,b$ are chosen such that $g_{sr}, g_{sl}$ become distortion functions;
	    for this purpose, we scale $g_{sr}$ and $g_{sl}$ such that $g_{sr}(1) = g_{sl}(1) = 1$ by setting 
	        \[ a = \frac{1}{\sum_{r \in R} g_+(r)}, \hspace{10mm}
	        b = \frac{1}{\sum_{l \in L} g_-(l)}. \]
	    If $R = \emptyset$ and/or $L = \emptyset$ the original distortion function $g$ is right- or left-continuous, or continuous; in this case, we choose $a = 0$ resp. $b = 0$, and consider only the remaining parts. The functions $g_{sr}$ and $g_{sl}$ are right- resp. left-continuous step distortion functions.  The continuous part of the decomposition is obtained by setting
	        \[ g_c (u) := c \left(g(u) - \sum_{r \in R } g_+ (r) \mathds{1}_{\{ r \leq
	        u \}} - \sum_{l \in L } g_-(l) \mathds{1}_{\{ l < u \}} \right) , \hspace{3mm} c:= \frac{1}{1- \sum_{r \in R} g_+(r) - \sum_{l \in L } g_-(l)}. \]
	    Finally, we obtain that
	        $$
                g(u) \; =  \; \sum_{r \in R} g_+(r) \mathds{1}_{\{
                r \leq u \}} + \sum_{l \in L} g_-(l) \mathds{1}_{\{ l < u\}}
                +\frac{1}{c} g_c(u) 
                \;  = \; \frac{1}{a} g_{sr}(u) + \frac{1}{b} g_{sl} (u) + \frac{1}{c} g_c(u),
	      $$
	    where 
	        \[ \frac{1}{a} + \frac{1}{b} + \frac{1}{c} \; = \;
	        1- \sum_{r \in R} g_+(r) - \sum_{l \in L } g_-(l) +\sum_{r \in R} g_+(r) + \sum_{l \in L } g_-(l)  \; = \; 1. \]
        This decomposition is unique, since the functions $g_{rl}$ and $g_{sl}$ are unique for every distortion function.
    \end{proof}
    
    \subsubsection{Appendix to  Section \ref{sec:right-conti}}\label{app:right-conti}
    
    \begin{proof}[Proof of Lemma \ref{lem:RightContException}]
         Since the conditional cumulative distribution functions $F_{M_t | M_{t-1}, \dots, M_1 }$ are continuous, we have for any $t$ that
   \begin{equation}\label{eq:u=}
   \mathsf{P} (M_t>q_{M_t}^{t-1}   (1-u) | M_{t-1}, \dots , M_1  ) = u .
   \end{equation}
      Hence, 
      $$   \mathsf{P}   (L_t>q_{M_t}^{t-1}   (1-u) )  \; \overset{L_t \overset{d}{=} M_t }{=}    \;  \mathsf{P}   (M_t>q_{M_t}^{t-1}   (1-u) )   \; =   \;\mathsf{E} \left[  {\mathsf{P} (M_t>q_{M_t}^{t-1}   (1-u) | M_{t-1}, \dots , M_1  )}   \right]   \; \overset{\eqref{eq:u=}}{=}  \; u .
        $$
     Thus, $  \mathsf{E} [\mathds{1}_{t,j} | G_{t,j}]      = G_{t,j}  $ which implies
     $  \mathsf{E} [\mathds{1}_{t,j}]    \; =  \;     \mathsf{E}  [G_{t,j} ] \; =  \;    \frac{ \mathsf{E}  \left[G  \mathds{1}_{\{    G\in[\alpha_{j-1}, \alpha_j )   \} }   \right]}{ g(\alpha_j)   - g(\alpha_{j-1})  }   . $
We define     
  	\[ \mathds{\tilde 1}_{t,j} = 
		\begin{cases}
				1 \hspace{15mm}	& \text{ if } M_t > q^{t-1}_{M_t} (1- G_{t,j}) \\
				0 & else.
			\end{cases}, 
		\]
and note that the processes $(\mathds{ 1}_{t,j})_t$ and $(\mathds{\tilde 1}_{t,j})_t$ possess the same law, since, first, $M$ and $L$ do and, second, $(G_{t,j})_t$ is independent of $M$ and $L$. We set $\gcal_t = \sigma\{  G_{s, j}: s\leq t, j = 1, \dots, m+1\}$. We observe that
\begin{equation}\label{eq:compuxy}
\mathsf{E}  [   \mathds{\tilde 1}_{t,j} |    M_{t-1}, \dots, M_1, \gcal_t ]     \; =   \;     \mathsf{E}    [ \mathds{\tilde 1}_{t,j}    |    M_{t-1}, \dots, M_1, G_{t,j}     ]\; \overset{\eqref{eq:u=}}{=}   \;      G_{t,j}    , 
\end{equation}
where in the second step we use that $\mathds{\tilde 1}_{t,j}$ is independent of $G_{t-1, j} , \dots, G_{1,j} $ and $G_{t, k} , \dots, G_{1,k}$ for $k \neq j$. 

Next, for fixed $j =1,2, \dots, m+1$, we prove that the indicators $(\mathds{\tilde 1}_{t,j} )_t$ are independent. It suffices to show that for any $T \subseteq \{1,2, \dots, n \}$ we have that
\begin{equation}\label{eq:ind_proof}
\mathsf{E} \left[ \prod_{t \in T }     \mathds{ 1}_{t,j_t}  \right]  \; =   \;  \prod_{t \in T }      \mathsf{E} [   \mathds{ 1}_{t,j_t}   ]
\end{equation}
with $j_t \in \{ 1,2, \dots, m+1\}$ for $t \in T$.

This can be shown by induction. Suppose that \eqref{eq:ind_proof} holds for any $T$ such that $s \in T$ implies $s< t$. We set $j:= j_t$. Then
\begin{align*}
\mathsf{E}   \left[   \mathds{ 1}_{t,j}   \prod_{s \in T }     \mathds{ 1}_{s,j_s}  \right] & \; = \;  \mathsf{E}   \left[ \mathds{ \tilde1}_{t,j}   \prod_{s \in T }     \mathds{ \tilde 1}_{s,j_s}  \right] \; = \; 
 \mathsf{E}   \left[  \mathsf{E}   \left[ \mathds{ \tilde1}_{t,j}   \prod_{s \in T }     \mathds{ \tilde 1}_{s,j_s} 
 |    M_{t-1}, \dots, M_1, \gcal_t   \right] \right] \\ & \; = \;  
 \mathsf{E}   \left[   \prod_{s \in T }     \mathds{ \tilde 1}_{s,j_s}  \mathsf{E}   \left[ \mathds{ \tilde1}_{t,j}  
 |    M_{t-1}, \dots, M_1, G_{t,j}   \right] \right] \quad  \hfill   \mbox{\footnotesize(due to measurability and independence)} \\
 & \; \overset{\eqref{eq:compuxy}}{=} \;   \mathsf{E}   \left[      \prod_{s \in T }     \mathds{ \tilde 1}_{s,j_s}    G_{t,j}   \right] \\
 & \; = \;   \mathsf{E}   \left[     \prod_{s \in T }      \mathds{ \tilde 1}_{s,j_s}  \right] \cdot     \mathsf{E}   \left[   G_{t,j}   \right]    \quad \quad \quad \quad  \quad \quad \quad \quad  \mbox{\footnotesize(due to independence)} \\
 &\; = \;    \mathsf{E}   \left[    \mathds{ 1}_{t,j}     \right] \cdot   \mathsf{E}   \left[      \prod_{s \in T }      \mathds{  1}_{s,j_s}  \right] 
\end{align*}
          
    \end{proof}
    \hspace{1em}
    
    \begin{proof}[Proof of Lemma \ref{lem:RightContBreachedLevels}]
			The random variable $X_t$ is a sum of the exception indicators $\mathds{1}_{t,j}$, $j=1, \dots, m+1$, for fixed $t=1, \dots, n$. Thus, the independence of $(X_t)_{t=1, \dots,n}$ follows from the assumption of the independence of the vectors $( \mathds{1}_{t,j} )_j$, $t=1, \dots,n$.
			
			For $1 \le k \le m$ we compute
			\begin{align*}
				\mathsf{P}(X_t = k) &= \mathsf{P} \left( L_t > q^{t-1}_{M_t} (1- G_{m+2-k}),
				L_t \leq q^{t-1}_{M_t} (1- G_{m+1-k}) \right) \\
				&= \mathsf{P} \left(q^{t-1}_{M_t} (1- G_{m+2-k}) < L_t \leq  
				q^{t-1}_{M_t} (1- G_{m+1-k}) \right) \\
				&= \mathsf{P} \left( L_t \leq q^{t-1}_{M_t} (1- G_{m+1-k})\right)
				- \mathsf{P}\left(L_t \leq q^{t-1}_{M_t} (1- G_{m+2-k}) \right) \\
				&= \frac{ \mathsf{E}\left[ G \mathds{1}_{\{G \in [\alpha_{m+1-k},
				\alpha_{m+2-k} )\}} \right]}{g(\alpha_{m+2-k}) - g(\alpha_{m+1-k})}
				- \frac{\mathsf{E}\left[ G \mathds{1}_{\{G \in [\alpha_{m-k},
				\alpha_{m+1-k} )\}}\right]}{g(\alpha_{m+1-k}) - g(\alpha_{m-k})}.
			\end{align*}
			For $k = 0$ we have that
			\[ \mathsf{P}(X_t = 0) = P ( L_t \leq q^{t-1}_{M_t} (1- G_{m+1}) )
			= 1 - \frac{\mathsf{E} \left[ G \mathds{1}_{\{ G \in [\alpha_m, 1) \}} \right]}{
			g(1) - g(\alpha_m)}. \]
			Therefore, $
				\mathsf{P} (X_t \leq k) = \sum_{i=0}^k  \mathsf{P}(X_t = i) 
				= 1 - \frac{\mathsf{E} \left[ G \mathds{1}_{\{ G \in [
				\alpha_{m-k}, \alpha_{m-k+1}) \}} \right] }{g(
				\alpha_{m-k+1}) - g(\alpha_{m-k})}.
			$
    \end{proof}
    \hspace{1em}
    
    \begin{proof}[Proof of Theorem \ref{theo:RightContCellCount}]
    Let $n_0, n_1, \dots, n_{m+1} \in \mathbb{N}$ such that $\sum n_i = n$.
	Then we have that
	    \begin{align*}
			 	&\mathsf{P} ( O_0 =  n_0, O_1 = n_1, \dots, O_{m+1} = n_{m+1} )  \\
			 	&= \mathsf{P} \left( \sum_{t=1}^n   \mathds{1}_{\{X_t = 0\}} = n_0, \sum_{t=1}^n \mathds{1}_{\{X_t = 1\}} = n_1, \dots,\sum_{t=1}^n \mathds{1}_{\{X_t = m+1\}} = n_{m+1} \right) \\
			 	&=\sum_{\pi \in \Pi} \mathsf{P} \big( X_{\pi(0)} = 0, \dots ,X_{\pi(n_0)} = 0, X_{\pi(n_0+1)} = 1, \dots, X_{\pi(n_0+n_1)} =1, \\
			 	& \dots, X_{\pi( n_0+ \dots + n_m +1 )} = m+1, \dots X_{\pi( n_0+ \dots + n_{m+1} )} = m+1 \big) \\
			 	&= \sum_{\pi \in \Pi} \prod_{k=0}^{m+1} p_k^{n_k}.
		\end{align*}
		Here, $\Pi$ is the set of permutations of $\{1, \dots, n \}$ such that $n_0$ of the $X_t$ are equal to $0$, $n_1$ of the $X_t$ are equal to $1$ and so on.
        We have $n!$ possible permutations of the set $\{1, \dots, n\}$, where the $n_0!$ permutations of the set $\{t \vert X_t = 0\}$ are indistinguishable. 
        The same holds for $n_1$, $n_2$, etc. 
        We conclude that
			 \[ P( O_0 =  n_0, O_1 = n_1, \dots, O_{m+1} = n_{m+1} ) 
			 = \frac{n!}{n_0!n_1!\dots n_{m+1}!} \prod_{k=0}^{m+1} p_k^{n_k}, \]
		which is the probability mass function of $MN(n, (p_0, p_1, p_2 , \dots, p_{m+1} ))$ for the corresponding probabilities $p_k$.
    \end{proof}
    
    \subsection{Statistical Tests} \label{sec:tests}
    
		This section describes multinomial tests for the null hypotheses \eqref{eq:h0_l} and \eqref{eq:nullgen} which we review for the convenience of the reader. \cite{CK}  provide a more detailed discussions, and the tests are also reviewed and used in  \cite{KLM}.  
		
We adopt three well-known approximate tests: Pearson's $\chi^2$-test, Nass' $\chi^2$-test, and the likelihood ratio test. We briefly review\footnote{\cite{CK} provide numerical comparisons of different methodologies for testing the parameters of multinomial distributions:  Pearson's $\chi^2$-test, Nass' $\chi^2$-test, the likelihood ratio test (LRT), Hoel's test and the exact test. The exact method is often impractical in applications.} the design of these tests.
		
		\subsubsection{Pearson's $\chi^2$-Test}
		
A standard test for the hypothesis that $O = (O_0, O_1, \dots, O_{m+1}) \sim \mathrm{MN}(n, (p_0, p_1, \dots, p_{m+1}))$ relies on results obtained by \cite{P}. Since $\mathsf{E} [O_k] = np_k$, $k =0, \dots, m+1$, the observed frequencies of the cell counts should be close to $np_k$ for a sufficiently large $n$. The cumulated relative squared deviations of the observations from the mean	\[ S_{m+1} := \sum_{k=0}^{m+1} \frac{(O_k - np_k)^2}{np_k}, \]
		are for large $n$ approximately $\chi^2_{m+1}$-distributed. Pearson's $\chi^2$-test at level $\kappa \in (0,1)$ rejects the hypothesis if $$S_{m+1}\;  >\;  F^{-1}_{\chi^2_{m+1}}(1-\kappa),$$ where $F^{-1}_{\chi^2_{m+1}}(1- \kappa)$ is the $\kappa$-quantile of the $\chi^2_{m+1}$ distribution.
		This test is the probably most widely used multinomial test that typically performs  well if the cell probabilities are not too small. 
		
		\subsubsection{Nass' $\chi^2$-Test}
		
\cite{N} suggests a finite sample correction of Pearson's $\chi^2$-test. Instead of approximating the distribution of $S_{m+1}$ with a $\chi^2_{m+1}$-distribution, \cite{N} proposes to use a distribution that depends on $n$, namely the distribution of $\frac 1 c Z$ with $Z \sim \chi^2_{\nu}$ 
where 
\begin{align*}
c \; = & \; \frac{2 \mathsf{E}[S_{m+1}]}
		{\mathrm{Var}(S_{m+1})}, \quad \nu \;= \;c E[S_{m+1}] ,  \quad E[S_{m+1}]\; =\; m+1,  \\ \mathrm{Var}(S_{m+1})  \;= \;  &  \; 2(m+1) \; -\; \frac{m^2 + 6m + 6}{n} \; \;+ \;  \; \sum_{k=0}^{m+1} \frac{1}{ n \cdot p_k}.
\end{align*}
Since $c \rightarrow 1$ and $\nu \rightarrow m +1$ as $n \to \infty$, the distributions used in the finite sample approximation are asymptotically equal to the asymptotic distribution known from Pearson's $\chi^2$-test. But in Nass' $\chi^2$-test the distribution of the approximating $\frac 1 c Z$ matches for each $n$ the first two moments of the distribution of $S_{m+1}$, while $\chi^2_{m+1}$ matches only the first moment. One can conjecture that this might typically lead to a better approximation than Pearson's test.
		
On the basis of this finite sample correction that matches the first two moments, the null hypothesis is rejected in Nass' $\chi^2$-test at level $\kappa \in (0,1)$, if
		$cS_{m+1}  >  F^{-1}_{\chi^2_{\nu}} (1-\kappa).$ 
		
		\subsubsection{Likelihood Ratio Test}
		A likelihood ratio test (LRT) is a standard procedure in hypothesis testing that compares the ratio of the likelihood of the sample under the null hypothesis and without any restriction. 
		More precisely,  suppose that $\Theta$ is the parameter set and that $\Theta_0 \subseteq \Theta$ denotes the null hypothesis. By $X$ we denote the observations and by $p_\vartheta$ the likelihood function for $\vartheta \in \Theta$. The corresponding likelihood ratio is given by 
		\[ \lambda(X) = \frac{\sup_{\vartheta \in \Theta_0} p_\vartheta(X)}{
		\sup_{\vartheta \in \Theta} p_\vartheta(X)}. \]
The asymptotic distribution of $- 2 \log \lambda (X)$ for the number of samples going to $\infty$ can conveniently be characterized in models that satisfy suitable regularity conditions, see e.g.~Section 10.3 of \cite{CasellaBergerRoger2002}. 

A result of this type holds in particular, if $\Theta$ includes all multinomial distributions for $n$ trials and $m+2$ possible outcomes and the null hypothesis contains only the single distribution $\mathrm{MN}(n, (p_0, p_1, \dots, p_{m+1}))$. In this case, the LRT statistic is
	    \[ R \; =  \;  2 \sum_{k=0}^{m+1} O_k \log \left( \frac{O_k }{n p_k} \right) \]	    with an asymptotic $\chi_{m+1}^2$-distribution for $n \to \infty$.
	    The corresponding LRT with level $\kappa \in (0,1)$ rejects
	    the null hypothesis if $R  >  
	    F^{-1}_{\chi_{m+1}^2} (1- \kappa).$

    \subsection{The Skewed t-Distribution of \cite{FS1}}\label{app:ST}
    Consider the probability density function (pdf) $ g_\nu $ of the t-distribution with $\nu$ degrees of freedom as
    \[ g_\nu(x) = \frac{\Gamma \left( \frac{\nu+1}{2} \right)}{\sqrt{ \nu \pi} \Gamma \left( \frac{\nu}{2} \right)} \left( 1 + \frac{x^2}{\nu} \right)^{- \frac{\nu +1}{2} }. \]
    \cite{FS1} proposes a class of skewed distributions with the pdf
    \[ f_{\nu, \gamma} (x) = \frac{2}{\gamma + \frac{1}{\gamma}}\left( g_\nu \left( \frac{x}{\gamma}\right) \mathds{1}_{\{x \in [0. \infty) \}} + g_\nu (\gamma x ) \mathds{1}_{\{x \in (-\infty. 0] \}} \right) \]
    for $\gamma \in (0.\infty)$. \\
    For the simulation we need to determine the expectation and variance of the skewed t-distribution defined above. 
    From \cite{FS1} we know that if $X \sim f_{\nu, \gamma}(x)$ then
    \[ E[X^r] = M_r \frac{ \gamma^{r+1} + \frac{(-1)^r}{\gamma^{r+1}}}{
    \gamma + \frac{1}{\gamma}} \]
    where
    $M_r = \int_0^\infty s^r 2 g_\nu(s) ds. $
    With $Y \sim g_\nu(s)$, we have that (see \cite{R} for the expectation of $\vert Y \vert$)
    $$
        M_1 \; = \; \int_0^\infty 2s \frac{\Gamma
            \left( \frac{\nu+1}{2} \right)}{ \sqrt{\nu \pi } \Gamma\left(
            \frac{\nu}{2} \right) } \left( 1 + \frac{x^2}{\nu} 
            \right)^{- \frac{\nu+1}{2}} ds \\
            \; = \; 2 E [ \vert Y \vert ] \\
            \; = \;  \frac{2 \nu}{(\nu -1) \sqrt{\nu \pi}} \frac{\Gamma \left( 
            \frac{\nu+1}{2}\right)     }{ \Gamma \left( \frac{\nu}{2}\right)}.
$$
    As $g_\nu(s)$ is symmetric we have
    \begin{align*}
        M_2 &= \int_0^\infty 2 s^2 g_\nu(s) ds 
        = \int_0^\infty s^2 g_\nu(s) ds + \int_{- \infty}^0 s^2 g_\nu(-s) 
        ds \\
        &= \int_{- \infty}^\infty s^2 g_\nu(s)) ds 
        = E[Y^2] = \frac{\nu}{\nu- 2}.
    \end{align*}
    Now we can calculate the first two non centered moments of
    $X \sim f_\nu(x)$ as
    \begin{align*}
        E [X] &= \left( \gamma - \frac{1}{\gamma} \right) \frac{2\nu}{
        (\nu -1) \sqrt{ \pi \nu}} \frac{\Gamma \left( \frac{\nu+1}{2}
        \right) }{ \Gamma \left( \frac{\nu}{2} \right) },  \\
        E[X^2] &= \frac{\nu}{\nu-2} \frac{\gamma^3 + \frac{1}{\gamma^3} }{
        \gamma + \frac{1}{\gamma}},
    \end{align*}
    where the second moment exists if $\nu > 2$. 
    So the variance of $X$ is given as
    \begin{align*}
        \mathrm{Var}(X) &= E[X^2] - E[X]^2 \\
        &= \frac{\nu}{\nu-2} \frac{\gamma^3 + \frac{1}{\gamma^3} }{
        \gamma + \frac{1}{\gamma}} - \left( 
        \left( \gamma - \frac{1}{\gamma} \right) \frac{2\nu}{
        (\nu -1) \sqrt{ \pi \nu}} \frac{\Gamma \left( \frac{\nu+1}{2}
        \right) }{ \Gamma \left( \frac{\nu}{2} \right) } 
        \right)^2.
    \end{align*}
    
    To sample from $f_{\nu, \gamma} (x)$, we propose an acceptance-rejection algorithm with sampling distribution $g_\nu(x)$.
    For this to work we need to calculate $k > 0$ such that 
    \begin{align*}
        k &\geq \frac{f_{\nu, \gamma}(x) }{g_\nu(x)} 
        = \frac{\gamma + \frac{1}{\gamma}}{2} 
        \frac{ \left( 1 + \frac{x^2}{\nu} \right)^{- \frac{\nu+1}{2}}}{
        \left( 1 + \frac{x^2}{\gamma^2 \nu} \right)^{- \frac{\nu+1}{2}}
        \mathds{1}_{\{x \in [0. \infty) \}} + 
        \left( 1 + \frac{x^2 \gamma^2}{\nu} \right)^{- \frac{\nu+1}{2}}
        \mathds{1}_{\{x \in (- \infty. 0 )\} }}.
    \end{align*}
    We distinguish the following two cases.
    \begin{itemize}
        \item[(i)] $x \geq 0$: we require that
        \begin{align*}
            k &\geq 
            \frac{ \gamma + \frac{1}{\gamma}}{2} 
            \left( \frac{1 + \frac{x^2}{\nu} }{1 + \frac{x^2}{
            \gamma^2 \nu }} \right)^{- \frac{\nu+1}{2}}. 
        \end{align*}
        The function $x^{- \frac{\nu +1}{2} }$ is strictly falling on $x \in (0. \infty)$ and the argument of the function above is positive. 
        Thus, we find $k$ by minimizing the argument 
        \[ i(x) := \frac{1 + \frac{x^2}{\nu}}{ 1 + \frac{x^2}{\gamma \nu}}. \]
        If $\gamma > 1$, $i(x)$ is strictly decreasing and the maximum is found at $\lim_{x \rightarrow 0+} i(x)$.
        On the other hand, if $\gamma \in (0.1)$ the function is strictly monotone increasing and the maximum is $\lim_{x \rightarrow \infty}
        i(x)$. 
        We can calculate this limit as
        \begin{align*}
            \lim_{x \rightarrow \infty} \frac{1 + \frac{x^2}{\nu}}{
            1 + \frac{x^2}{\gamma^2 \nu}} 
            = \lim_{x \rightarrow \infty} \underbrace{\frac{1}{1 + \frac{x^2}{\gamma^2\nu}} }_{ \rightarrow 1} + \underbrace{\frac{1}{ \nu \left(\frac{1}{x^2} + \frac{1}{ \gamma^2 \nu}\right)} }_{\rightarrow \gamma^2}
            = 1 + \gamma^2.
        \end{align*}
        Consequently, for $x \geq 0$ we have that
        \[
        \frac{\gamma + \frac{1}{\gamma}}{2}\left( \frac{1+  \frac{x^2}{\nu}}{1 + \frac{x^2}{\gamma^2 \nu}} \right)^{- \frac{\nu+1}{2}} \leq 
        \begin{cases}
        \frac{\gamma + \frac{1}{\gamma}}{2} (1 + \gamma^2)^{- \frac{\nu+1}{2}}
        \hspace{5mm} &\gamma < 1 \\
        1 & \gamma = 1 \\
        \frac{\gamma+ \frac{1}{\gamma}}{2} & \gamma > 1.
        \end{cases}\]
    \item[(ii)] $x < 0$: we want to determine $k$ such that 
        \[ k \geq \frac{\gamma + \frac{1}{\gamma}}{2} \left( 
        \frac{1 + \frac{x^2}{\nu}}{ 1 + \frac{x^2 \gamma^2}{\nu}} \right)^{-
        \frac{\nu+1}{2}}. \]
        Using the same reasoning as above we minimize the argument
        \[ i(x) := \frac{1 + \frac{x^2}{\nu}}{ 1 + \frac{x^2 \gamma^2}{\nu}} \]
        to find the maximum of the density quotient. 
        If $\gamma < 1$ we find the minimum of $i(x)$ at $\lim_{x \rightarrow 0-} i(x)$. 
        We have that
        \begin{align*}
            \lim_{x \rightarrow 0-} \frac{1+ \frac{x^2}{ \nu}}{ 1 + 
            \frac{x^2 \gamma^2 }{\nu}} 
            = \lim_{x \rightarrow 0-} \frac{1 }{1 + \frac{x^2 \gamma^2}{\nu}}
            + \frac{1}{\nu \left( \frac{1}{x^2} + \frac{\gamma^2 }{\nu} \right)} 
            = 1 + \frac{1}{\gamma^2}.
        \end{align*}
        Therefore,
        \[ \frac{\gamma + \frac{1}{\gamma}}{2} 
        \left( \frac{1 + \frac{x^2}{\nu}}{ 1 + \frac{x^2 \gamma^2}{\nu}} \right)
        ^{- \frac{\nu+1}{2}} \leq \begin{cases}
            \frac{\gamma + \frac{1}{\gamma}}{2 } \left( \frac{1}{\gamma^2}\right)
            ^{- \frac{\nu+1}{2}}  \hspace{10mm} & \gamma > 1 \\
            1  & \gamma =1 \\
            \frac{\gamma + \frac{1}{\gamma}}{2} & \gamma < 1
        \end{cases}. \]
    \end{itemize}
    
    From this computation, we determine $k$ as follows: 
    if $\gamma \geq 1$ we set 
    \begin{align*}
        k:= \max \left\{ \frac{\gamma + \frac{1}{\gamma}}{ 
        \frac{\gamma + \frac{1}{\gamma}}{2}} \left( \frac{1}{\gamma^2}
        \right)^{- \frac{\nu +1 }{2}} \right\} 
        = \frac{\gamma + \frac{1}{\gamma}}{2} \gamma^{\nu+1};
    \end{align*}
    if $\gamma < 1$ we set
    \begin{align*}
        k:= \max \left\{ \frac{\gamma + \frac{1}{\gamma}}{2} 
        \frac{\gamma + \frac{1}{\gamma}}{2} \left( 1 + \gamma^2 \right)^
        {- \frac{\nu+1}{2}} \right\} 
        = \frac{\gamma + \frac{1}{\gamma}}{2}.
    \end{align*}
                                        
    \subsection{Computations for Section \ref{sec:avarstudy}} \label{app:avarstudy}
    Consider the distortion function $g(u) =  \frac{u}{\alpha} \mathds{1}_{\{0 \leq u < \alpha\}} + \mathds{1}_{\{u \geq \alpha \}}$ of the AV@R at level $\alpha$.
    For $G \sim g$ and $0 \leq a < b \leq 1$ we can calculate $\mathsf{E}[ G \mathds{1}_{\{ G \in [a, b) \}}]$ by considering the following cases:
    \begin{itemize}
        \item[(i)]if $a, b < \alpha$, 
        \begin{align*}
            \mathsf{E} \left[ G \mathds{1}_{\{G \in [a,b)\}} \right] = 
            \int_a^b u \frac{1}{\alpha} du  =
            \frac{1}{2 \alpha} (b^2 - a^2 ); 
        \end{align*} 
        \item[(ii)] if $a < \alpha, b \geq \alpha$,
        \begin{align*}
            \mathsf{E} \left[ G \mathds{1}_{\{G \in [a,b)\}} \right] = 
            \int_a^\alpha u \frac{1}{\alpha} du  =
            \frac{1}{2\alpha} (\alpha^2 - a^2 );
        \end{align*} 
        \item[(iii)] if $a, b \geq \alpha$,
        \begin{align*}
            \mathsf{E} \left[ G \mathds{1}_{\{G \in [a,b)\}} \right] = 0.
        \end{align*} 
    \end{itemize}
    Let $\alpha_0, \alpha_1, \dots, \alpha_{m}$ be a partition of 
    $[0, \alpha]$, where $\alpha_m \leq \alpha$ and set $\alpha_{m+1} = 1$. 
    Plugging this in the definition of the probabilities $p_i$ in Theorem \ref{theo:RightContCellCount}, we obtain
     \begin{align*}
         p_0 &= 1 - \frac{\mathsf{E} [ G\mathds{1}_{\{ G \in [\alpha_m,1)\}} ]}{g(1) - g(\alpha_m)} \\
         &= 1 - \frac{ \frac{1}{2 \alpha} (\alpha^2 - \alpha_m^2) }{1 - \frac{\alpha_m}{\alpha}}; \\
         p_k &= \frac{E \left[ G \mathds{1}_{\{ G \in [\alpha_{m+1-k}, \alpha_{m+2-k})\}} \right] }{g(\alpha_{m+2-k}) - g(\alpha_{m+1-k})} - \frac{E \left[ G \mathds{1}_{\{G \in [\alpha_{m-k}, \alpha_{m+1-k}) \}} \right] }{g(\alpha_{m+1-k})-g(\alpha_{m-k})} \\
         &= \frac{\frac{1}{2\alpha} \left( \alpha_{m+2-k}^2 - \alpha_{m+1-k}^2 \right)}{ \frac{1}{2\alpha} ( \alpha_{m+2-k} - \alpha_{m+1-k})} - \frac{\frac{1}{2\alpha} \left( \alpha_{m+1-k}^2 - \alpha_{m-k}^2 \right)}{ \frac{1}{2\alpha} ( \alpha_{m+1-k} - \alpha_{m-k})} \\ 
         &= \frac{\left( \alpha_{m+2-k}^2 - \alpha_{m+1-k}^2 \right)}{  ( \alpha_{m+2-k} - \alpha_{m+1-k})} - \frac{ \left( \alpha_{m+1-k}^2 - \alpha_{m-k}^2 \right)}{  ( \alpha_{m+1-k} - \alpha_{m-k})}; \mbox{ and } \\
         p_{m+1} &= \frac{\mathsf{E} [ G \mathds{1}_{\{G \in [0, \alpha_1) \}} ]}{g(\alpha_1)} \\
         &= \frac{\frac{1}{2\alpha} \alpha_1^2}{ \frac{\alpha_1}{\alpha}} = \frac{\alpha_1}{2},
     \end{align*}
     where $k \in \{1, \dots, m\}$.

    \subsection{Computations for Section \ref{sec:gvsim}} \label{app:gvsim}

	    Recall the GlueV@R risk measure introduced in Example \ref{ex:drm} corresponds to the distortion function
	        \[  g(u) = \begin{cases} \frac{h_1}{ \beta} u \hspace{30mm} & 
			\text{ if } 0 \leq u  \leq \beta \\
			h_1 + \frac{h_2 - h_1}{ \alpha - \beta} ( u-  \beta) 
			& \text{ if }  \beta <  u \leq \alpha \\
			1 & \text{ if }  \alpha < u \leq 1 
			\end{cases}. \]
		where $0  \leq \beta \leq \alpha\leq 1$ and $0 \leq  h_1 \leq h_2 \leq 1$.
	    The partition is again set such that $(\alpha_1, \dots, \alpha_m)$ are equidistant in $[0, \alpha]$ and $\alpha_0 = 0$, $\alpha_{m+1}= 1$.
	    
	    To sample from $G \vert G \in [\alpha_j, \alpha_{j+1})$ for $G \sim g $, we use the inverse transform method. 
	    If $U \sim \mathrm{unif}(0,1)$, then we have that $q_{g_c}(U) \sim g$. 
	    If we let $V$ be the restriction of $U$ to the interval $[\alpha_j, \alpha_{j+1})$, as 
            \[ V := g(\alpha_{j}) + (g(\alpha_{j+1})- g( \alpha_j)) U, \]
        we obtain that $q_{g_c}(V)$ is distributed as $G \vert G \in [\alpha_j, \alpha_{j+1})$. The left quantile function $q_{g_c}$ of the distrotion function $g$ can be calculated as
            \[ q_{g_c}(v) = \begin{cases}
                \frac{\beta}{h_1} v  \hspace{10mm} & ,v \in [0, h_1) \\
                \beta + \frac{\alpha - \beta}{h_2 - h_1} (v - h_1) &, v \in [h_1, h_2] \\
                \alpha &, v \in (h_2, 1]. 
            \end{cases} \]
	    
	    Now we can compute the probabilities $p_k$ from Theorem \ref{theo:RightContCellCount}.
	    The distortion function of the Glue-V@R defines the probability $\mathsf{P}$ with an atom at $\alpha$, where $\mathsf{P}(G = \alpha) = 1- h_2$.
        For the absolutely continuous part of $P$ we have
        $$  dg(u) = \begin{cases}
            \frac{h_1}{\beta} du &\, u \in [0, \beta) \\
            \frac{h_2 - h_1}{\alpha - \beta}  du&\, u \in [\beta, \alpha)\\
            0 &\, u \in [\alpha, 1].
        \end{cases}$$
    The computation of the $p_k$ requires calculating 
        $E \left[ G \mathds{1}_{\{ G \in [\alpha_{j}, \alpha_{j+1})\}} \right]$
    in the following five cases.
            \begin{itemize} 
                \item[(i)] \textit{$\alpha_j \in [0,\beta)$, $ \alpha_{j+1} \in [0, \beta)$}:
                \begin{align*}
                    \frac{\mathsf{E} \left[ G \mathds{1}_{\{G \in [ \alpha_j, \alpha_{j+1})\}} \right]}{g(\alpha_{j+1}) - g(\alpha_j)} &=  \frac{\int_{\alpha_j}^{\alpha_{j+1}} u \frac{h_1}{\beta}
                    du}{ \frac{h_1}{\beta} 
                    (\alpha_{j+1} - \alpha_j)} 
                    = \frac{1}{2 } ( \alpha_{j+1}  - \alpha_j).
                \end{align*}
                \item[(ii)] \textit{$\alpha_j \in [0, \beta)$, $\alpha_{j+1} \in [\beta, \alpha)$}: 
                \begin{align*}
                    \frac{\mathsf{E} \left[ G  \mathds{1}_{\{ G \in [ \alpha_j, \alpha_{j+1}) \}} \right]}{g(\alpha_{j+1}) - g(\alpha_j) } 
                    &= \frac{\int_{\alpha_j}^\beta u dg(u) + \int_\beta^{\alpha_{j+1}} u dg(u)}{h_1 + \frac{h_2 - h_1 }{\alpha - \beta} (\alpha_{j+1}  - \beta) - \frac{h_1}{\beta} \alpha_j} 
                    = \frac{\frac{h_1}{\beta} ( \beta - \alpha_j)^2 + \frac{h_2 - h_1 }{2 ( \alpha - \beta) } ( \alpha_{j+1} - \beta)^2 }{h_1 + \frac{h_2 - h_1}{\alpha - \beta} 
                    ( \alpha_{j+1} - \beta) - \frac{h_1}{\beta} \alpha_j}.
                \end{align*}
                \item[(iii)] \textit{$\alpha_j \in [0, \beta)$, $\alpha_{j+1} \in [\alpha, 1)$}:
                \begin{align*}
                    \frac{\mathsf{E} \left[ G \mathds{1}_{\{ G \in [ \alpha_j, \alpha_{j+1})\}} \right]}{ g(\alpha_{j+1}) - g( \alpha_j)} &= \frac{\frac{h_1}{\beta} \int_{\alpha_j}^\beta u du + \frac{h_2 - h_1 }{\alpha- \beta } \int_{\beta}^\alpha u du + \alpha \mathsf{P}( G = \alpha)}{ 1 - \frac{h_1 }{\beta} \alpha_j} \\
                    &= \frac{\frac{h_1}{2 \beta} ( \beta - \alpha_j)^2 + \frac{h_2 - h_1 }{2} (\alpha - \beta) + (1- h_2)\alpha }{1 - \frac{h_1}{\beta} \alpha_j}.
                \end{align*}
                \item[(iv)] \textit{$\alpha_j \in [\beta, \alpha)$, $\alpha_{j+1} \in [\beta, \alpha)$}: 
                \begin{align*}
                    \frac{\mathsf{E} \left[ G \mathds{1}_{\{ G \in [ \alpha_j, \alpha_{j+1})\}} \right]}{ g(\alpha_{j+1}) - g( \alpha_j)} &= \frac{\frac{h_2- h_1}{\alpha- \beta} \int_{\alpha_j}^{\alpha_{j+1}} u du}{ h_1 + \frac{h_2 - h_1 }{ \alpha - \beta} ( \alpha_{j+1} - \beta) - h_1
                   - \frac{h_2 - h_1 }{\alpha- \beta} ( \alpha_j - \beta)} 
                = \frac{1}{2} ( \alpha_{j+1} - \alpha_j).
                \end{align*}
                
                \item[(v)] \textit{$ \alpha_j \in [\beta, \alpha)$, $\alpha_{j+1} \in [\alpha,1 )$}:
                \begin{align*}
                        \frac{\mathsf{E} \left[ G \mathds{1}_{\{ G \in [ \alpha_j, \alpha_{j+1})\}} \right]}{ g(\alpha_{j+1}) - g( \alpha_j)} &= \frac{\frac{h_2 - h_1}{\alpha- \beta} \int_\beta^\alpha u du + \alpha \mathsf{P}(G = \alpha)}{1 - h_1 + \frac{h_2 - h_1 }{\alpha - \beta}(\alpha_j - \beta)} 
                        = \frac{\frac{h_2-h_1 }{2} (\alpha - \beta) + \alpha(1- h_2) }{1 - h_1 + \frac{h_2 - h_1}{\alpha - \beta}(\alpha_j - \beta)}.
                \end{align*}
            \end{itemize}
        We can plug these calculations in Theorem \ref{theo:RightContCellCount} to obtain $p_k$ for the multinomial distribution of the cell counts $O$. \\
    
    \subsection{Computations for Section \ref{sec:genstudy}}\label{app:4.3}
    We consider the distortion function 
    \[ g(u) := \begin{cases}  \frac{h_1}{\beta} u 
    \hspace{30mm}& u \in [0 , \beta]\\
    h_2 + \frac{h_3 - h_2}{\alpha-\beta} (u - \beta) & u \in (\beta, 
    \alpha) \\
    1 & u \in [\alpha, 1]
    \end{cases}, \]
    where $0 \leq h_1 < h_2 < h_3 < 1$ and $\beta \leq \alpha$.
    
    \subsubsection{The Unique Decomposition from Theorem \ref{theo:uniqueconvexdecomposition}}
   By Theorem \ref{theo:uniqueconvexdecomposition} the unique decompoisition of the distortion function $g(u)$ is given as
    \[ g(u) = c_r g_{r} (u) + c_l g_{l}(u) + c_c g_c(u), \]
    where $g_{r}, g_{l}$ are right- resp. left-continuous step distortion functions and $g_c$ is a continuous distortion function. \\
    As in the proof of Theorem \ref{theo:uniqueconvexdecomposition} we have then
    \begin{align*}
        g_{r} (u) &= \mathds{1}_{\{u \geq \alpha\}}, \hspace{3mm} 
        &c_r = 1- h_3, \\
        g_{l} (u) &= \mathds{1}_{\{u > \beta\}}, &c_l = h_2 - h_1. 
    \end{align*}
    Therefore
    \begin{align*}
        g_c &= c ( g(u) - c_r g_{r}(u) - c_l g_{l}(u))  \\
        &= c \cdot \begin{cases}
            \frac{h_1}{\beta} u \hspace{40mm} &u \in [0, \beta] \\
            h_2 + \frac{h_3 - h_2 }{\alpha- \beta}(u-\beta) - h_2 + h_1
            & u \in (\beta, \alpha) \\
            1 - (1- h_3) - (h_2 - h_1) & u \in [\alpha,1] 
        \end{cases} \\
        &= c \cdot \begin{cases}
            \frac{h_1}{\beta} u \hspace{40mm} & u \in [0, \beta] \\
            h_1 + \frac{h_3 - h_2}{\alpha - \beta} (u - \beta) 
            & u \in ( \beta, \alpha) \\
            h_3 - h_2 + h_1  & u \in [\alpha,1]
        \end{cases}.
    \end{align*}
    To normalize $g_c$, we set
    \[ c = \frac{1}{h_3 - h_2 + h_1} \]
    and thus 
    \[ c_c = h_3 - h_2 + h_1. \]

	\subsubsection{Sampling from $g_l$, $g_r$, $g_c$}
		The distortion functions $g_r$ and $g_l$ describe the trivial distributions only taking $\alpha$ resp. $\beta$.	
        To sample from $g_c$ we calculate the quantile function $q_{g_c}$ as
        \begin{align*}
           	q_{g_c} ( p) = \begin{cases}
           		\frac{(h_3 - h_2 + h_1) \beta}{h_1} p \hspace*{40mm}& 0 \leq p \leq \frac{h_1}{h_3 - h_2 + h_1} \\
           		\frac{(\alpha- \beta) ((h_3 - h_2 + h_1) u - h_1)}{h_3 - h_2} + \beta & \frac{h_1}{h_3 -h_2 + h_1} < p \leq 1 
           	\end{cases}
        \end{align*}
       We can then sample $G \vert G \in [\alpha_j, \alpha_{j+1})$ be setting
       \[ V := g(\alpha_j) + ( g( \alpha_{j+1} ) - g( \alpha_j)) U \]
       and the inverse transform method, by considering 
       \begin{align*}
       		G \vert G \in [ \alpha_j, \alpha_{j+1} ) \overset{d}{=}  \begin{cases}
       			 \frac{\beta}{h_1} V \hspace{20mm} &  V \in [0, h_1) \\
       			 \beta  & V \in [h_1, h_2) \\
       			 \frac{(\alpha- \beta) ( V - h_2)}{h_3 - h_2} + \beta & V \in [h_1, h_3) \\
       			 \alpha & V \in  [ h_3, 1]
       		\end{cases}
		\end{align*}        
        
     \subsubsection{Calculation of the probabilities $p_k$}
   
   To calculate 
   $ \mathsf{E}[ G\mathds{1}_{\{ G \in [\alpha_j , \alpha_{j+1}) \}} ] $
   for $j \in \{0, \dots, m \}$, we first have that
   \begin{align*}
   		\mathsf{E} \left[ G \mathds{1}_{\{ G \in [\alpha_j, \alpha_{j+1}) \}} \right] &= \mathsf{E} \left[ \mathsf{E}\left[ 
   		G \mathds{1}_{\{ G \in [\alpha_j, \alpha_{j+1}) \}} \vert C \right] \right]  \\
   		&= c_l \mathsf{E}[ G \mathds{1}_{\{ G \in [\alpha_j, \alpha_{j+1}) \}} \vert C = l ] 
   		+ c_r \mathsf{E}[ G \mathds{1}_{\{ G \in [\alpha_j, \alpha_{j+1}) \}} \vert C = r ] \\
   		&+ c_c \mathsf{E}[ G \mathds{1}_{\{ G \in [\alpha_j, \alpha_{j+1}) \}}\vert C = c ] \\
   		&= c_l  \beta \mathds{1}_{\{ \beta \in [\alpha_j, \alpha_{j+1}) \}} + c_r \alpha \mathds{1}_{\{ \alpha \in [\alpha_j, \alpha_{j+1}) \}} + c_c \int_{[\alpha_j, \alpha_{j+1})} u dg_c(u).   		
   \end{align*}
   The differential of $g_c(u)$ is given as
   \begin{align*}
   		dg_c(u) = \begin{cases}
   				\frac{h_1}{\beta(h_3 - h_2 + h_1)} du\hspace{20mm} & u \in [0, \beta] \\
   				\frac{h_3 - h_2}{(\alpha - \beta)(h_3 - h_2 + h_1)} du& u \in (\beta, \alpha) \\
   				0 & u \in [ \alpha , 1]
   		\end{cases}.   
   \end{align*}
   Then we consider the following cases.
   \begin{itemize}
   		\item[(i)] If $\alpha_j \in [0, \beta]$, $\alpha_{j+1} \in [0, \beta]$, then
   		\begin{align*}
   			\int_{[\alpha_j, \alpha_{j+1})} u dg_c(u) &= \int_{[\alpha_j, \alpha_{j+1})} u \frac{h_1}{\beta(h_3- h_2 + h_1)} du \\
   			&= \frac{h_1}{\beta( h_3- h_2 + h_1)} \left[ \frac{1}{2} u^2 \right]^{\alpha_{j+1}}_{\alpha_j} \\
   			&= \frac{h_1}{2\beta( h_3 - h_2 + h_1 )} \left( \alpha_{j+1}^2 - \alpha_j^2 \right). 
   		\end{align*}
   		
   		\item[(ii)] If $\alpha_j \in [0, \beta]$, $\alpha_{j+1} \in ( \beta, \alpha)$, then
   		\begin{align*}
   			\int_{[\alpha_j, \alpha_{j+1})} u dg_c(u) &= \int_{[ \alpha_j, \beta] } u dg_c(u) + \int_{(\beta, \alpha_{j+1})} u dg_c(u) \\
   			&= \frac{h_1}{2 \beta( h_3 - h_2 + h_1)} \left( \beta^2 - \alpha_j^2 \right) + \frac{h_3 - h_2 }{(\alpha - \beta)(h_3 - h_2 + h_1 )} 
   			\left[ \frac{1}{2} u^2 \right]^{\alpha_{j+1}}_{\beta}  \\
   			&=\frac{h_1}{2 \beta( h_3 - h_2 + h_1)} \left( \beta^2 - \alpha_j^2 \right) + \frac{h_3 - h_2 }{2 (\alpha - \beta)(h_3 - h_2 + h_1 )} 
   			\left(\alpha_{j+1}^2 - \beta^2 \right). 
   		\end{align*}
   		
   		\item[(iii)] If $\alpha_j \in [0,\beta]$, $\alpha_{j+1} \in [\alpha, 1]$, then
   		\begin{align*}
   			\int_{[\alpha_j , \alpha_{j+1})} u dg_c(u) &= \int_{ [ \alpha_j, \alpha]} u dg_c(u) \\
   			&=\frac{h_1}{2 \beta( h_3 - h_2 + h_1)} \left( \beta^2 - \alpha_j^2 \right) + \frac{h_3 - h_2 }{2 (\alpha - \beta)(h_3 - h_2 + h_1 )} 
   			\left(\alpha^2 - \beta^2 \right).
   		\end{align*}
   		
   		\item[(iv)] If $\alpha_j \in (\beta, \alpha)$, $\alpha_{j+1} \in (\beta, \alpha)$, then
   		\begin{align*}
   			\int_{[\alpha_j , \alpha_{j+1})} u dg_c(u) &= \int_{[\alpha_j, \alpha_{j+1})} u \frac{h_3- h_2 }{(\alpha - \beta)(h_3- h_2 + h_1)} du \\
   			&= \frac{h_3 - h_2 }{2(\alpha - \beta)(h_3 - h_2 + h_1)}  \left( \alpha_{j+1}^2 - \alpha_j^2 \right) .
   		\end{align*}
   		
   		\item[(v)] If $\alpha_j \in (\beta, \alpha)$, $\alpha_{j+1} \in [\alpha, 1]$, then
   		\begin{align*}
   			\int_{[\alpha_j, \alpha_{j+1})} u dg_c(u) &= \int_{[\alpha_j, \alpha)} u dg_c(u) \\
   			&=\frac{h_3 - h_2 }{2(\alpha - \beta)(h_3 - h_2 + h_1)}  \left( \alpha^2 - \alpha_j ^2\right).
   		\end{align*}
   		
   		\item[(vi)] If $\alpha_j \in [\alpha,1]$, $\alpha_{j+1} \in [\alpha, 1]$, then
   		$
   			\int_{[\alpha_j, \alpha_{j+1})} u dg_c(u) = 0. 
   		$
   \end{itemize}
   Finally, we plug in the above results in the formulas of Lemma \ref{lem:GenExProb} to obtain the $p_k$.

\subsection{Range Value at Risk}\label{sec:rv@r}   

RV@R is a distortion risk measure with distortion function
$g(u) = \frac{u- \beta}{\alpha - \beta} \mathds{1}_{\{ \beta < u \leq \alpha\}}
		+ \mathds{1}_{\{ u > \alpha \}}$, $0<\beta <\alpha$. This risk measure ignores the tail beyond the value at risk at level $\alpha$ and is therefore -- just as the value at risk -- not sensitive to extreme events. In the limiting case $\beta \to 0$, it coincides with tail-sensitive $AV@R_\alpha$. We repeat the case studies of Section~\ref{sec:DistStudies} for $RV@R$. The results of the backtests are shown in Tables \ref{fig:RV@R_1}, \ref{fig:RV@R_2} \& \ref{fig:RV@R_3} for 	$\alpha = 0.025$,  $\beta \in \{0.015, 0.005, 0.001 \}$ and sample size $N = 20000$. We use the partition $0 = \alpha_0$, $\alpha_i = \beta + (\alpha - \beta) i / (m+1)$ for $i \in \{1, \dots, m \}$, and $\alpha_{m+1} =1$, and use otherwise the same backtesting strategy as for $AV@R$. The size of the asymptotic tests is in all cases reasonable. The power is generally poor for $\beta = 0.015$ and improves, when smaller values of $\beta$ are chosen. This is due to the fact that for decreasing $\beta$ the risk measures $RV@R$ depends on more extreme parts of the tail. This is also apparent from Figure~\ref{fig:RV@R} where we plot the size of the test and the power for the three considered alternatives a function of $\beta$ in the case $m=16$, $n= 1000$. The size is quite reasonable for all values of $\beta$, but the power is low unless $\beta$ is very close to zero. It increases to acceptable levels as $\beta \to 0$ from above, i.e., when approaching the limiting case $AV@R_\alpha$. A similar behavior holds also for other levels of $\alpha$; this is displayed in Figures~\ref{fig:RV@R3}~\&~\ref{fig:RV@R4} for $\alpha = 0.5\%$ and $\alpha = 10\%$, respectively. The key finding is that risk measures that neglect -- like V@R -- the tail beyond some threshold level or quantile are, of course, less sensitive to tail properties of the distribution, and this diminishes the power.
  
 The numerical experiments in Tables \ref{fig:RV@R_1}, \ref{fig:RV@R_2} \& \ref{fig:RV@R_3} for 	$\alpha = 0.025$,  $\beta \in \{0.015, 0.005, 0.001 \}$ were also repeated by us for the non-randomized method suggested in \cite{KLM}.  The comparison showed again that our randomized method generally improves the power, while retaining a similar size. For $\beta = 0.015$, the power is low for both methods and quite similar in most cases. For $\beta \in \{0.005, 0.001 \}$, the power improves. Im particular, for large $n$ and $m=2, 4$ the randomized method performes best as well as better than \cite{KLM}. As an example, we include the detailed comparison for  $\beta = 0.005$ in Table \ref{fig:RV@R_2_KLMComp}.
  
\begin{table}
\begin{tabularx}{1.0\textwidth}{p{5mm} | X | X X X X X X X}
    	        \multicolumn{9}{c}{\textbf{Nass}}  \\
    	        \hline \\
    	        $L_t$ & $n | m$ & 1 & 2 & 4 & 8 & 16 & 32 & 64 \\
    	        \hline \hline 
    	        $\mathcal{N}$ 
    	        & 250 & \cellcolor{vg} $ 0.89 $&$ 0.75 $&\cellcolor{vg} $ 1.20 $&\cellcolor{g} $ 0.98 $&\cellcolor{g} $ 0.99 
    	        $&\cellcolor{vg}$ 0.87 $&$ 1.26 $ \\
    	        & 500 & $ 0.84 $&\cellcolor{g} $ 0.97 $&\cellcolor{vg}$ 0.87 $&\cellcolor{g} $ 0.97 $&\cellcolor{g} $ 1.06 
    	        $&\cellcolor{g} $ 1.04 $&\cellcolor{vg}$ 0.88 $ \\
    	         & 1000 & \cellcolor{vg}$0.88 $&\cellcolor{g} $ 0.92 $&\cellcolor{g} $ 0.94 $&\cellcolor{g} $ 1.06 $&
    	         \cellcolor{g} $ 1.03 $&\cellcolor{g} $ 1.10 $&\cellcolor{g} $ 1.08 $  \\   	        
    	        & 2000 &\cellcolor{g}  $ 0.96 $&\cellcolor{g} $ 0.97 $&\cellcolor{g} $ 0.97 $&\cellcolor{g} $ 0.98 $&\cellcolor{g} $ 
    	        1.04 $&\cellcolor{vg} $ 1.11 $&\cellcolor{g} $ 1.03 $ \\
    	        \hline 
    	         T3 
    	        & 250 & \cellcolor{vp}$ 1.68 $& \cellcolor{vp}$ 1.65 $& \cellcolor{vp}$ 1.36 $& \cellcolor{vp}$ 0.57 $
    	        & \cellcolor{vp}$ 0.40 $& \cellcolor{vp}$ 0.32 $& \cellcolor{vp}$ 0.46 $ \\
    	        & 500 & \cellcolor{vp} $3.17 $& \cellcolor{vp}$ 1.97 $& \cellcolor{vp}$ 0.78 $& \cellcolor{vp}$ 0.40 $
    	        & \cellcolor{vp}$ 0.24 $& \cellcolor{vp}$ 0.26 $& \cellcolor{vp}$ 0.14 $ \\
    	         & 1000 & \cellcolor{vp} $ 6.48 $& \cellcolor{vp}$ 4.32 $& \cellcolor{vp}$ 1.49 $
    	         & \cellcolor{vp}$ 0.38 $& \cellcolor{vp}$ 0.12 $& \cellcolor{vp}$ 0.07 $& \cellcolor{vp}$ 0.06 $ \\ 	        
    	        & 2000 &  \cellcolor{p}$ 21.70 $& \cellcolor{p}$ 20.51 $& \cellcolor{vp}$ 9.90 $
    	        & \cellcolor{vp}$ 1.91 $& \cellcolor{vp}$ 0.17 $& \cellcolor{vp}$ 0.04 $& \cellcolor{vp}$ 0.03 $ \\
    	         \hline 
    	         T5
    	        & 250 &  \cellcolor{vp}$ 4.96 $& \cellcolor{vp}$ 4.21 $& \cellcolor{vp}$ 3.81 $& \cellcolor{vp}$ 2.29 $
    	        & \cellcolor{vp}$ 2.03 $& \cellcolor{vp}$ 1.80 $& \cellcolor{vp}$ 2.08 $ \\
    	        & 500 & \cellcolor{vp}$ 6.46 $& \cellcolor{vp}$ 6.17 $& \cellcolor{vp}$ 3.49 $& \cellcolor{vp}$ 2.38 $
    	        & \cellcolor{vp}$ 1.81 $& \cellcolor{vp}$ 1.54 $& \cellcolor{vp} $ 1.04 $ \\
    	         & 1000 &  \cellcolor{p}$ 10.62 $& \cellcolor{vp}$ 9.39 $& \cellcolor{vp}$ 6.16 $& \cellcolor{vp}$ 3.38 $
    	         & \cellcolor{vp}$ 2 $& \cellcolor{vp}$ 1.18 $& \cellcolor{vp}$ 0.95 $ \\  	        
    	        & 2000 & \cellcolor{p}$ 19.98 $& \cellcolor{p}$ 18.61 $& \cellcolor{p}$ 13.33 $& \cellcolor{vp}$ 7.34 $
    	        & \cellcolor{vp}$ 2.94 $& \cellcolor{vp}$ 1.17 $& \cellcolor{vp}$ 0.57 $ \\
    	        \hline 
    	         ST
    	        & 250 &\cellcolor{vp} $ 7.46 $&\cellcolor{vp}$ 6.68 $&\cellcolor{vp}$ 4.29 $&\cellcolor{vp}$ 1.47 $
    	        &\cellcolor{vp}$ 0.95 $&\cellcolor{vp}$ 0.75 $&\cellcolor{vp}$ 0.78 $ \\
    	        & 500 & \cellcolor{p}$ 13.13 $&\cellcolor{p}$ 11.37 $&\cellcolor{vp}$ 6.12 $&\cellcolor{vp}$ 2.83 $
    	        &\cellcolor{vp}$ 1.09 $&\cellcolor{vp}$ 0.71 $&\cellcolor{vp}$ 0.36 $ \\
    	         & 1000 &\cellcolor{p}$ 26.51 $&\cellcolor{p}$ 24.41 $&\cellcolor{p}$ 16.48 $&\cellcolor{vp}$ 7.76 $&
    	         \cellcolor{vp}$ 2.19 $&\cellcolor{vp}$ 0.64 $&\cellcolor{vp}$ 0.17 $ \\  
    	        & 2000 &$ 53.41 $&$ 53.58 $&$ 44.59 $&\cellcolor{p}$ 27.66 $&\cellcolor{vp}$ 9.49 $&\cellcolor{vp}$ 1.36 $&
    	        \cellcolor{vp}$ 0.17 $ \\ 
    	        \hline 	        
    	        \end{tabularx}
\caption{Backtesting $RV@R$ with $\beta = 0.015$ and $\alpha = 0.025$: Estimated size (for hypothesis $H_0$ with distribution
$\mathcal{N}$) and power in $\%$ (for the alternatives $H_1$ with distributions $T3$, $T5$, $ST$, respectively) for the Pearson test, Nass test and LRT. The size is represented as the fraction of the true size according to our simulations divided by the desired level $\kappa = 5\%$. The colouring scheme for the size is as follows: Values between $0.8 - 1.2$ are green, values between
$0.9 - 1.1$ are dark green; values above $1.5$ are red, above $2$ dark red. The colouring scheme for 
the power is adopted from \cite{KLM}: Green refers to a power $\geq 70\%$; light red indicates a power $\leq 30\%$; dark 
red indicates poor results with a power $\leq 10\%$. }\label{fig:RV@R_1}
\end{table}

\begin{table}
\begin{tabularx}{1.0\textwidth}{p{5mm} | X | X X X X X X X}
    	        \multicolumn{9}{c}{\textbf{Nass}}  \\
    	        \hline \\
    	        $L_t$ & $n | m$ & 1 & 2 & 4 & 8 & 16 & 32 & 64 \\
    	        \hline \hline 
    	        $\mathcal{N}$ 
    	        & 250 & $ 0.71 $&\cellcolor{g} $ 0.90 $&\cellcolor{vg} $ 0.86 $&\cellcolor{g}$ 0.98 $&
    	        \cellcolor{vg} $ 1.14 $&\cellcolor{g}$ 1.06 $&\cellcolor{vg} $ 0.88 $ \\
    	        & 500 &\cellcolor{vg}  $ 0.85 $&\cellcolor{g}$ 0.96 $&\cellcolor{vg} $ 0.89 $&\cellcolor{g}$ 1.10 $
    	        &\cellcolor{g}$ 1.08 $&\cellcolor{vg}$ 1.11 $&\cellcolor{g}$ 1.08 $ \\
    	         & 1000  &\cellcolor{g} $ 0.98 $&\cellcolor{g}$ 0.95 $&\cellcolor{g}$ 0.90 $&\cellcolor{g}$ 1.04 $&\cellcolor{g}$ 1.08 
    	         $&\cellcolor{g}$ 1.10 $&\cellcolor{vg}$ 1.12 $ \\
    	        & 2000 &\cellcolor{g} $ 0.98 $&\cellcolor{g}$ 0.95 $&\cellcolor{g}$ 0.98 $&\cellcolor{g}$ 1 $&\cellcolor{g}$ 1.06 $&\cellcolor{g}$ 1.07 $&\cellcolor{vg}$ 1.12 $ \\
    	        \hline 
    	         T3 
    	        & 250 & \cellcolor{vp}$ 7.49 $&\cellcolor{vp}$ 9.01 $&\cellcolor{vp}$ 6.99 $&\cellcolor{vp}$ 4.21 $
    	        &\cellcolor{vp}$ 1.57 $&\cellcolor{vp}$ 0.66 $&\cellcolor{vp}$ 0.26 $ \\
    	        & 500 &\cellcolor{p}$ 15.55 $&\cellcolor{p}$ 16.54 $&\cellcolor{p}$ 13.44 $&\cellcolor{vp}$ 8.24 $
    	        &\cellcolor{vp}$ 3.33 $&\cellcolor{vp}$ 0.75 $&\cellcolor{vp}$ 0.18 $ \\
    	         & 1000 &$ 35.82 $&$ 38.16 $&$ 35.91 $&\cellcolor{p}$ 24.43 $&\cellcolor{p}$ 10.32 $
    	         &\cellcolor{vp}$ 2.33 $&\cellcolor{vp}$ 0.27 $ \\   	        
    	        & 2000 &$ 68.38 $&\cellcolor{g}$ 76.86 $&\cellcolor{g}$ 78.44 $&$ 67.86 $&$ 43.69 $&\cellcolor{p}$ 14.09 $
    	        &\cellcolor{vp}$ 1.44 $ \\
    	         \hline 
    	         T5
    	        & 250 & \cellcolor{p}$ 10.27 $&\cellcolor{p}$ 11.76 $&\cellcolor{vp}$ 8.94 $&\cellcolor{vp}$ 6.97 $&
    	        \cellcolor{vp}$ 4.32 $&\cellcolor{vp}$ 2.73 $&\cellcolor{vp}$ 2.06 $ \\
    	        & 500 &\cellcolor{p}$ 17.80 $&\cellcolor{p}$ 17.92 $&\cellcolor{p}$ 15.22 $&\cellcolor{p}$ 11.84 $&
    	        \cellcolor{vp}$ 7.03 $&\cellcolor{vp}$ 3.43 $&\cellcolor{vp}$ 2.07 $ \\
    	         & 1000 &$ 32.11 $&$ 32.20 $&$ 30.80 $&\cellcolor{p}$ 23.69 $&\cellcolor{p}$ 14.71 $&\cellcolor{vp}$ 6.34 $
    	         &\cellcolor{vp}$ 2.58 $ \\  	        
    	        & 2000 & $ 54.40 $&$ 59.11 $&$ 58.86 $&$ 50.96 $&$ 37.52 $&\cellcolor{p}$ 19.02 $&\cellcolor{vp}$ 6.39 $ \\
    	        \hline 
    	         ST
    	        & 250 &\cellcolor{p}$ 21.31 $&\cellcolor{p}$ 23.41 $&\cellcolor{p}$ 19.40 $&\cellcolor{p}$ 13.94 $
    	        &\cellcolor{vp}$ 6.83 $&\cellcolor{vp}$ 2.64 $&\cellcolor{vp}$ 1.17 $ \\
    	        & 500 & $ 39.52 $&$ 42.03 $&$ 38.82 $&$ 30.55 $&\cellcolor{p}$ 17.89 $
    	        &\cellcolor{vp}$ 6.57 $&\cellcolor{vp}$ 1.63 $ \\
    	         & 1000 &$ 69.27 $&\cellcolor{g}$ 72.11 $&\cellcolor{g}$ 71.50 $&$ 64.32 $&$ 46.04 $&\cellcolor{p}$ 23.13 $&
    	         \cellcolor{vp}$ 6 $ \\  	        
    	        & 2000 &\cellcolor{g}$ 93.38 $&\cellcolor{g}$ 96.14 $&\cellcolor{g}$ 96.59 $&\cellcolor{g}$ 94.72 $&
    	        \cellcolor{g}$ 88.51 $&\cellcolor{g}$ 69.92 $&$ 32.97 $ \\
    	        \hline 	        
    	        \end{tabularx}
\caption{Backtesting $RV@R$ with $\beta= 0.005$ and $\alpha = 0.025$: Estimated size (for hypothesis $H_0$ with distribution
$\mathcal{N}$) and power in $\%$ (for the alternatives $H_1$ with distributions $T3$, $T5$, $ST$, respectively) for the Pearson test, Nass test and LRT. The size is represented as the fraction of the true size according to our simulations divided by the desired level $\kappa = 5\%$. The colouring scheme for the size is as follows: Values between $0.8 - 1.2$ are green, values between
$0.9 - 1.1$ are dark green; values above $1.5$ are red, above $2$ dark red. The colouring scheme for 
the power is adopted from \cite{KLM}: Green refers to a power $\geq 70\%$; light red indicates a power $\leq 30\%$; dark 
red indicates poor results with a power $\leq 10\%$. }\label{fig:RV@R_2}
\end{table}

\begin{table}
\begin{tabularx}{1.0\textwidth}{p{5mm} | X | X X X X X X X}
    	        \multicolumn{9}{c}{\textbf{Nass}}  \\
    	        \hline \\
    	        $L_t$ & $n | m$ & 1 & 2 & 4 & 8 & 16 & 32 & 64 \\
    	        \hline \hline 
    	        $\mathcal{N}$ 
    	        & 250 &   $ 0.74 $&\cellcolor{g}$ 0.99 $&\cellcolor{g}$ 1.04 $&\cellcolor{g}$ 1.09 $&
    	        \cellcolor{g}$ 1.09 $&\cellcolor{g}$ 1.04 $&\cellcolor{g}$ 0.91 $ \\
    	        & 500 & $ 0.79 $&\cellcolor{g}$ 0.90 $&\cellcolor{g}$ 1.04 $&\cellcolor{g}$ 1.03 $&
    	        \cellcolor{g}$ 1.05 $&\cellcolor{g}$ 1.04 $&\cellcolor{g}$ 1.01 $ \\
    	         & 1000 & \cellcolor{vg}$ 0.86 $&\cellcolor{g}$ 0.95 $&\cellcolor{g}$ 0.98 $&
    	         \cellcolor{g}$ 1.06 $&\cellcolor{g}$ 1.04 $&\cellcolor{g}$ 1.09 $&\cellcolor{g}$ 1.08 $ \\   	        
    	        & 2000 &\cellcolor{g}$ 0.99 $&\cellcolor{g}$ 0.97 $&\cellcolor{g}$ 0.96 $&\cellcolor{g}$ 1.04 $&
    	        \cellcolor{g}$ 1.05 $&\cellcolor{g}$ 1.09 $&\cellcolor{vg}$ 1.11 $ \\
    	        \hline 
    	         T3 
    	        & 250 & \cellcolor{p} $ 12.14 $&\cellcolor{p} $ 18.22 $&\cellcolor{p} $ 19.82 $&\cellcolor{p} $ 16 $
    	        &\cellcolor{p} $ 13.33 $&\cellcolor{vp} $ 6.71 $&\cellcolor{vp} $ 2.60 $ \\
    	        & 500 &\cellcolor{p} $ 24.55 $&$ 31.22 $&$ 34.09 $&$ 32.98 $&\cellcolor{p} $ 24.34 $
    	        &\cellcolor{vp} $ 16.74 $&\cellcolor{vp} $ 6.12 $ \\
    	         & 1000 &$ 48.77 $&$ 62.13 $&$ 65.91 $&$ 64.52 $&$ 55.48 $&$ 39.21 $&\cellcolor{p} $ 19.71 $ \\  
    	        & 2000 &\cellcolor{g} $ 82.67 $&\cellcolor{g} $ 91.71 $&\cellcolor{g} $ 94.80 $&\cellcolor{g} $ 95.10 $&\cellcolor{g} 
    	        $ 91.81 $&\cellcolor{g} $ 81.75 $&$ 58.79 $ \\
    	         \hline 
    	         T5
    	        & 250 &\cellcolor{p} $ 14.62 $&\cellcolor{p} $ 17.46 $&\cellcolor{p} $ 19.05 $&\cellcolor{p} $ 16.11 $&
    	        \cellcolor{p} $ 12.88 $&\cellcolor{vp} $ 9.62 $&\cellcolor{vp} $ 5.31 $ \\
    	        & 500 &\cellcolor{p} $ 23.34 $&\cellcolor{p} $ 28.62 $&$ 30.15 $&\cellcolor{p} $ 28.05 $&\cellcolor{p} $ 23.26 $&
    	        \cellcolor{p} $ 16.82 $&\cellcolor{vp} $ 9.05 $ \\
    	         & 1000 &  $ 40.33 $&$ 47.04 $&$ 51.19 $&$ 50.85 $&$ 45.08 $&$ 34.89 $&\cellcolor{p} $ 20.55 $ \\    	        
    	        & 2000 & $ 67.08 $&\cellcolor{g}  $ 76.53 $&\cellcolor{g}$ 81.26 $&\cellcolor{g}$ 81.16 $&\cellcolor{g}$ 77.37 $&
    	        $ 67.37 $&$ 50.02 $ \\
    	        \hline 
    	         ST
    	        & 250 &\cellcolor{p}$ 29.59 $&$ 35.95 $&$ 39.63 $&$ 36.11 $&$ 31.47 $&\cellcolor{p}$ 21.02 $
    	        &\cellcolor{p} $10.08 $ \\
    	        & 500 &$ 51.02 $&$ 59.44 $&$ 63.81 $&$ 62.87 $&$ 55.77 $&$ 45.35 $&\cellcolor{p}$ 27.75 $ \\
    	         & 1000 &\cellcolor{g}$ 80.22 $&\cellcolor{g}$ 86.84 $&\cellcolor{g}$ 90.03 $
    	         &\cellcolor{g}$ 90.38 $&\cellcolor{g}$ 87.37 $&\cellcolor{g}$ 79.75 $&$ 63.28 $ \\ 	        
    	        & 2000 & \cellcolor{g}$ 97.76 $&\cellcolor{g}$ 99.28 $&\cellcolor{g}$ 99.60 $&\cellcolor{g}$ 99.76 $
    	        &\cellcolor{g}$ 99.56 $&\cellcolor{g}$ 98.81 $&\cellcolor{g}$ 95.99 $ \\
    	        \hline 	        
    	        \end{tabularx}
\caption{Backtesting $RV@R$ with $\beta = 0.001$ and $\alpha = 0.025$: Estimated size (for hypothesis $H_0$ with distribution
$\mathcal{N}$) and power in $\%$ (for the alternatives $H_1$ with distributions $T3$, $T5$, $ST$, respectively) for the Pearson test, Nass test and LRT. The size is represented as the fraction of the true size according to our simulations divided by the desired level $\kappa = 5\%$. The colouring scheme for the size is as follows: Values between $0.8 - 1.2$ are green, values between
$0.9 - 1.1$ are dark green; values above $1.5$ are red, above $2$ dark red. The colouring scheme for 
the power is adopted from \cite{KLM}: Green refers to a power $\geq 70\%$; light red indicates a power $\leq 30\%$; dark 
red indicates poor results with a power $\leq 10\%$. }\label{fig:RV@R_3}
\end{table}

\begin{table}
\begin{tabularx}{1.0\textwidth}{p{5mm} | X | X X X X X X X}
    	        \multicolumn{9}{c}{\textbf{Nass}}  \\
    	        \hline \\
    	        $L_t$ & $n | m$ & 1 & 2 & 4 & 8 & 16 & 32 & 64 \\
    	        \hline \hline 
    	        $\mathcal{N}$ 
    	        & 250 & $ 0.74 $&\cellcolor{g}$ 0.92 $&\cellcolor{g}$ 1.05 $&\cellcolor{vg}$ 0.89 $&\cellcolor{g}$ 1.08 $&\cellcolor{g}$ 1.08 $&\cellcolor{g}$ 0.90 $\\
    	        & 500 &\cellcolor{vg} $ 1.18 $&\cellcolor{g}$ 0.90 $&\cellcolor{g}$ 0.98 $&\cellcolor{g}$ 1.07 $&\cellcolor{g}$ 1.10 $&\cellcolor{g}$ 1.10 $&\cellcolor{g}$ 1.06 $ \\
    	         & 1000 &\cellcolor{g} $ 0.94 $&\cellcolor{g}$ 0.98 $&\cellcolor{g}$ 0.92 $&\cellcolor{g}$ 0.97 $&\cellcolor{g}$ 1.07 $&\cellcolor{g}$ 1.09 $&\cellcolor{g}$ 1.02 $ \\
    	        & 2000 &\cellcolor{g} $ 1.05 $&\cellcolor{g}$ 0.90 $&\cellcolor{g}$ 0.94 $&\cellcolor{g}$ 0.99 $&\cellcolor{g}$ 1.03 $&\cellcolor{g}$ 1.07 $&\cellcolor{g}$ 1.09 $\\
    	        \hline 
    	        T3 
    	        & 250 & $ 1.91 $&$ 2.95 $&$ 0.94 $&$ 0.72 $&$ 0.03 $&$ 0.01 $&$ 0.02 $  \\
    	        & 500 &$ 7.02 $&$ 5.10 $&$ 4.04 $&$ 2.05 $&$ 0.66 $&$ 0.11 $&$ 0.03 $   \\
    	         & 1000 & \cellcolor{g} $ 28.51 $& \cellcolor{vg}$ 19.27 $& \cellcolor{vg}$ 11.13 $&$ 5.38 $&$ 1.59 $&$ 0.46 $&$ 0.04 $ \\
    	        & 2000 & \cellcolor{g} $ 57.49 $& \cellcolor{g}$ 30.93 $& \cellcolor{vg}$ 15.93 $& \cellcolor{vg} $ 10.05 $&$ 6.55 $&$ 2.05 $&$ -0.05 $\\
    	        \hline 
    	        T5 
    	        & 250 & $ -0.54 $&$ 2.37 $&$ 0.37 $&$ 0.45 $&$ -0.03 $&$ 0.27 $&$ 0.05 $  \\
    	        & 500 & $ 0.93 $&$ 2 $&$ 2.92 $&$ 1.67 $&$ 0.32 $&$ 0.27 $&$ 0.03 $  \\
    	         & 1000 & \cellcolor{vg} $ 10.97 $&$ 7.56 $&$ 5.56 $&$ 2.52 $&$ 1.24 $&$ 0.16 $&$ 0.07 $ \\
    	        & 2000 &  \cellcolor{g}$ 21.63 $& \cellcolor{vg}$ 14.92 $& \cellcolor{vg}$ 10.96 $&$ 6 $&$ 2.70 $&$ 1.11 $&$ 0.43 $ \\
    	        \hline 
    	       ST
    	        & 250 &$ 3.18 $&$ 6.53 $&$ 1.45 $&$ 1.31 $&$ 0.19 $&$ 0.16 $&$ 0.12 $ \\
    	        & 500 & \cellcolor{vg}$ 10.53 $&$ 8.16 $&$ 8.03 $&$ 4.55 $&$ 2.06 $&$ -0.05 $&$ -0.27 $ \\
    	         & 1000 &  \cellcolor{g} $ 26.54 $&\cellcolor{vg}$ 15.32 $&$ 9.70 $&$ 6.60 $&$ 3.69 $&$ 1.21 $&$ 0.91 $  \\
    	        & 2000 &  \cellcolor{g}$ 25.50 $&\cellcolor{vg}$ 9.30 $&$ 4.29 $&$ 2.86 $&$ 3.18 $&$ 2.96 $&$ 0.99 $ \\
    	        \hline 
    	        \end{tabularx}
\caption{ Backtesting $RV@R$: Comparison of the results in Table~\ref{fig:RV@R_2} to the method of 
\cite{KLM}. The size is represented as the fraction of the true size according to our simulations divided by the desired level $\kappa = 5\%$.  The colouring scheme for the size is as follows: Values between $0.8 - 1.2$ are dark green, values between $0.9 - 1.1$ are dark green; values above $1.5$ are red, above $2$ dark red. 
    For the alternative $T3$, $T5$ and $ST$ , the table shows the difference of the power of 
    our method and the method of KLM. The colouring scheme for the power is as follows: 
    Dark green are notable improvements of the power $\geq 20\%$; light green are improvements 
    $\geq 10\%$.}
\label{fig:RV@R_2_KLMComp}
\end{table}

\begin{figure}
    \begin{minipage}{0.49\textwidth}
        \includegraphics[scale=0.5]{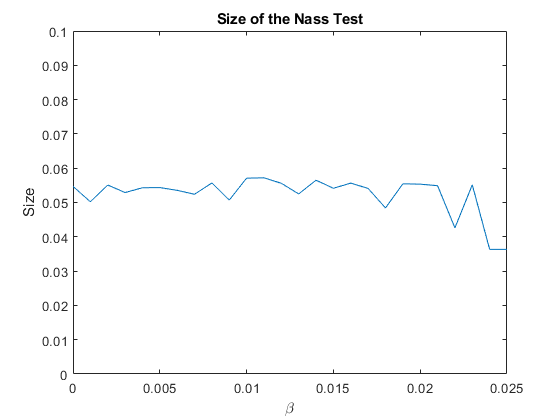}
    \end{minipage}
       \begin{minipage}{0.49\textwidth}
        \includegraphics[scale=0.5]{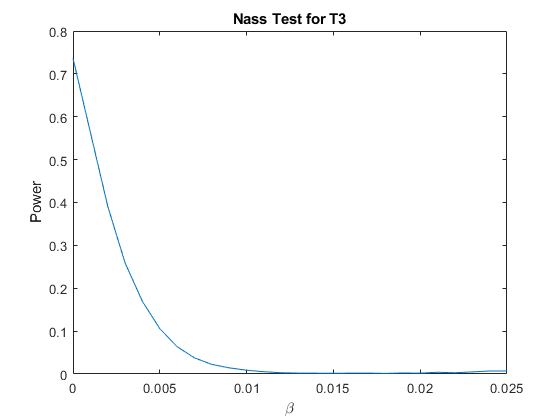}
    \end{minipage}   
     \begin{minipage}{0.49\textwidth}
        \includegraphics[scale=0.5]{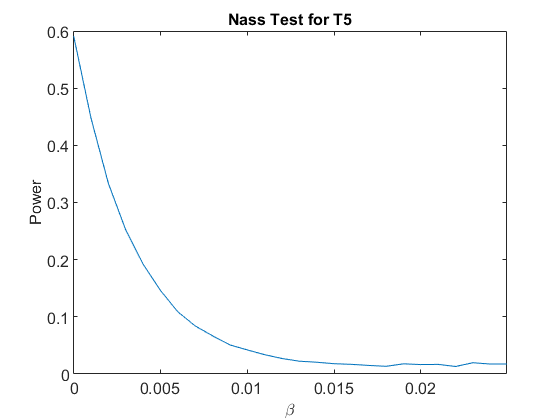}
    \end{minipage}
       \begin{minipage}{0.49\textwidth}
        \includegraphics[scale=0.5]{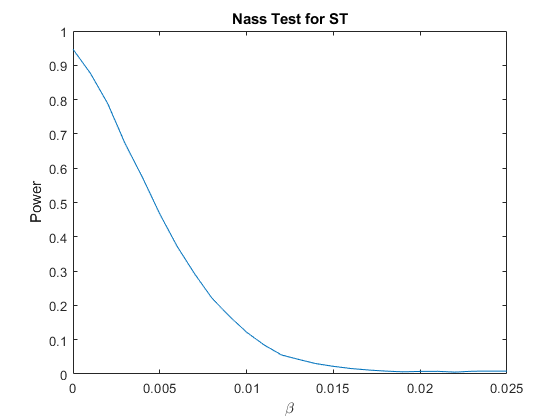}
    \end{minipage}
    \caption{Size and Power for the RV@R approaching the AV@R at level $0.025$}\label{fig:RV@R}
\end{figure}

\begin{figure}
    \begin{minipage}{0.49\textwidth}
        \includegraphics[scale=0.5]{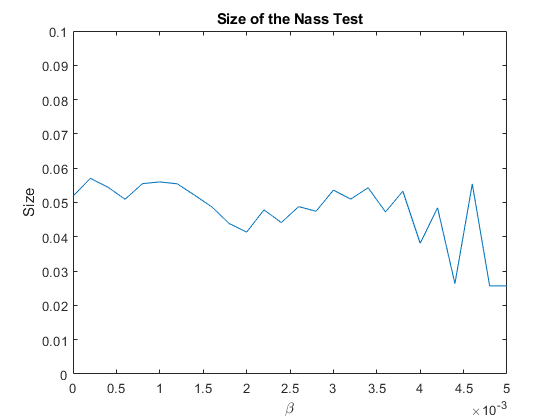}
    \end{minipage}
       \begin{minipage}{0.49\textwidth}
        \includegraphics[scale=0.5]{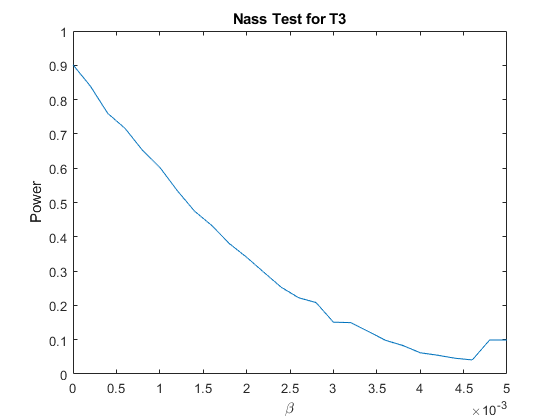}
    \end{minipage}   
     \begin{minipage}{0.49\textwidth}
        \includegraphics[scale=0.5]{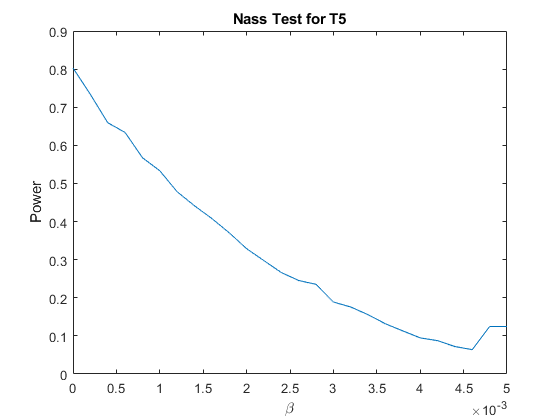}
    \end{minipage}
       \begin{minipage}{0.49\textwidth}
        \includegraphics[scale=0.5]{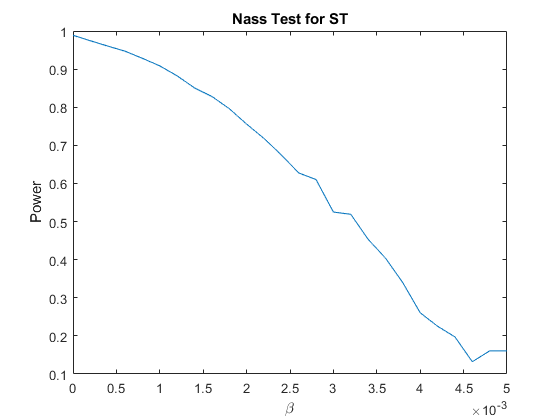}
    \end{minipage}
    \caption{Size and Power for the RV@R approaching the AV@R at level $0.005$}\label{fig:RV@R3}
\end{figure}

\begin{figure}
    \begin{minipage}{0.49\textwidth}
        \includegraphics[scale=0.5]{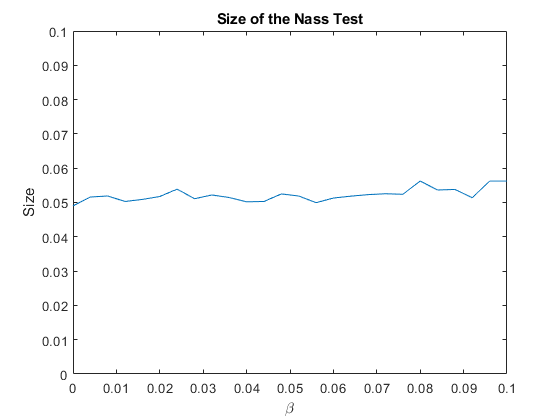}
    \end{minipage}
       \begin{minipage}{0.49\textwidth}
        \includegraphics[scale=0.5]{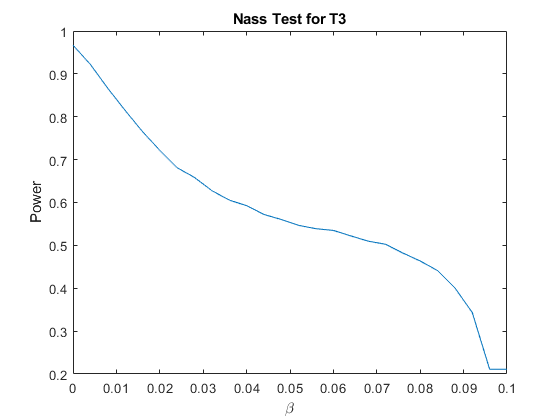}
    \end{minipage}   
     \begin{minipage}{0.49\textwidth}
        \includegraphics[scale=0.5]{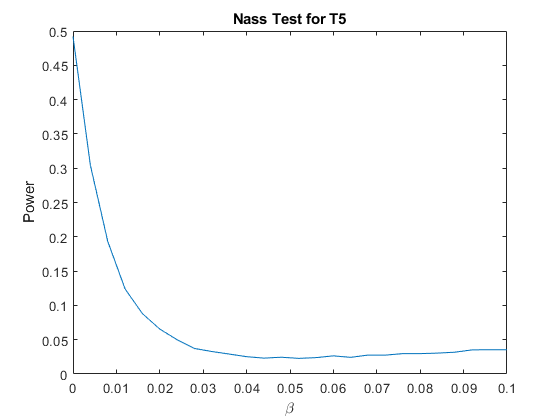}
    \end{minipage}
       \begin{minipage}{0.49\textwidth}
        \includegraphics[scale=0.5]{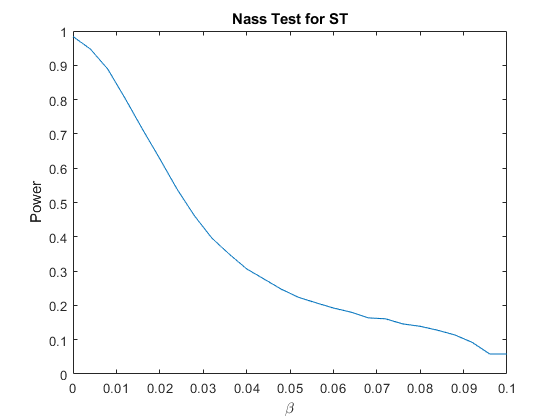}
    \end{minipage}
    \caption{Size and Power for the RV@R approaching the AV@R at level $0.1$}\label{fig:RV@R4}
\end{figure}

\subsection{S\&P 500}\label{sec:s&p500}

\begin{figure}[h!]
    \begin{minipage}{0.32\textwidth}
        \includegraphics[scale=0.2]{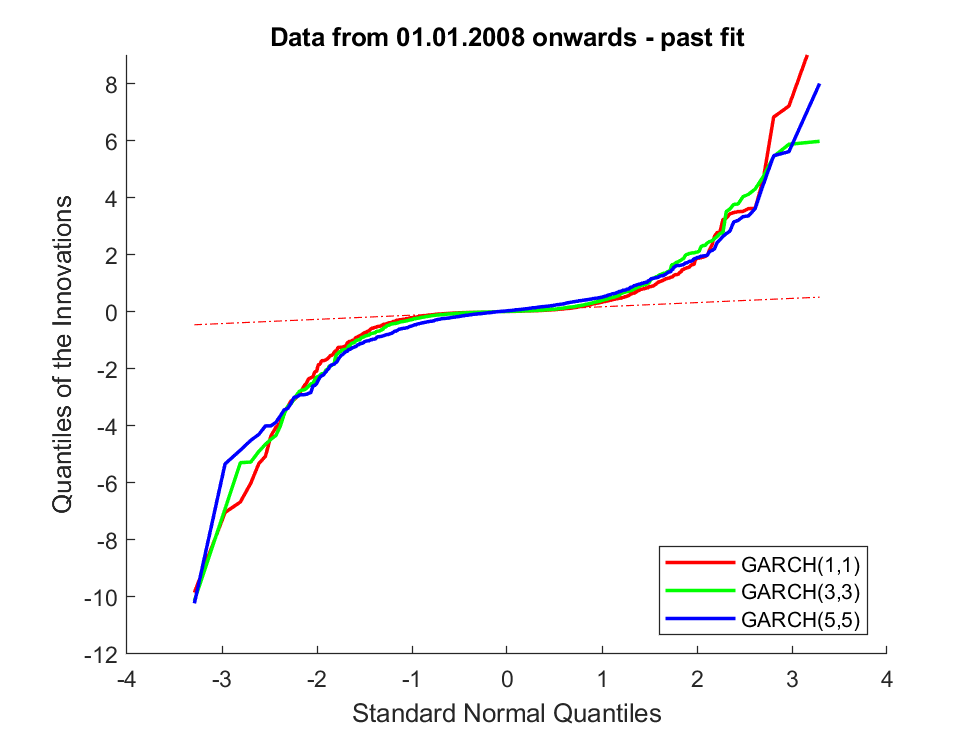}
    \end{minipage}
        \begin{minipage}{0.32\textwidth}
        \includegraphics[scale=0.2]{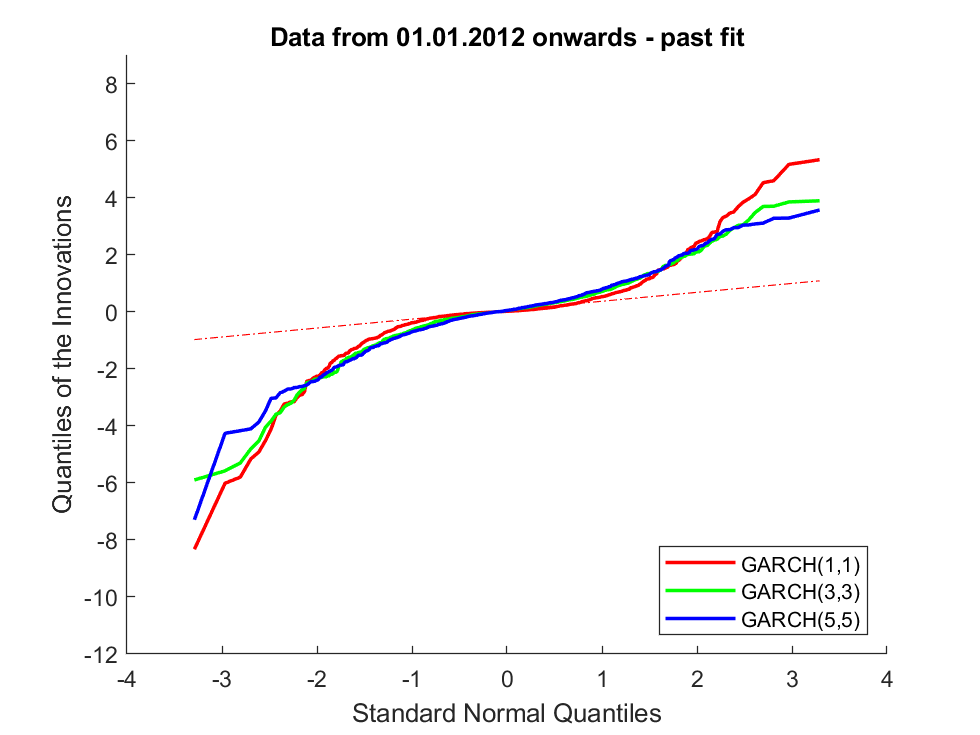}
    \end{minipage}
        \begin{minipage}{0.32\textwidth}
        \includegraphics[scale=0.2]{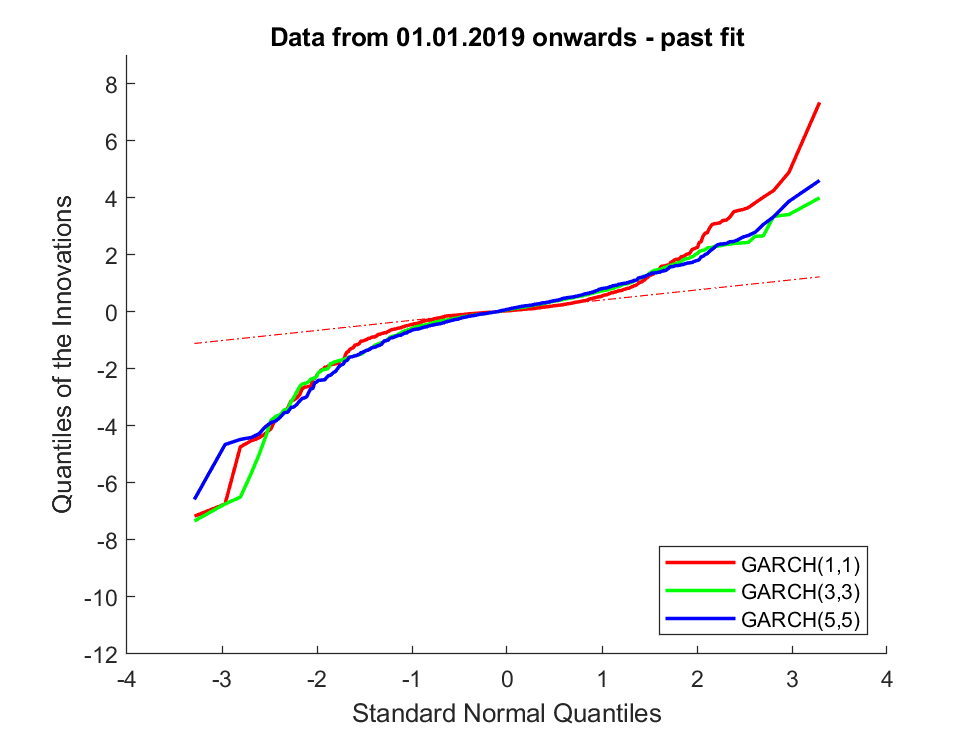}
    \end{minipage}
        \begin{minipage}{0.32\textwidth}
        \includegraphics[scale=0.2]{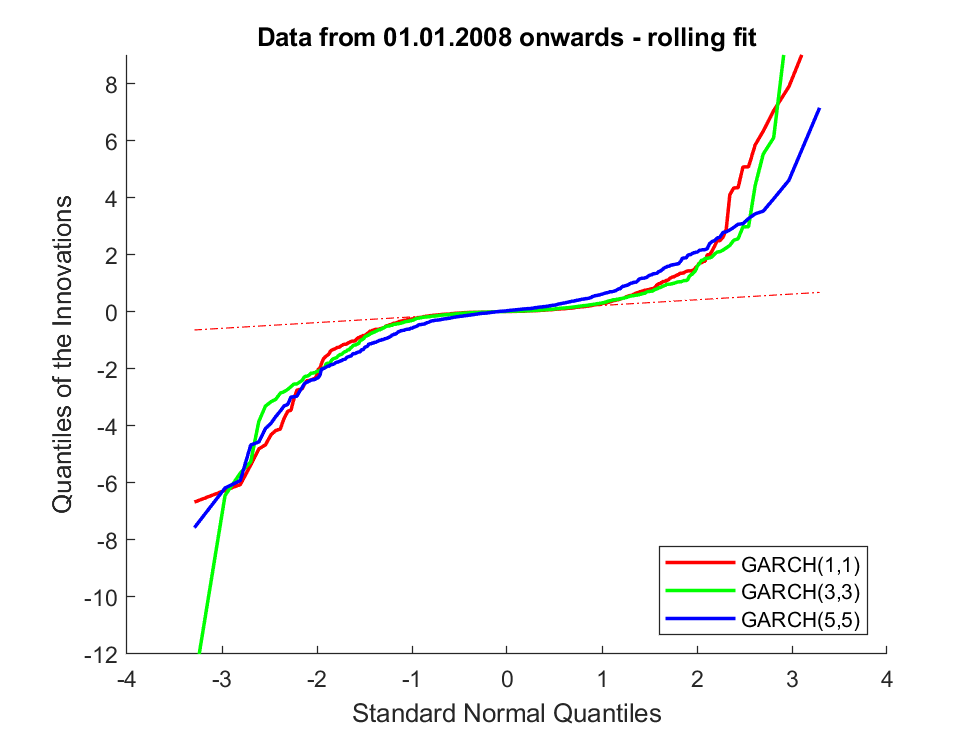}
    \end{minipage}
        \begin{minipage}{0.32\textwidth}
        \includegraphics[scale=0.2]{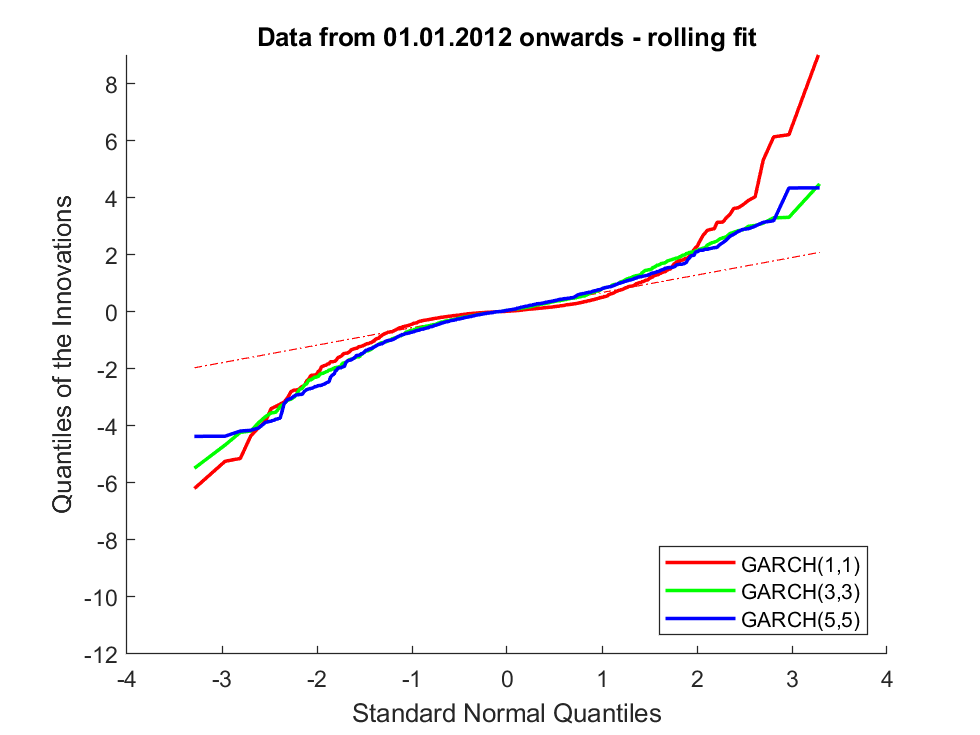}
    \end{minipage}
        \begin{minipage}{0.32\textwidth}
        \includegraphics[scale=0.2]{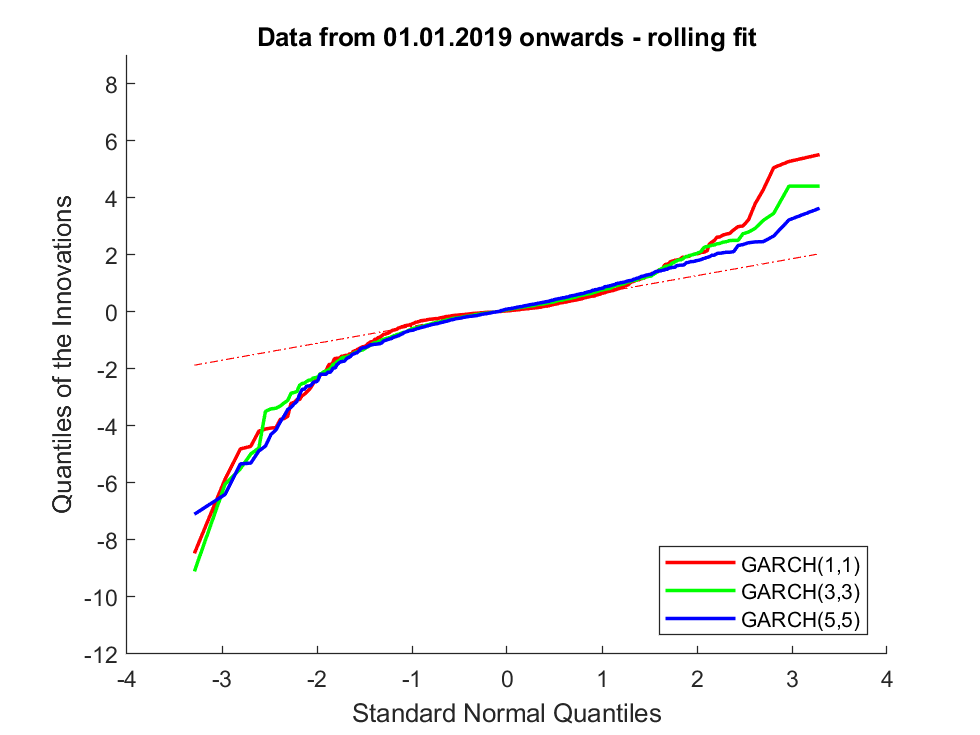}
    \end{minipage}
\caption{QQ-plots of the standard normal distributions and the normalized innovations of the fitted GARCH models
to the S\&P 500. The GARCH models are fitted on $500$ data points and the innovations are observed for the following $1000$ data points.  }\label{QQ-plots}
\end{figure}

We investigate the power of the backtests for log returns of the S\&P 500 during three different time periods, starting at the beginning of the years 2008, 2012 and 2019, respectively. Since we are focussing on losses, we consider the negative of the log returns, i.e., downside risk measures are statistical functionals of the upper tail of the corresponding distributions.

The exemplary time series models used are GARCH processes with standard normal innovations $GARCH(p,q)$ with $p=q$ and $p=1,3,5$. The model classes are nested; larger $p$ corresponds to a larger class and higher complexity.  We consider two different methodologies to fit the models and to apply our backtesting procedure. The first approach fits the GARCH-processes to the first 250 resp.~500 data points. The backtest is then applied to the 1000 following observations of the processes. The second approach does not assume that the true data-generating mechanism is a GARCH-process with constant parameters, but refits the processes at each point of time based on the latest 250 resp.~500 data points to produce a model for the next innovation. Again 1000 observations are used for the backtest.

QQ-plots demonstrate that innovations are not Gaussian and that the models are misspecified. We illustrate this for fits based on 500 data points in Figure~\ref{QQ-plots}. The deviations from normality in the upper tail are of different size when varying the time period, model, and fitting methodology. Risk measures backtests capture the importance of these discrepancies in terms of specific functionals of the downside risk. All fitted GARCH processes satisfy Bollerslev's condition for weak stationarity, cf. \cite{BOLLERSLEV1986307}, or also \cite{Lindner2009} for a concise survey of stationarity conditions.

We consider $AV@R$ and $GlueV@R$ with the same parameters as before. The probabilities of rejection (corresponding to the power) is estimated in $2000$ runs of the backtest for $m=8$ and a partition $\alpha_0 < \alpha_1 < \dots < \alpha_m < \alpha_{m+1} = 1$ such that
        $\alpha_i = i / (m+1) \cdot \alpha$, $i \in \{1, \dots, m \}$.

\begin{table}[h!]
    \centering
    \begin{tabularx}{1.0\textwidth}{p{25mm} | X X  | X X  | X X  }
        \multicolumn{7}{c}{\textbf{Nass}}  \\
         &  \multicolumn{2}{c}{01.01.2008} & \multicolumn{2}{c}{01.01.2012} & \multicolumn{2}{c}{01.01.2019} \\
          Model / Fit &  $250$ & $500 $ & $250$ & $500 $& $250$ & $500 $  \\
         \hline  
         $GARCH(1,1)$ & $1$ & $1$  & $1$ & $1$ & $1$& $1$ \\
         $GARCH(3,3)$ &  $1$ & $1$ & $1$ & $0.13$  & $0.28$ & $1$  \\
         $GARCH(5,5)$ & $1$ & $1$   & $0$ & $0.07$  & $0.03$ & $0$   \\
         \hline
         \hline      
    	\end{tabularx}
     \caption{Rejection probabilities of the backtest on log returns of the S\&P500 versus GARCH models fitted {\bf on past 
     observations from a fixed time} point with respect to the $AV@R_{0.975}$. }\label{TabSP1}
    \centering\vspace{10pt}
    \begin{tabularx}{1.0\textwidth}{p{25mm} | X X  | X X  | X X  }
        \multicolumn{7}{c}{\textbf{Nass}}  \\
         &  \multicolumn{2}{c}{01.01.2008} & \multicolumn{2}{c}{01.01.2012} & \multicolumn{2}{c}{01.01.2019} \\
        Model / Fit &  $250$ & $500 $ & $250$ & $500 $ & $250$ & $500 $  \\
         \hline  
         $GARCH(1,1)$ & $1$ & $1$   & $1$ & $1$  & $1$ & $1$    \\
         $GARCH(3,3)$ & $1$ & $1$ & $1$ & $1$  & $0.09$ & $0.48$   \\
         $GARCH(5,5)$ &$1$ &$1$   & $1$ & $1$  & $0.67$ & $0$    \\
         \hline
         \hline      
    	\end{tabularx}
     \caption{Rejection probabilities of the backtest on log returns of the S\&P500 versus GARCH models fitted {\bf with 
     a rolling collection of observations} with respect to the $AV@R_{0.975}$.}\label{TabSP2}
    \centering\vspace{10pt}
    \begin{tabularx}{1.0\textwidth}{p{25mm} | X X  | X X  | X X  }
          \multicolumn{7}{c}{\textbf{Nass}}  \\
         &  \multicolumn{2}{c}{01.01.2008} & \multicolumn{2}{c}{01.01.2012} & \multicolumn{2}{c}{01.01.2019} \\
          Model / Fit &  $250$ & $500 $ & $250$ & $500 $ & $250$ & $500 $  \\
         \hline  
         $GARCH(1,1)$ & $1$ & $1$  &   $1$ & $1$   & $1$  & $1$   \\
         $GARCH(3,3)$ & $1$ &  $1$ &   $0.87$ & $0.07$ & $0.01$ & $1$ \\
         $GARCH(5,5)$& $1$ & $1$  &  $0$ & $0$   & $0$ & $0$  \\
         \hline
         \hline      
    	\end{tabularx}
     \caption{Rejection probabilities of the backtest on log returns of the S\&P500 versus GARCH models fitted {\bf on past 
     observations from a fixed time point} with respect to the $GlueV@R^{h_1, h_2}_{\beta, \alpha}$, where $h_1 = 2/5$, $h_2 = 
     2/3$, $\alpha = 0.05$, $\beta = 0.01$.}\label{TabSP3}
    \centering\vspace{10pt}
    \begin{tabularx}{1.0\textwidth}{p{25mm} | X X  | X X | X  X }
        \multicolumn{7}{c}{\textbf{Nass}}  \\
         &  \multicolumn{2}{c}{01.01.2008} & \multicolumn{2}{c}{01.01.2012} & \multicolumn{2}{c}{01.01.2019} \\
         Model / Fit &  $250$ & $500 $ & $250$ & $500 $ & $250$ & $500 $  \\
         \hline  
         $GARCH(1,1)$& $1$ & $1$ & $1$ & $1$  & $1$ & $1$   \\
         $GARCH(3,3)$ &$1$ & $1$ & $1$ &$1$    &$1$ &$1$    \\
         $GARCH(5,5)$ & $1$ & $1$  & $1$ &$1$    & $1$ & $1$   \\
         \hline
         \hline      
    	\end{tabularx}
     \caption{Rejection probabilities of the backtest on log returns of the S\&P500 versus GARCH models fitted {\bf with a rolling collection of observations} with respect to the 
     $GlueV@R^{h_1, h_2}_{\beta, \alpha}$, where $h_1 = 2/5$, $h_2 =  2/3$, $\alpha = 0.05$, $\beta = 0.01$.}\label{TabSP4}
\end{table}

Tables \ref{TabSP1}, \ref{TabSP2}, \ref{TabSP3} \& \ref{TabSP4} show in most cases the null hypothesis is rejected, especially always during the financial crisis corresponding to the time period starting at 2008. Only during the other periods for $p\geq 3$ rejection probabilities are low in some cases corresponding to smaller deviations from normality in the upper tail. The differences between the two fitting methods make misspecification even more apparent.

\begin{table}
\begin{tabularx}{1.0\textwidth}{p{10mm} | X | X X X X X X X}
    	        \multicolumn{9}{c}{\textbf{Nass}}  \\
    	        \hline \\
    	        $L_t$ & $n | m$ & 1 & 2 & 4 & 8 & 16 & 32 & 64 \\
    	        \hline \hline 
    	        Size
    	        & 250 &  \cellcolor{vg}$ 1.11 $&$ 1.47 $&  \cellcolor{vg}$ 1.20 $& \cellcolor{g}$ 1.07 $&  \cellcolor{g}$ 1.07 $& \cellcolor{vg} $ 1.18 $&  \cellcolor{vg} $ 1.11 $ \\
    	        & 500 &\cellcolor{p}$ 1.68 $&$ 1.36 $&$ 1.25 $& \cellcolor{vg} $ 1.19 $&$ 1.26 $&$ 1.21 $&$ 1.21 $ \\
    	         & 1000 & \cellcolor{p}$ 1.73 $&\cellcolor{p}$ 1.54 $&$ 1.35 $&$ 1.21 $&  \cellcolor{vg}$ 1.20 $& \cellcolor{vg} $ 1.17 $& \cellcolor{g} $ 1.06 $ \\    	        
    	        & 2000 &\cellcolor{p} $ 1.90 $&$ 1.35 $&$ 1.28 $& \cellcolor{vg} $ 1.18 $& \cellcolor{vg}$ 1.16 $&$ 1.21 $& \cellcolor{vg} $ 1.18 $ \\
    	        \hline 
    	         NB 
    	        & 250 &$ 66.60 $&$ 65.23 $&$ 63.64 $&$ 62.22 $&$ 58.74 $&$ 58.74 $&$ 59.65 $ \\
    	        & 500 &\cellcolor{g} $ 87.85 $&\cellcolor{g}$ 88.07 $&\cellcolor{g}$ 84.25 $&\cellcolor{g}$ 83.95 $&
    	        \cellcolor{g}$ 81.84 $&\cellcolor{g}$ 79.24 $&\cellcolor{g}$ 81.95 $ \\
    	         & 1000 &\cellcolor{g}$ 98.82 $&\cellcolor{g}$ 98.87 $&\cellcolor{g}$ 98.30 $&\cellcolor{g}$ 98 $&
    	         \cellcolor{g}$ 96.69 $&\cellcolor{g}$ 96.24 $&\cellcolor{g}$ 95.89 $ \\    	        
    	        & 2000 &\cellcolor{g}$ 99.98 $&\cellcolor{g}$ 99.99 $&\cellcolor{g}$ 99.99 $&\cellcolor{g}$ 99.98 $&\cellcolor{g}$ 99.96 $&
    	        \cellcolor{g}$ 99.93 $&\cellcolor{g}$ 99.81 $ \\
    	         \hline 
    	         PAR
    	        & 250 & \cellcolor{g}$ 95.40 $&\cellcolor{g}$ 97.19 $&\cellcolor{g}$ 98.29 $&\cellcolor{g}$ 98.04 $&\cellcolor{g}$ 97.64 $&
    	       \cellcolor{g} $ 97.59 $&\cellcolor{g}$ 96.92 $ \\
    	        & 500 &\cellcolor{g} $ 99.86 $&\cellcolor{g}$ 100 $&\cellcolor{g}$ 99.94 $&\cellcolor{g}$ 99.94 $&\cellcolor{g}$ 99.98 $&
    	        \cellcolor{g}$ 99.96 $&\cellcolor{g}$ 99.96 $ \\
    	         & 1000 & \cellcolor{g}$ 100 $&\cellcolor{g}$ 100 $&\cellcolor{g}$ 100 $&\cellcolor{g}$ 100 $&\cellcolor{g}$ 100 $&
    	         \cellcolor{g}$ 100 $&\cellcolor{g}$ 100 $ \\	        
    	        & 2000 &\cellcolor{g} $ 100 $&\cellcolor{g}$ 100 $&\cellcolor{g}$ 100 $&\cellcolor{g}$ 100 $&\cellcolor{g}$ 100 $&
    	        \cellcolor{g}$ 100 $&\cellcolor{g}$ 100 $ \\
    	        \hline 
    	        LOGN
    	        & 250 & $ 46.91 $&$ 49.27 $&$ 49.89 $&$ 47.54 $&$ 46.90 $&$ 46.55 $&$ 44.87 $ \\
    	        & 500 & \cellcolor{g}$ 72.78 $&$ 69.83 $&\cellcolor{g}$ 70.99 $&$ 69.46 $&$ 69.24 $&\cellcolor{g}$ 70.11 $&$ 67.91 $ \\
    	         & 1000 & \cellcolor{g}$ 91.83 $&\cellcolor{g}$ 92.73 $&\cellcolor{g}$ 92.77 $&\cellcolor{g}$ 92.04 $&\cellcolor{g}$ 91.53 $
    	         &\cellcolor{g}$ 90.53 $&\cellcolor{g}$ 89.30 $ \\      
    	        & 2000 &\cellcolor{g} $ 99.50 $&\cellcolor{g}$ 99.68 $&\cellcolor{g}$ 99.62 $&\cellcolor{g}$ 99.68 $&\cellcolor{g}$ 99.50 $&
    	        \cellcolor{g}$ 99.50 $&\cellcolor{g}$ 99.05 $ \\
    	        \hline 	        
    	        \end{tabularx}
\caption{ALM Backtest for $AV@R_{\alpha}$ with $\alpha = 0.05$. The size is represented as fraction of estimated size divided
by the desired level $\kappa = 0.05$. Values of the size between $0.9 - 1.1$ are dark green, between $0.8 - 1.2$ are light green, above $1.5$ are red and above $2$ dark red. For the power green refers to a power $\geq 70\%$; light red to a power $\leq 30\%$
and dark red indicates a power $\leq 10\%$.}\label{ALM-AV@R}
\end{table}

\subsection{Some Further Results in the Context of ALM}\label{app:ALM}.

We repeat the case studies for $AV@R$. The parameters are chosen as in the $GlueV@R$-ALM simulation, and the partition for the risk measure test statistics is set as in the distribution simulations.  Qualitatively the results are similar as for $GlueV@R$ and shown in Table~\ref{ALM-AV@R}.

Since the true size of the test deviates from the a priori targeted level, since its construction is based on asymptotic results, we construct alternative tests which are not based on the target size $0.05$ but on a smaller size $\kappa=0.025$ at the expense of deterioration of the power. This is displayed in Tables~\ref{ALMapp_size_1} \& \ref{ALMapp_size_2}. The realized size is in most case smaller than the targeted size, and also the power of the tests not reduced by a large amount. The strategy of using a smaller than target size for the asymptotic construction appears to be a viable method.

\begin{table}
\begin{tabularx}{1.0\textwidth}{p{10mm} | X | X X X X X X X}
    	        \multicolumn{9}{c}{\textbf{Nass}}  \\
    	        \hline \\
    	        $L_t$ & $n | m$ & 1 & 2 & 4 & 8 & 16 & 32 & 64 \\
    	        \hline \hline 
    	        Size
    	        & 250 & $ 0.57 $& \cellcolor{vb} $ 0.91 $& \cellcolor{vb}$ 0.93 $&$ 0.62 $&$ 0.60 $&$ 0.72 $&$ 0.57 $ \\
    	        & 500 & \cellcolor{b}$ 1.12 $&$ 0.85 $&$ 0.68 $&$ 0.66 $&$ 0.68 $&$ 0.74 $&$ 0.67 $ \\
    	         & 1000 &  $ 1.22 $& \cellcolor{b}$ 0.89 $& \cellcolor{b}$ 0.82 $&$ 0.70 $&$ 0.62 $&$ 0.65 $&$ 0.72 $ \\   	        
    	        & 2000 &$ 1.24 $& \cellcolor{vb}$ 0.95 $&$ 0.75 $&$ 0.65 $&$ 0.69 $&$ 0.59 $&$ 0.67 $ \\
    	        \hline 
    	         NB 
    	        & 250 & $ 60.81 $&$ 56.89 $&$ 55.63 $&$ 54 $&$ 50.71 $&$ 49.48 $&$ 49.31 $ \\
    	        & 500 &\cellcolor{g}$ 86.17 $&\cellcolor{g}$ 85.37 $&\cellcolor{g}$ 79.84 $&\cellcolor{g}$ 78.95 $&\cellcolor{g}$ 75.99 $&
    	        \cellcolor{g}$ 72.64 $&\cellcolor{g}$ 75.97 $ \\
    	         & 1000 &\cellcolor{g}$ 98.29 $&\cellcolor{g}$ 98.09 $&\cellcolor{g}$ 97.40 $&\cellcolor{g}$ 96.61 $
    	         &\cellcolor{g}$ 94.81 $&\cellcolor{g}$ 94.31 $&\cellcolor{g}$ 93.30 $ \\
    	        & 2000 &\cellcolor{g}$ 99.99 $&\cellcolor{g}$ 99.98 $&\cellcolor{g}$ 99.98 $&\cellcolor{g}$ 99.95 $&\cellcolor{g}$ 99.92 $&
    	        \cellcolor{g}$ 99.83 $&\cellcolor{g}$ 99.64 $ \\
    	         \hline 
    	         PAR
    	        & 250 &\cellcolor{g}$ 93.59 $&\cellcolor{g}$ 95.62 $&\cellcolor{g}$ 97.34 $&\cellcolor{g}$ 97.08 $&\cellcolor{g}$ 97.04 $&
    	        \cellcolor{g}$ 96.99 $&\cellcolor{g}$ 95.84 $ \\
    	        & 500 &\cellcolor{g} $ 99.83 $&\cellcolor{g}$ 99.91 $&\cellcolor{g}$ 99.88 $&\cellcolor{g}$ 99.93 $&\cellcolor{g}$ 99.94 $&
    	        \cellcolor{g}$ 99.93 $&\cellcolor{g}$ 99.92 $ \\
    	         & 1000 &\cellcolor{g}$ 100 $&\cellcolor{g}$ 100 $&\cellcolor{g}$ 100 $&\cellcolor{g}$ 100 $&\cellcolor{g}$ 100 $
    	         &\cellcolor{g}$ 100 $&\cellcolor{g}$ 100 $ \\
    	        & 2000 &\cellcolor{g}$ 100 $&\cellcolor{g}$ 100 $&\cellcolor{g}$ 100 $&\cellcolor{g}$ 100 $&\cellcolor{g}$ 100 $
    	        &\cellcolor{g}$ 100 $&\cellcolor{g}$ 100 $ \\
    	        \hline 
    	        LOGN
    	        & 250 &  $ 40.21 $&$ 39.43 $&$ 43.04 $&$ 40.07 $&$ 38.90 $&$ 38.12 $&$ 36.79 $ \\
    	        & 500 &$ 66.71 $&$ 65.37 $&$ 64.89 $&$ 63.40 $&$ 62.49 $&$ 61.93 $&$ 60.08 $ \\
    	         & 1000 &\cellcolor{g}$ 89.83 $&\cellcolor{g}$ 89.96 $&\cellcolor{g}$ 89.24 $&\cellcolor{g}$ 88.46 $&
    	         \cellcolor{g}$ 88.86 $&\cellcolor{g}$ 87.85 $&\cellcolor{g}$ 86 $ \\	        
    	        & 2000 &\cellcolor{g} $ 99.15 $&\cellcolor{g}$ 99.24 $&\cellcolor{g}$ 99.24 $&\cellcolor{g}$ 99.43 $&\cellcolor{g}$ 99.03 $&
    	        \cellcolor{g}$ 99.22 $&\cellcolor{g}$ 98.66 $ \\ 
    	        \hline 	        
    	        \end{tabularx}
\caption{ALM Backtest for $AV@R_{\alpha}$ with $\alpha = 0.05$. The level of the tests is set to $\kappa = 0.025$. The estimated size is represented as fraction to $0.05$. Values of the size between $0.9 - 1.1$ are dark blue, between $0.8 - 1.2$ are light blue. For the power green refers to a power $\geq 70\%$; light red to a power $\leq 30\%$
and dark red indicates a power $\leq 10\%$.}\label{ALMapp_size_1}
\end{table}

\begin{table}
\begin{tabularx}{1.0\textwidth}{p{10mm} | X | X X X X X X X}
    	        \multicolumn{9}{c}{\textbf{Nass}}  \\
    	        \hline \\
    	        $L_t$ & $n | m$ & 1 & 2 & 4 & 8 & 16 & 32 & 64 \\
    	        \hline \hline 
    	        Size
    	        & 250 & \cellcolor{vb}$ 1.08 $& \cellcolor{vb}$ 0.99 $& \cellcolor{vb}$ 0.95 $&$ 0.69 $&$ 0.66 $&$ 0.74 $&$ 0.53 $ \\
    	        & 500 & \cellcolor{vb}$ 1.03 $&$ 0.79 $&$ 0.69 $&$ 0.67 $&$ 0.69 $&$ 0.62 $&$ 0.63 $ \\
    	         & 1000 & \cellcolor{b}$ 1.17 $& \cellcolor{vb}$ 0.93 $&$ 0.71 $&$ 0.66 $&$ 0.66 $&$ 0.63 $&$ 0.75 $ \\
    	        & 2000 & $ 1.33 $& \cellcolor{vb}$ 0.92 $&$ 0.74 $&$ 0.68 $&$ 0.68 $&$ 0.57 $&$ 0.71 $ \\
    	        \hline 
    	         NB 
    	        & 250 &$ 58.97 $&$ 60.05 $&$ 56.49 $&$ 54.13 $&$ 49.21 $&$ 49.72 $&$ 49.94 $ \\
    	        & 500 &\cellcolor{g} $ 86.05 $&\cellcolor{g}$ 86.23 $&\cellcolor{g}$ 79.43 $&\cellcolor{g}$ 78.97 $&\cellcolor{g}$ 76.74 $&
    	        \cellcolor{g}$ 72.93 $&\cellcolor{g}$ 75.03 $ \\
    	         & 1000 &\cellcolor{g} $ 98.79 $&\cellcolor{g}$ 98.37 $&\cellcolor{g}$ 97.97 $&\cellcolor{g}$ 96.91 $&\cellcolor{g}$ 94.86 $&
    	         \cellcolor{g}$ 94.55 $&\cellcolor{g}$ 93.93 $ \\  	        
    	        & 2000 &\cellcolor{g}$ 99.98 $&\cellcolor{g}$ 99.97 $&\cellcolor{g}$ 99.97 $&\cellcolor{g}$ 99.93 $&\cellcolor{g}$ 99.90 $&
    	        \cellcolor{g}$ 99.86 $&\cellcolor{g} $ 99.72 $ \\
    	         \hline 
    	         PAR
    	        & 250 &\cellcolor{g}$ 94.68 $&\cellcolor{g}$ 98.03 $&\cellcolor{g}$ 97.40 $&\cellcolor{g}$ 96.81 $&\cellcolor{g}$ 96.69 $&
    	        \cellcolor{g}$ 96.80 $&\cellcolor{g}$ 96.26 $ \\
    	        & 500 &\cellcolor{g}$ 99.95 $&\cellcolor{g}$ 99.96 $&\cellcolor{g}$ 99.89 $&\cellcolor{g}$ 99.94 $&\cellcolor{g}$ 99.94 $
    	        &\cellcolor{g}$ 99.94 $&\cellcolor{g}$ 99.90 $ \\
    	         & 1000 &\cellcolor{g}$ 100 $&\cellcolor{g}$ 100 $&\cellcolor{g}$ 100 $&\cellcolor{g}$ 100 $&\cellcolor{g}$ 100 $&
    	         \cellcolor{g}$ 100 $&\cellcolor{g}$ 100 $ \\     
    	        & 2000 &\cellcolor{g}$ 100 $&\cellcolor{g}$ 100 $&\cellcolor{g}$ 100 $&\cellcolor{g}$ 100 $&\cellcolor{g}$ 100 $
    	        &\cellcolor{g}$ 100 $&\cellcolor{g}$ 100 $ \\
    	        \hline 
    	        LOGN
    	        & 250 & $ 40.75 $&$ 45.98 $&$ 43.49 $&$ 40.20 $&$ 38.59 $&$ 38.37 $&$ 36.38 $ \\
    	        & 500 & $ 69.05 $&$ 67.28 $&$ 64.63 $&$ 63.81 $&$ 62.54 $&$ 62.07 $&$ 59.76 $ \\
    	         & 1000 &\cellcolor{g}$ 91.99 $&\cellcolor{g}$ 91.47 $&\cellcolor{g}$ 89.80 $&\cellcolor{g}$ 89.11 $
    	         &\cellcolor{g}$ 87.65 $&\cellcolor{g}$ 86.83 $&\cellcolor{g}$ 85.32 $ \\   	        
    	        & 2000 &\cellcolor{g} $ 99.51 $&\cellcolor{g}$ 99.60 $&\cellcolor{g}$ 99.19 $&\cellcolor{g}$ 99.41 $
    	        &\cellcolor{g}$ 99.25 $&\cellcolor{g}$ 98.96 $&\cellcolor{g}$ 98.84 $ \\
    	        \hline 	        
    	        \end{tabularx}
\caption{ALM Backtest for the $GlueV@R$ with $h_1 = 2/5$, $h_2 = 2/3$, $\alpha = 0.05$, $\beta = 0.01$. he level of the tests is set to $\kappa = 0.025$. The estimated size is represented as fraction to $0.05$. Values of the size between $0.9 - 1.1$ are dark blue, between $0.8 - 1.2$ are light blue. For the power green refers to a power $\geq 70\%$; light red to a power $\leq 30\%$
and dark red indicates a power $\leq 10\%$.}\label{ALMapp_size_2}
\end{table}

\end{document}